\documentclass[11pt,reqno]{amsart}
\usepackage{amssymb,amsmath,amsmath,amsthm,amsfonts}
\usepackage{graphicx,psfrag,epsf}
\usepackage{enumerate}
\usepackage{natbib}
\usepackage{url} 
\usepackage{amsfonts}
\usepackage{bm}
\usepackage{bbm} 
\usepackage{dsfont}
\usepackage{multirow}
\usepackage[table]{xcolor}
\usepackage{graphicx}
\usepackage{xr}
\usepackage{xcolor}

\usepackage{fullpage}
\numberwithin{equation}{section}
\newtheorem{theorem}{Theorem}
\newtheorem{corollary}{Corollary}
\newtheorem{definition}{Definition}
\newtheorem{proposition}{Proposition}
\newtheorem{lemma}{Lemma}

\theoremstyle{definition}

\newtheorem{example}{Example}
\newtheorem{remark}{Remark}

\DeclareMathOperator{\Var}{Var}

\DeclareMathOperator{\diag}{diag}
\DeclareMathOperator*{\argmin}{arg\,min}

\DeclareMathOperator{\supp}{supp}

\newcommand{\abs}[1]{\left\vert #1 \right\vert}

\newcommand{\calE}{{\mathcal{E}}}
\newcommand{\calF}{\mathcal{F}}

\newcommand{\calJ}{\mathcal{J}}
\newcommand{\calI}{\mathcal{I}}

\newcommand{\calR}{\mathcal{R}}
\newcommand{\calS}{\mathcal{S}}

\newcommand{\calU}{\mathcal{U}}

\newcommand{\calW}{\mathcal{W}}

\newcommand{\E}{\mathbb{E}}
\newcommand{\R}{\mathbb{R}}

\newcommand{\N}{\mathbb{N}}
\newcommand{\Prob}{\mathbb{P}}

\newcommand{\sumi}{\sum_{i=1}^n }
\newcommand{\sumj}{\sum_{j=1}^n }
\newcommand{\sumk}{\sum_{k=1}^n }

\renewcommand{\i}{{\bm{i}}}
\newcommand{\e}{{\bm{e}}}
\newcommand{\bN}{{\bm{N}}}
\newcommand{\bX}{{\bm{X}}}
\newcommand{\bW}{{\bm{W}}}
\newcommand{\bS}{{\bm{S}}}
\newcommand{\bV}{{\bm{V}}}
\newcommand{\1}{\mathbbm{1}}

\newcommand{\red}[1]{{\color{black}#1}}

\begin{document}

\title{Inference for high-dimensional exchangeable arrays}
\thanks{H.D. Chiang is supported by the Office of the Vice Chancellor for Research and Graduate Education at the University of Wisconsin–Madison with funding from the Wisconsin Alumni Research Foundation. K. Kato is partially supported by NSF grants DMS-1952306 and DMS-2014636.}
\author[H. D. Chiang]{Harold D. Chiang}
\author[K. Kato]{Kengo Kato}
\author[Y. Sasaki]{Yuya Sasaki}

\date{First arXiv version: September 10, 2020. This version:
\today}

\address[H. D. Chiang]{
Department of Economics, University of Wisconsin-Madison\\
 William H. Sewell Social Science Building, 1180 Observatory Drive,
Madison, WI 53706, USA.}
\email{hdchiang@wisc.edu}

\address[K. Kato]{
Department of Statistics and Data Science, Cornell University \\
 1194 Comstock Hall, Ithaca, NY 14853, USA.}
\email{kk976@cornell.edu}

\address[Y. Sasaki]{Department of Economics, Vanderbilt University\\
 VU Station B \#351819, 2301 Vanderbilt Place, Nashville, TN 37235-1819, USA.}
\email{yuya.sasaki@vanderbilt.edu}

\begin{abstract}
We consider inference for high-dimensional separately and jointly exchangeable arrays where the dimensions may be much larger than the sample sizes. 
For both exchangeable arrays, we first derive high-dimensional central limit theorems over the rectangles and subsequently develop
novel multiplier bootstraps with theoretical guarantees. 
These theoretical results rely on new technical tools such as Hoeffding-type decomposition and maximal inequalities for the degenerate components in the Hoeffiding-type decomposition for the exchangeable arrays. 
We exhibit applications of our methods to uniform confidence bands for density estimation under joint exchangeability and penalty choice for $\ell_1$-penalized regression under separate exchangeability. 
Extensive simulations demonstrate precise uniform coverage rates. 
We illustrate by constructing uniform confidence bands for international trade network densities. 
\end{abstract}


\maketitle

%%%%%%%%%%%%%%%%%%%%%%%%%%%%%%%%%%%%%%%%%%%%%%%%%%%%%%%%%%%%%%%%%%%%%
\section{Introduction}
%%%%%%%%%%%%%%%%%%%%%%%%%%%%%%%%%%%%%%%%%%%%%%%%%%%%%%%%%%%%%%%%%%%%%
%\textcolor{red}{Write new intro}

%\textcolor{red}{KK note: I propose the following structure. Th first two paragraphs are devoted to background, stating why inference for exchangeable arrays is important and relevant for statistics audience. Then we can discuss our results, spending maybe three or four paragraphs. All these will amount to 2+ pages. }

Many recent statistical problems involve non-independent observations indexed by multiple interlocking sets of entities. 
Examples include dyadic/polyadic networks, bipartite networks, and multiway clustering. %These data structures are well represented by various multivariate notions of exchangeability.
When the sets of entities that form each of these indices are different, as is the case with market-product data and book-reader data,  %and data arising from search engines and information retrieval, 
a natural stochastic framework is \textit{separate exchangeability} \citep{MacKinnonNielsenWebb2019}.
Separately exchangeable arrays include row-column exchangeable models \citep{Mccullagh2000}, additive cross random effect models \citep{Owen2007,OwenEckles2012}, and multiway clustering \citep{CGM2011}.
Meanwhile, when all indices belong to a common set of entities, as is the case with friendship network data, %or motif counting in biology, %and subgraph counts in random graph theory, 
the underlying structure is well-captured by \textit{joint exchangeability} \citep{BickelChen2009}.
Joint exchangeability covers nonparametric random graph models of \cite{BickelChen2009} for dyadic networks, which contain widely used  models in the statistical network analysis literature such as stochastic block models.

%
%Separately exchangeable random arrays encompass row-column exchangeable model studied in \cite{Mccullagh2000}, the additive cross random effect models investigated by \cite{Owen2007} and \cite{OwenEckles2012}, and the multiway nonnested clustering of \cite{CGM2011}. %Data with such structures arise for search engines and information retrieval problems. 
%Jointly exchangeable arrays, on the other hand, has been widely used in statistical network analysis literature. Importantly, it covers \citeauthor{BickelChen2009}’s (\citeyear{BickelChen2009}) nonparametric random graph model for dyadic networks. %where node labels carry no information, which is widely utilized in the literature of statistical network analysis. %This formulation is subsequently utilized by \cite{BickelChenLevina2011}, \cite{bhattacharyya2015subsampling}, \cite{zhang2017estimating}, \cite{crane2018edge}, and more. %The first paper provides asymptotic theory for empirical distribution of count features or “empirical moments.” Such count statistics are widely utilized in different scientific fields, such as in biology under the name motif counts and in probability as subgraph counts. 
%, which has since become a popular model for such data structures.
%see \cite{Graham2019} and \cite{Graham_de_Paula2019} for recent reviews as well as the issue edited by \citet{abbring2017special}.

Analysis of these types of data requires accounting for the underlying complex dependence structures induced by these exchangeability notions. Thus, developing valid inference methods for exchangeable arrays is challenging. 
The literature has witnessed some research on statistical inference that focuses on exchangeable arrays with low or fixed dimensions. 
For modern statistical learning methods, it is crucial to allow the dimension of data to increase with sample size.
However, the existing literature has been silent about statistical inference for such high-dimensional exchangeable arrays.

This paper is concerned with the problem of inference for separately or jointly exchangeable high-dimensional arrays.
We develop new high-dimensional central limit theorems (CLTs) over the rectangles for the sample mean under both exchangeability notions.
Building on the high-dimensional CLTs, we propose new multiplier bootstrap methods tailored to separate and jointly exchangeable arrays and derive their nonasymptotic error bounds. Such nonasymptotic results can be translated into asymptotic results that hold uniformly over a large set of distributions, which is crucial in many high-dimensional statistical applications.
%Non-asymptotic bounds are derived for all these theoretical results and thus they can be used to derive asymptotic theory that are valid over a large set of probability measures, which is crucial in many high-dimensional statistical applications.  

To derive these theoretical results, we develop several new technical tools, which are of independent interest and would be useful for other analyses of exchangeable arrays. Specifically, we develop novel Hoeffding-type decompositions for both separately and jointly exchangeable arrays and establish novel maximal inequalities for Hoeffding-type projections in both cases.  Such maximal inequalities lead to sharp rates for degenerate components in Hoeffding-type decompositions in both cases and play a crucial role in establishing the high-dimensional CLTs and the validity of the bootstrap methods.
The proofs of these technical results are highly nontrivial. For example, the proof of the symmetrization inequality for exchangeable arrays involves a careful induction argument (see Lemma \ref{lem: conditional zero mean} in the Appendix) combined with a repeated conditioning argument. \red{Furthermore, the proof of the maximal inequality for jointly exchangeable arrays involves a delicate conditioning argument combined with the decoupling inequalities for $U$-statistics with index-dependent kernels \citep[cf.][]{delaPenaGine1999}.}

We illustrate applications of the bootstrap methods to a couple of concrete statistical problems. 
Specifically, 1) we develop a method to construct simultaneous or uniform confidence bands for density functions with jointly exchangeable dyadic arrays, and 2) we develop a method to choose a penalty level for $\ell_1$-penalized regression (Lasso) and establish error bounds for the Lasso with separately exchangeable arrays. These applications are also new in the literature. 

We conduct extensive simulation studies, which demonstrate precise uniform coverage across various designs and under both  notions of exchangeability, thereby supporting our theoretical results. 
Finally, we apply our bootstrap  method to international trade network data to draw uniform confidence bands for  trade flow volumes in 1990, 1995, 2000, and 2005. 
The results indicate that there have been increasing numbers of bilateral trading pairs with high flow volumes as time progresses.

%%%%%%%%%%%%%%%%%%%%%%%%%%%%%%%%%%%%%%%%%%%%%%%%%%%%%%%%%%%%%%%%%%%%%
\subsection{Relation to the literature}
%%%%%%%%%%%%%%%%%%%%%%%%%%%%%%%%%%%%%%%%%%%%%%%%%%%%%%%%%%%%%%%%%%%%%
There is now a large literature on high-dimensional CLTs and bootstraps
with the  ``$p \gg n$'' regime; see  \cite{CCK2013AoS,CCK2014AoS,CCK2015PTRF,CCK2016SPA,CCK2017AoP}, \cite{DengZhang2017}, \cite{CCKK2019}, \cite{kuchibotla2020}, and \cite{fang2020} for the independent case, \cite{Chen2018}, \cite{ChenKato2019b,ChenKato2019} for $U$-statistics and processes, and \cite{ZhangWu2017}, \cite{ZhangCheng2018}, \cite{CCK2019RES}, \cite{Koike2019} for time series dependence.
However, none of the above references covers extensions to exchangeable arrays. 
The present paper builds on and contributes to this literature  by developing high-dimensional CLTs and bootstrap methods for exchangeable arrays.

Early applications of exchangeable arrays in statistics include \cite{arnold1979linear}, \cite{bowman1995saturated}, and \cite{Andrews2005}, to name a few.  For reviews, see, e.g. \cite{goldenberg2010survey}, \cite{orbanz2014bayesian}, and \cite{kuchibhotla2020exchangeability}.
Analysis of exchangeable random graphs has been an active research area in the recent statistics literature; see, e.g.,  \cite{DiaconisJanson2008}, \cite{BickelChenLevina2011}, \cite{lloyd2012random}, \cite{choi2014co}, \cite{caron2017sparse}, \cite{choi2017co}, \cite{zhang2017estimating}, \cite{crane2018edge}.  Limit theorems for 
jointly exchangeable arrays (in the fixed dimensional case) date back to \cite{Silverman1976} and \cite{EaglesonWeber1978}. \cite{FafchampsGubert2007} and \cite{CGM2011} derive standard error formulas for jointly exchangeable dyadic arrays and separately exchangeable arrays, respectively; see also  \cite{CameronMiller2014, CM2015}, \cite{AronowSamiiAssenova2015}, and \cite{Tabord-Meehan2019} for further development.
\cite{Menzel2017} studies inference for separately exchangeable arrays, covering both degenerate and non-degenerate cases. 
 \cite{DDG2019} develop functional limit theorems for Donsker classes under separate and joint exchangeability.
To the best of our knowledge, however, no existing work  in this literature permits  high-dimensional inference.
We note that \cite{DDG2019} develop symmetrization
inequalities different from ours. Specifically, symmetrization inequalities developed in \cite{DDG2019} are applied
to the whole empirical process and do not lead to correct orders for degenerate components
in Hoeffding-type decompositions (indeed, \cite{DDG2019} do not derive Hoeffding-type decompositions),
thereby not powerful enough to derive our results; see Remarks \ref{rem: comparison with DDG} and \ref{rem: comparison with DDG 2} in the Appendix for details.

%Exchangeable random variables have been playing important roles in statistics. 

%In the context of subgraph count statistics, it is studied by %\cite{BickelChen2009},  
%\cite{BickelChenLevina2011}, \cite{lloyd2012random}, %\cite{bhattacharyya2015subsampling}, 
%\cite{zhang2017estimating}, \cite{crane2018edge}, and more.   %\cite{DDG2019} develop functional limit theorems for Donsker classes under polyadic sampling. 
%To the best of our knowledge, no existing theory in this literature permits increasing or high-dimensional inference.

Methodologically, this paper is also related to the recent literature on high-dimensional $U$-statistics, such as \cite{Chen2018}, \cite{ChenKato2019,ChenKato2019b}, among others. 
Under suitable assumptions, the data of our interest can be written as $U$-statistic-like latent structure (in distribution) via the Aldous-Hoover-Kallenberg representation \citep{Aldous1981,Hoover1979,Kallenberg2006}, i.e. the data can be written as a kernel function of some latent independent random variables. 
However, unlike in $U$-statistics, neither the kernel nor the latent independent random variables is known to us.
In addition, we need to cope with the existence of extra %idiosyncratic 
higher-order shocks in the latent structure. 
Both  aspects present extra challenges.

Regarding our bootstraps, \cite{Mccullagh2000} shows that no resampling scheme for the raw data is consistent for variance of a sample mean under separate exchangeability.
A Pigeonhole bootstrap is subsequently proposed by \cite{Owen2007} and its different variants are further investigated in \cite{OwenEckles2012}, \cite{Menzel2017} and \cite{DDG2019}.
Whether the pigeonhole bootstrap works for increasing or high-dimensional test statistics remains unknown to us.
We therefore develop a novel bootstrap method in this paper which we argue works for high-dimensional data.

\subsection{Notations and organization}
%%%%%%%%%%%%%%%%%%%%%%%%%%%%%%%%%%%%%%%%%%%%%%%%%%%%%%%%%%%%%%%%%%%%%
Let $\N$ denote the set of positive integers. We use $\left\| \cdot \right\|, \left\| \cdot \right\|_{0}, \left\| \cdot \right\|_{1}$, and $\left\| \cdot \right\|_{\infty}$ to denote the Euclidean, $\ell_{0}$,  $\ell_1$, and $\ell^{\infty}$-norms for vectors, respectively (precisely, $\left\| \cdot \right\|_{0}$ is not a norm but a seminorm). For two real vectors $\bm{a}= (a_{1},\dots,a_{p})^{T}$ and $\bm{b} = (b_{1},\dots,b_{p})^{T}$, the notation $\bm{a} \le \bm{b}$ means that $a_{j} \le b_{j}$ for all $1 \le j \le p$. 
Let $\supp (\bm{a})$ denote the support of $\bm{a} = (a,\dots,a_p)^{T}$, i.e., $\supp(\bm{a}) = \{ j : a_j \ne 0\}$. We denote by $\odot$ the Hadamard (element-wise) product, i.e., for $\i = (i_1,\dots,i_K)$ and $\bm{j} = (j_1,\dots,j_K)$, $\i \odot \bm{j}= (i_1 j_1,\dots,i_K j_K)$. For any $a,b \in \R$, let $a \vee b = \max \{ a,b \}$. 
For $0 < \beta < \infty$, let $\psi_{\beta}$ be the function on $[0,\infty)$ defined by $\psi_{\beta} (x) = e^{x^{\beta}}-1$.
Let $\| \cdot \|_{\psi_{\beta}}$ denote the associated Orlicz norm, i.e., 
$\| \xi \|_{\psi_\beta}=\inf \{ C>0: \E[ \psi_{\beta}( | \xi | /C)] \leq 1\}$ for a real-valued random variable $\xi$.
%For $\beta \in (0,1)$, $\| \cdot \|_{\psi_{\beta}}$ is not a norm but a quasi-norm, i.e., there exists a constant $C_{\beta}$ depending only on $\beta$ such that $\| \xi_{1} + \xi_{2} \|_{\psi_{\beta}} \leq C_{\beta} ( \| \xi_{1} \|_{\psi_{\beta}} + \| \xi_{2} \|_{\psi_{\beta}})$. Let $U[0,1]$ denote the uniform distribution on $[0,1]$. 
``Constants'' refer to nonstochastic and finite positive numbers. 
%%%%%%%%%%%%%%%%%%%%%%%%%%%%%%%%%%%%%%%%%%%%%%%%%%%%%%%%%%%%%%%%%%%%%

%%%%%%%%%%%%%%%%%%%%%%%%%%%%%%%%%%%%%%%%%%%%%%%%%%%%%%%%%%%%%%%%%%%%%
The rest of the paper is organized as follows . In  Section \ref{sec:multiway_clustered_data}, we develop a high-dimensionl CLT (over the rectangles) and a bootstrap method for separately exchangeable arrays. 
In Section \ref{sec:polyadic_data}, we develop analogous results to jointly exchangeable arrays. 
We illustrate \red{two applications} in Section \ref{sec:applications}, present simulation results in Section \ref{sec:simulation}, and demonstrate an empirical application in Section \ref{sec:empirical}.
We defer all the technical proofs to the Appendix.

%%%%%%%%%%%%%%%%%%%%%%%%%%%%%%%%%%%%%%%%%%%%%%%%%%%%%%%%%%%%%%%%%%%%%
\section{Separately exchangeable arrays}\label{sec:multiway_clustered_data}
%%%%%%%%%%%%%%%%%%%%%%%%%%%%%%%%%%%%%%%%%%%%%%%%%%%%%%%%%%%%%%%%%%%%%
In this section, we consider separately exchangeable arrays that correspond to multiway clustered data. 
Pick any $K \in \N$. 
With $\i = (i_{1},\dots, i_{K}) \in \N^{K}$, we consider a $K$-array $( \bX_{\i})_{\i \in \N^{K}}$ consisting of random vectors in $\R^{p}$ \red{with $p \ge 2$}. We denote by $X_{\i}^{j}$ the $j$-th coordinate of $\bX_{\i}$: $\bX_{\i} = (X_{\i}^{1},\dots,X_{\i}^{p})^{T}$. 
We say that the array $(\bX_{\i})_{\i \in \N^K}$ is \textit{separately exchangeable} if the following condition is satisfied \citep[cf.][Section 3.1]{Kallenberg2006}.

\begin{definition}[Separate exchangeability]
A $K$-array $(\bX_{\i})_{\i \in \N^{K}}$ is called separately exchangeable if 
for any $K$ permutations $\pi_1,\dots,\pi_K$ of $\N$,  the arrays $(\bX_{\i})_{\i \in \N^{K}}$ and $(\bX_{(\pi_1(i_1),\dots,\pi_K(i_K))})_{\i \in \N^{K}}$ are identically distributed in the sense that their finite dimensional distributions agree. 
\end{definition}

See Appendix \ref{sec:more_exchangeable_arrays} in the supplementary material for more details, discussions, and examples.
From the Aldous-Hoover-Kallenberg representation \citep[see][Corollary 7.23]{Kallenberg2006}, any separately exchangeable array $(\bX_{\i})_{\i \in \N^{K}}$ is generated by the structure
\[
\bX_{\i} = \mathfrak{f} ((U_{\i \odot \e})_{\e \in \{ 0,1 \}^{K}}), \ \i \in \N^{K}, \quad \{ U_{\i \odot \e} : \i \in \N^{K}, \e \in \{ 0,1 \}^{K} \} \stackrel{i.i.d.}{\sim} U[0,1]
\]
for some Borel measurable map $\mathfrak{f}:[0,1]^{2^{K}} \to \R^{p}$.
%where $\{ U_{\i \odot \e} : \i \in \N^{K}, \e \in \{ 0,1 \}^{K} \}$ are independent uniform random variables on $[0,1]$.
%\footnote{Under such construction, $\mathfrak{g}$ is random and thus all the analysis are conditional on the realization of $\mathfrak{g}$. Alternatively, suppose we assume in addition that $(\bX_{\i})_{\i \in \mathcal{I}_1^K}$ and $(\bX_{\i})_{\i \in \mathcal{I}_2^K}$ are independent if $\mathcal{I}_1 \subset \N^K$ and $\mathcal{I}_2 \subset \N^K$ are disjoint. Then under this assumption and (SE), a different version of Aldous-Hoover-Kallenberg representation \citep[see][Lemma 7.35]{Kallenberg2006} yields a deterministic representation of $\mathfrak{g}$.} 
%For example, when $K=2$, then $\bX_{\i}$ is generated as $\bX_{(i_1,i_2)} = \mathfrak{f} (U_{(0,0)}, U_{(i_1,0)},U_{(0,i_2)},U_{(i_1,i_2)})$. 
%We emphasize that, despite the uniform distribution of $U_{\i \odot \e}$, this representation accommodates any dimensionality $p$ for random vector $\bX_{\i}$ -- see \citet[][Lemma 2.1]{Aldous1981} for instance.

The latent variable $U_{\bm{0}}$ appears commonly in all $\bX_{\i}$'s. In the present paper, as in \cite{Andrews2005} and \cite{Menzel2017}, we consider inference conditional on $U_{\bm{0}}$ and treat it as fixed. In the rest of Section \ref{sec:multiway_clustered_data}, we will assume (without further mentioning) that the array $(\bX_{\i})_{\i \in \N^K}$ has mean zero (conditional on $U_{\bm{0}}$) and is generated by the structure 
\begin{equation}
\bX_{\i} = \mathfrak{g} ((U_{\i \odot \e})_{\e \in \{ 0,1 \}^{K}\setminus \{ \bm{0} \}}), \ \i \in \N^{K},
\label{eq: Aldous-Hoover_SE}
\end{equation}
where $\mathfrak{g}$ is now a map from $[0,1]^{2^{K}-1}$ into $\R^{p}$.

Suppose that we observe $\{ \bX_{\i} : \i \in [\bm N] \}$ with $\bN = (N_1,\dots,N_K)$ and $[\bm N] = \prod_{k=1}^{K} \{ 1,\dots,N_{k} \}$. 
We are interested in approximating the distribution of the sample mean 
\[
\bS_{\bN} =\frac{1}{\prod_{k=1}^{K}N_{k}}\sum_{\i \in [\bN]} \bX_{\i}
\]
in the high-dimensional setting where the dimension $p$ is allowed to entail $p \gg \min \{N_1,\dots,N_K\}$. 

\begin{example}[Empirical process indexed by function class with increasing cardinality]
\label{ex: discretized EP}
Our setting covers the following situation: let $\{ Y_{\i} : \i \in \N^{K} \}$ be random variables taking values in an abstract measurable space $(S,\calS)$, and  suppose that they are generated as $Y_{\i} = \check{\mathfrak{g}} ((U_{\i \odot \e})_{\e \in \{ 0,1 \}^{K} \setminus \{ \bm{0} \}})$. 
Let $f_{j}: S \to \R$ for $1 \le j \le p$ be measurable functions, and define $X_{\i}^{j} = f_{j}(Y_{\i}) - \E[f_{j}(Y_{\i})]$. In this case, the sample mean $\bS_{\bN}$ can be regarded as the empirical process $f \mapsto (\prod_{k=1}^{K} N_{k})^{-1} \sum_{\i \in [\bN]} (f(Y_{\i}) - \E[f(Y_{\i})])$ indexed by the function class $\calF = \{ f_{1},\dots,f_{p} \}$. Allowing $p \to \infty$ as $\min_{1 \le k \le K} N_{k} \to \infty$ enables us to cover empirical processes indexed by function classes with increasing cardinality. 
\end{example}

For later convenience, we fix some additional notations.
Let $n = \min_{1 \le k \le K} N_{k}$ and $\overline{N} = \max_{1 \le k \le K}N_{k}$ denote the minimum and maximum cluster sizes, respectively.
For $1 \le k \le K$, denote by $\calE_{k} = \{ \e= (e_1,\dots,e_K)  \in \{ 0,1 \}^{K}: \sum_{k=1}^{K} e_k = k \}$ the set of vectors in $\{ 0,1 \}^{K}$ whose support has cardinality $k$. Let $\e_{k} \in \R^{K}$ denote the vector such that the $k$-th coordinate of $\e_{k}$ is $1$ and the other coordinates are $0$. 
For a given $\e \in \{ 0,1 \}^{K}$, define 
\[
I_{\e} ([\bN]) = \{ \i \odot \e : \i \in [\bN] \} \subset \N_{0}^{K} \quad \text{with} \ \N_0 = \N \cup \{ 0 \}. 
\]
%For instance, if $K=3$ and $\e = (1,1,0)$, then 
%\[
%I_{\e} ([\bN]) = \{ (i_{1},i_{2},0) : i_{1} = 1,\dots,N_{1},i_{2} = 1,\dots,N_{2} \}.
%\] 
The following decomposition of the sample mean $\bS_{\bN}$ will play a fundamental role in our analysis, which is reminiscent of the Hoeffding decomposition for $U$-statistics \citep{Lee1990, delaPenaGine1999}.

\begin{lemma}[Hoeffding decomposition of separately exchangeable array]
\label{lem: H decomposition}
For any $\i \in \N^{K}$, define recursively $\hat{\bX}_{\i \odot \e_{k}} = \E[ \bX_{\i} \mid U_{\i\odot \e_{k}} ]$ for $ k=1,\dots,K$ and 
$\hat{\bX}_{\i\odot \e} = \E[\bX_{\i} \mid (U_{\i\odot \e'})_{ \e' \le \e}] - \sum_{\substack{\e' \le \e \\ \e' \ne \e}} \hat{\bX}_{\i \odot \e'}$ for $\e \in \bigcup_{k=2}^{K} \mathcal{E}_{k}$. 
Then, we have $\bX_{\i} = \sum_{\e \in \{ 0,1 \}^{K} \setminus \{ \bm {0} \}} \hat{\bX}_{\i \odot \e}$. 
Consequently,  we can decompose the sample mean $\bS_{\bN} =(\prod_{k=1}^{K}N_{k})^{-1}\sum_{\i \in [\bN]} \bX_{\i}$ as 
\begin{equation}
\label{eq: H decomposition}
\bS_{\bN} = \sum_{k=1}^{K} \sum_{\e \in \mathcal{E}_{k}} \frac{1}{\prod_{k' \in \supp (\e)} N_{k'}} \sum_{\i \in I_{\e}([\bN])} \hat{\bX}_{\i}.
\end{equation}
\end{lemma}
The proof of this lemma can be found in Appendix \ref{sec:proof of Lemma Hoeffding multiway}.

%\begin{example}[$K=3$ case]
%For instance, if $K=3$, then for $\i = (i_{1},i_{2},i_{3}) \in \N^{3}$, 
%\[
%\begin{split}
%&\hat{\bX}_{(i_{1},0,0)} = \E[\bX_{\i} \mid U_{(i_{1},0,0)}], \ \hat{\bX}_{(0,i_{2},0)} = \E[\bX_{\i} \mid U_{(0,i_{2},0)}], \ \hat{\bX}_{(0,0,i_{3})} = \E[\bX_{\i} \mid U_{(0,0,i_{3})}], \\
%&\hat{\bX}_{(i_{1},i_{2},0)} = \E[\bX_{\i} \mid U_{(i_{1},0,0)},U_{(0,i_{2},0)},U_{(i_{1},i_{2},0)}] - \hat{\bX}_{(i_{1},0,0)} - \hat{\bX}_{(0,i_{2},0)}, \ \text{etc.}, \\
%&\hat{\bX}_{(i_{1},i_{2},i_{3})} = \bX_{\i} - \hat{\bX}_{(i_{1},i_{2},0)} - \hat{\bX}_{(0,i_{2},i_{3})} - \hat{\bX}_{(i_{1},0,i_{3})} - \hat{\bX}_{(i_{1},0,0)} - \hat{\bX}_{(0,i_{2},0)} - \hat{\bX}_{(0,0,i_{3})}.
%\end{split}
%\]
%\end{example}

\begin{remark}[Hoeffding decomposition]
The reason that we call (\ref{eq: H decomposition}) the Hoeffding decomposition comes from the fact that if the dimension $p$ is fixed, for each fixed $k=1,\dots,K$ and $\e \in \calE_{k}$, the component $(\prod_{k' \in \supp (\e)} N_{k'})^{-1} \sum_{\i \in I_{\e}([\bN])} \hat{\bX}_{\i}$
scales as $(\prod_{k' \in \supp(\e)} N_{k'})^{-1/2} = O(n^{-k/2})$ with $n = \min_{1 \le k' \le K} N_{k'}$ under moment conditions.  See Corollary \ref{cor: maximal inequality} in Appendix \ref{sec:maximal_inequality_for_multiway_data}. This is completely analogous to the Hoeffiding decomposition of $U$-statistics and from this analogy we shall call (\ref{eq: H decomposition}) the Hoeffding decomposition. 
\end{remark}

The leading term in the decomposition (\ref{eq: H decomposition}) is 
\[
\sum_{\e \in \mathcal{E}_{1}} \frac{1}{\prod_{k' \in \supp (\e)} N_{k'}} \sum_{\i \in I_{\e}([\bN])} \hat{\bX}_{\i} =\sum_{k=1}^{K}N_{k}^{-1} \sum_{i_{k}=1}^{N_{k}} \E[ \bX_{\i} \mid U_{(0,\dots,0,i_{k},0,\dots,0)} ],
\]
which we call the H\'{a}jek projection of $\bS_{\bN}$. With this in mind, 
define $\bW_{k,i_{k}} = \E[ \bX_{\i} \mid U_{(0,\dots,0,i_{k},0,\dots,0)} ]$ for $k=1,\dots,K$. 
%\[
%%\bW_{k,i_{k}} &= \E[ \bX_{\i} \mid U_{(0,\dots,0,i_{k},0,\dots,0)} ], \ k=1,\dots,K, \\
%\bS_{\bN}^{W} =\sum_{k=1}^{K} N_{k}^{-1}\sum_{i_{k}=1}^{N_{k}} \bW_{k,i_{k}}, \ \text{and} \ \ \Sigma_{W_{k}} = \E[\bW_{k,1} \bW_{k,1}^{T}], \ k=1,\dots,K.
%\]
%Since $\bS_{\bN}^{W}$ is the sum of independent random vectors, it is expected that the distribution of $\sqrt{n}\bS_{\bN}$ can be approximated by $N(\bm{0},\Sigma)$, where $\Sigma = \sum_{k=1}^{K} (n/N_{k}) \Sigma_{W_{k}}$, 
%as long as the remainder term is negligible.

% This suggests the following multiplier bootstrap for separately exchangeable arrays.

%%%%%%%%%%%%%%%%%%%%%%%%%%%%%%%%%%%%%%%%%%%%%%%%%%%%%%%%%%%%%%%%%%%%%
%\subsection{Multiplier bootstrap for separately exchangeable arrays}
%%%%%%%%%%%%%%%%%%%%%%%%%%%%%%%%%%%%%%%%%%%%%%%%%%%%%%%%%%%%%%%%%%%%%

%%%%%%%%%%%%%%%%%%%%%%%%%%%%%%%%%%%%%%%%%%%%%%%%%%%%%%%%%%%%%%%%%%%%%
\subsection{High-dimensional CLT for separately exchangeable arrays}
%%%%%%%%%%%%%%%%%%%%%%%%%%%%%%%%%%%%%%%%%%%%%%%%%%%%%%%%%%%%%%%%%%%%%
We first establish a high-dimensional CLT  for $\bS_{\bN}$ over the class of rectangles, $\calR= \{ \prod_{j=1}^{p} [a_{j},b_{j}] : - \infty \le a_{j} \le b_{j} \le \infty, \  1 \le j \le p \}$.
This high-dimensional CLT will be a building block for establishing the validity of the multiplier bootstrap (cf. Section \ref{sec: MB sep}). 

We start with discussing regularity conditions. Denote by $\bm{1} = (1,\dots,1)$ the vector of ones. Let $D_{\bN} \ge 1$ be a given constant that may depend on the cluster sizes $\bN$ \red{(and $p$; when we consider asymptotics we have in mind that $p$ is a function of $\bN$ or $n$ so we omit the dependence of $D_{\bN}$ on $p$)}, and let $\underline{\sigma} > 0$ be another given constant independent of the cluster sizes $\bN$. 
We will assume \textit{either} of the following moment conditions.
\begin{align}
&\max_{1 \le j \le p} \| X_{\bm{1}}^{j} \|_{\psi_{1}} \le D_{\bN},  \quad \text{or} \label{eq:condition1}\\
&\E[\| \bX_{\bm{1}} \|_{\infty}^{q}] \le D_{\bN}^{q} \quad \text{for some $q \in (4,\infty)$}. \label{eq:condition1_poly}
\end{align}
We will also assume the following condition.
\begin{align}
\max_{1 \le j \le p; 1 \le k \le K} \E[| W_{k,1}^{j} |^{2+\kappa}] \le D_{\bN}^{\kappa}, \ \kappa=1,2, \quad \text{and} \quad \min_{1 \le j \le p; 1 \le k \le K} \E[ |W_{k,1}^{j} |^{2}] \ge \underline{\sigma}^{2}. \label{eq:condition2} 
%&.\label{eq:condition3}
\end{align}

Condition (\ref{eq:condition1}) requires that each coordinate of  $\bX_{\bm 1}$ is sub-exponential. By Jensen's inequality, Condition (\ref{eq:condition1}) implies that $\max_{1 \le j \le p;1 \le k \le K}\| W_{k,1}^j \|_{\psi_1} \le D_{\bN}$.
Condition (\ref{eq:condition1_poly}) is an alternative moment condition on $\bX_{\bm{1}}$. 
Condition (\ref{eq:condition1_poly}) is satisfied for example under the following situation: Suppose that $\bX_{\i}$ is given by $\bX_{\i} = \varepsilon_{\i} \bm{Z}_{\i}$ where $\varepsilon_{\i}$ is a scalar ``error'' variable while $\bm{Z}$ is a vector of ``covariates''. 
If each coordinate of $\bm{Z}_{\i}$ is bounded by a constant $\overline{D}$ and $\varepsilon_{\i}$ has finite $q$-th moment, then $\E[\| \bX_{\i} \|_{\infty}^{q}] \le \overline{D}^{q} \E[|\varepsilon_{\i}|^{q}]$. \red{Also Condition (\ref{eq:condition1_poly}) is satisfied if, in the discretized empirical process application (cf. Example \ref{ex: discretized EP}), the function class possesses an envelope function with finite $q$-th moment.} Again, by Jensen's inequality, Condition (\ref{eq:condition1_poly}) implies that  $\max_{1 \le k \le K}\E[\| \bW_{k,1} \|_{\infty}^{q}] \le D_{\bN}^{q}$. \red{The restriction $q > 4$ is needed to guarantee that Condition (\ref{eq:bootstrap_rate_poly}) appearing in Theorem \ref{thm:bootsrap_validity} to be non-void.}

%as the distribution of $\bX_{\bm 1}$ depends on $N$ under high-dimensional settings. 
Condition (\ref{eq:condition2}) requires the maximum of third (respectively, fourth) moment across coordinates to be increasing at speed no faster than the first (respectively, second) power of $D_{\bN}$.  By Jensen's inequality, the first part of Condition (\ref{eq:condition2}) is satisfied if $\max_{1 \le j \le p} \E[|X_{\bm{1}}^{j}|^{2+\kappa}] \le D_{\bN}^{\kappa}$ for $\kappa=1,2$. The second part of Condition (\ref{eq:condition2}) guarantees that the H\'ajek projection is nondegenerate. 
%\end{remark}

Let $\gamma = N(\bm{0},\Sigma)$ with $\Sigma = \sum_{k=1}^K (n/N_k) \Sigma_{W_k}$  and $\Sigma_{W_k} = \E[\bW_{k,1}\bW_{k,1}^T]$ for $k=1,\dots,K$.

\begin{theorem}[High-dimensional CLT for separately exchangeable arrays]
\label{thm: high-d CLT}
Suppose that either  Condition (\ref{eq:condition1}) or (\ref{eq:condition1_poly}) holds, and further that  Condition (\ref{eq:condition2}) holds. Then, there exists a constant $C$ such that 
\[
\begin{split}
\sup_{R \in \mathcal{R}} | \Prob (\sqrt{n}\bS_{\bN} \in R) -\gamma_{\Sigma}(R) | 
 \le 
\begin{cases}
C \left ( \frac{D_{\bN}^{2} \log^{7} (p\overline{N})}{n} \right )^{1/6} & \text{if (\ref{eq:condition1}) holds,} \\
C \left [ \left ( \frac{D_{\bN}^{2} \log^{7} (p\overline{N})}{n} \right )^{1/6} + \left ( \frac{D_{\bN}^{2} \log^{3} (p\overline{N})}{n^{1-2/q}} \right )^{1/3} \right ] & \text{if  (\ref{eq:condition1_poly}) holds,}
\end{cases}
\end{split}
\]
where the constant $C$ depends only on $\underline{\sigma}$ and $K$ if Condition (\ref{eq:condition1}) holds, while $C$ depends only on $q,\underline{\sigma}$, and $K$ if Condition (\ref{eq:condition1_poly}) holds. 
\end{theorem}

\begin{remark}[Refinement under subgaussianity]
The recent paper of \cite{CCKK2019} provides some improvements on convergence rate of Gaussian approximation under the subgaussian tail assumption for the sample mean of independent random vectors. With this new technique, if we strengthen Condition (\ref{eq:condition1}) by replacing the $\psi_1$-norm $\| \cdot \|_{\psi_1}$ with the $\psi_2$-norm $\| \cdot \|_{\psi_2}$ (i.e., each coordinate $\bX_{\bm{1}}$ is sub-Gaussian),  the bound $C\left(n^{-1}D_\bN^2 \log^7(p\overline N)\right)^{1/6}$ in Theorem \ref{thm: high-d CLT} can be improved to $C\left(n^{-1}D_\bN^2 \log^5(p\overline N)\right)^{1/4}$. 
\end{remark}

\subsection{Multiplier bootstrap for separately exchangeable arrays}
\label{sec: MB sep}

Let $\{ \xi_{1,i_{1}} \}_{i_1=1}^{N_{1}}, \dots, \{ \xi_{K,i_{K}} \}_{i_K=1}^{N_K}$ be independent $N(0,1)$ random variables independent of the data. Ideally, we want to make use of the bootstrap statistic
$
\red{\sum_{k=1}^K} N_k^{-1}\sum_{i_k=1}^{N_k}\xi_{k,i_k}(\bW_{k,i_k}-\bS_{\bN}).
$
However, this bootstrap is infeasible as $\bW_{k,i_k} = \E[\bX_{\i} \mid U_{(0,\dots,i_{k},\dots,0)}]$ are unknown to us.
Estimation of $\bW_{k,i_k}$ is nontrivial as $U_{(0,\dots,i_{k},\dots,0)}$ is a latent variable. We propose to estimate each $\bW_{k,i_{k}}$ by 
\[
\overline{\bX}_{k,i_{k}} = \frac{1}{\prod_{k' \ne k} N_{k'}} \sum_{i_{1},\dots,i_{k-1},i_{k+1},\dots,i_{K}} \bX_{\i}, \ i_{k} = 1,\dots,N_{k}; k=1,\dots,K,
\]
i.e., the sample mean taken over all indices but $i_k$. 
Then, we apply the multiplier bootstrap to $\overline{\bX}_{k,i_k}$ in place of $\bW_{k,i_k}$
\[
\bS_{\bN}^{MB} =  \sum_{k=1}^{K} N_{k}^{-1} \sum_{i_{k}=1}^{N_{k}} \xi_{k,i_{k}} (\overline{\bX}_{k,i_k} - \bS_{\bN}). 
\]
To the best of our knowledge, this multiplier bootstrap for separately exchangeable arrays is new in the literature. 
We will formally study the validity of this multiplier bootstrap for high-dimensional separately exchangeable arrays with $p \gg n$.

% in the next two subsections. 
%
%To gain an insight into how to estimate $\bW_{k,i_k}$, consider the case where $K=2$. Then $\bW_{1,i_1} = \E[\bX_{(i_1,i_2)} \mid U_{(i_1,0)}] = \E[\mathfrak{g}(U_{(i_1,0)}, V_{(i_1,i_2)}) \mid U_{(i_1,0)}]$ with $ V_{(i_1,i_2)} = (U_{(0,i_2)},U_{(i_1,i_2)})$. Since $U_{(i_1,0)}$ and $V_{(i_1,i_2)}$ are independent and the latter variable is independent across $i_2$, we see that $\bW_{1,i_1}$ can be estimated by taking the average of $\bX_{(i_1,i_2)}$ over $i_2$. 

%
%Building on this intuition, in general, 

We are now in position to establish the validity of the proposed multiplier bootstrap for separately exchangeable arrays. 
Let $\Prob_{|\bX_{[\bN]}}$ denote the law conditional on the data $\bX_{[\bN]} = (\bX_{\i})_{\i\in[\bN]}$
%Define
%\[
%\hat \Delta_{W}=\max_{1\le j \le p;1 \le k \le K} \frac{1}{N_k} \sum_{i_k=1}^{N_k} ( \overline X_{k,i_k}^j 
%- W_{k,i_k}^j )^2,
%\]
%which accounts for the estimation error of $\overline{\bX}_{k,i_k}$ for $\bW_{k,i_k}$. 
and $\overline \sigma=\max_{1\le j \le p; 1 \le k \le K} \sqrt{\E[|W_{k,1}^j|^2]}$. 
%The following theorem shows that as soon as $\hat{\Delta}_{W}$ is sufficiently small (i.e, $\overline{\sigma}^2 \hat{\Delta}_{W} \log^4 p = o_{P}(1)$), then the multiplier bootstrap is consistent over the rectangles under mild conditions on the dimension $p$. 

%%%%%%%%%%%%%%%%%%%%%%%%%%%%%%%%%%%%%%%%%%
\begin{theorem}[Validity of multiplier bootstrap for separately exchangeable arrays]\label{thm:bootsrap_validity}
Consider the following two cases:
\begin{enumerate}
\item[(i).] \red{Assume that} Conditions (\ref{eq:condition1}) and (\ref{eq:condition2}) hold, and \red{further} there exist constants $C_1$ and $\zeta \in (0,1)$ such that 
\begin{align}
%&\Prob \left ( \overline \sigma^2 \hat \Delta_{W}\log^4 p >C_1 n^{-\zeta_2} \right)\le C_1 n^{-1} \quad \text{and}\label{eq:bootstrap_Delta_1} \\
\frac{\overline \sigma^2 D_{\bN}^2 \log^{7} p }{n} \bigvee \frac{D_{\bN}^2 (\log^2 n ) \log^5(p \overline N)}{n} \le C_1 n^{-\zeta}. \label{eq:bootstrap_rate}
\end{align}
\item[(ii).]  \red{Assume that} Conditions (\ref{eq:condition1_poly}) and (\ref{eq:condition2}) hold, and \red{further} there exist constants $C_1$ and $\zeta \in (2/q,1)$ such that 
\begin{equation}
\frac{\overline \sigma^2 D_{\bN}^2 \log^{5} (pn) }{n}  \bigvee \left ( \frac{D_{\bN}^2  \log^3 p}{n^{1-4/q}} \right)^2 \le C_1 n^{-\zeta}.  \label{eq:bootstrap_rate_poly}
\end{equation}
\end{enumerate}
Then, under  Case (i), for any $\nu \in (1/\zeta,\infty)$, there exists a constant $C$ depending only on $\nu, \underline{\sigma}, K$, and $C_1$ such that
$
\sup_{R\in \calR}\left|\Prob_{|\bX_{[\bN]}}(\sqrt{n} \bS_{\bN}^{MB}\in R)-\gamma_{\Sigma}(R) \right|\le C n^{-(\zeta-1/\nu)/4}
$
with probability at least $1-Cn^{-1}$. Under Case (ii), the same conclusion holds with $n^{-(\zeta-1/\nu)/4}$ replaced by $n^{-(\zeta-2/q)/4}$, while the constant $C$ depends only on $q, \underline{\sigma}, K$, and $C_1$.
\end{theorem}

\begin{remark}[Discussion on Conditions (\ref{eq:bootstrap_rate}) and (\ref{eq:bootstrap_rate_poly})]\label{remark:discussion_polynomial_rates}
Conditions (\ref{eq:bootstrap_rate}) and (\ref{eq:bootstrap_rate_poly}) are placed to guarantee that the error bound for our multiplier bootstrap decreases at a polynomial rate in $n$. 
If we are to show a weaker result, namely, 
\begin{equation}
\sup_{R\in \calR}|\Prob_{|\bX_{[\bN]}}(\sqrt{n} \bS_{\bN}^{MB}\in R)-\gamma_{\Sigma}(R) | = o_{P}(1)
\label{eq: consistency}
\end{equation}
 as $n \to \infty$ (with the understanding that $p, \overline{\sigma}, D_{\bN}$, and $\overline{N}$ are functions of $n$), then Conditions  (\ref{eq:bootstrap_rate}) and (\ref{eq:bootstrap_rate_poly}) can be weakened to $(\overline{\sigma}^2D_{\bN}^{2} \log^{7}p) \vee D_{\bN}^2\log^5 (p\overline{N}) = o(n)$ and $(n^{-1}\overline{\sigma}^2D_{\bN}^{2} \log^{5}(pn)) \vee (n^{-(1-2/q)}D_{\bN}^{2}\log^{3}p) = o(1)$, respectively. (The critical case $q=4$ is allowed for (\ref{eq: consistency}); note that the high-dimensional CLT (Theorem \ref{thm: high-d CLT}) also holds with $q=4$.) 
\end{remark}

\begin{remark}[Normalized sample mean]
In practice, we often normalize the coordinates of the sample mean by estimates of the standard deviations, so that each coordinate is approximately distributed as $N(0,1)$. 
We can estimate the variance of the $j$-th coordinate of $\sqrt{n}\bS_{\bN}$ by the conditional variance of the $j$-th coordinate of $\sqrt{n} \bS_{\bN}^{MB}$. The validity of the multiplier bootstrap to the normalized sample mean follows similarly to the preceding theorem; see Appendix \ref{sec: normalized sample mean SE} for details. 
A similar comment applies to the joint exchangeable case; see  Appendix \ref{sec: normalized sample mean JE} for details. 
\end{remark}

\section{Jointly exchangeable arrays}\label{sec:polyadic_data}
%%%%%%%%%%%%%%%%%%%%%%%%%%%%%%%%%%%%%%%%%%%%%%%%%%%%%%%%%%%%%%%%%%%%%
In this section, we consider another class of exchangeable arrays, namely, jointly exchangeable arrays.
The notations in the current section are independent from those in Section \ref{sec:multiway_clustered_data} unless otherwise noted. 
Joint exchangeability induces a more complex dependence structure on arrays than separate exchangeability, but still we are able to develop analogous results to the preceding section for jointly exchangeable arrays as well. It should be noted, however, that we do require a different bootstrap  and technical tools (cf. Appendix \ref{sec:maximal_inequality_for_polyadic_data}) to accommodate a specific dependence structure induced from joint exchangeability.

%For any set $A\subset \R^K$, denote by $\overline A=\{(i_1,\dots,i_K)\in A: i_1,\dots,i_K \ \text{are distinct} \}$ the collection of vectors in $A$ with distinct coordinates. Let $\mathbb I^K=\overline{\N^K}$, the set of $K$-tuples of positive integers without repetitions, and $I_{n,K}=\overline{\{1,\dots,n\}^K}$. 
Pick any $K \in \N$. For a given positive integer $n \ge K$, let $I_{n,K} = \{ (i_1,\dots,i_K) : 1 \le i_1, \dots, i_K \le n \ \text{and $i_1,\dots,i_K$ are distinct} \}$. Also let $I_{\infty,K} = \bigcup_{n=K}^{\infty} I_{n,K}$. 
For any $\i = (i_1,\dots,i_K) \in \N^K$, let $\{ \i \}^{+}$ denote the set of distinct \textit{nonzero} elements of $(i_1,\dots,i_K)$.
%For example, $\{ (2,0,1,2) \}^{+} = \{ 1,2 \}$. 

In this section, we consider a $K$-array $( \bX_{\i})_{\i \in I_{\infty,K}}$ consisting of random vectors in $\R^{p}$ \red{with $p \ge 2$}. We say that the array $(\bX_{\i})_{\i \in I_{\infty,K}}$ is \textit{jointly exchangeable} if the following condition is satisfied \citep[cf.][Section 3.1]{Kallenberg2006}.

\begin{definition}[Joint exchangeability]
A $K$-array $(\bX_{\i})_{\i \in I_{\infty,K}}$ is called jointly exchangeable if 
for any permutation $\pi$ of $\N$, the arrays $(\bX_{\i})_{\i \in I_{\infty,K}}$ and $(\bX_{(\pi(i_1),\dots,\pi(i_K))})_{\i\in I_{\infty,K}}$ are identically distributed.
\end{definition}

See Appendix \ref{sec:more_exchangeable_arrays} in the supplementary material for more details, discussions, and examples.
From the Aldous-Hoover-Kallenberg representation \citep[see][Theorem 7.22]{Kallenberg2006}, any jointly exchangeable array $(\bX_{\i})_{\i \in I_{\infty,K}}$ is generated by the structure
\[
\bX_{\i} = \mathfrak{f} ( (U_{\{\i \odot \e\}^{+}})_{\e \in \{ 0,1 \}^{K}}), \ \i \in I_{\infty,K}, \quad \{ U_{\{\i \odot \e\}^{+}} : \i \in I_{\infty,K}, \e \in \{ 0,1 \}^{K} \} \stackrel{i.i.d.}{\sim} U[0,1]
\]
for some Borel measurable map $\mathfrak{f}: [0,1]^{2^{K}} \to \R^{p}$. 
%\footnote{Under such construction, $\mathfrak{g}$ is random and thus all the analysis are conditional on the realization of $\mathfrak{g}$. Alternatively, suppose we assume that for any $A, B\subset \R$ with $A\cap B=\varnothing$, $(\bX_{\i})_{\i\in \overline{A^K}}$ and $(\bX_{\i})_{\i\in \overline{B^K}}$ are independent. Then under this assumption and (JE), a different version of Aldous-Hoover-Kallenberg representation \citep[see][Lemma 7.35]{Kallenberg2006} yields a deterministic representation of $\mathfrak{g}$.}
%For example, when $K=2$, then $\bX_{(i_1,i_2)}$ is generated as $\bX_{(i_1,i_2)} = \mathfrak{f} (U_{\varnothing}, U_{i_1}, U_{i_2}, U_{\{ i_1,i_2 \}})$. (We will write $U_{i_k} = U_{\{ i_k \}}$ for the notational convenience.) 
Here the coordinates of the vector $(U_{\{\i \odot \e\}^{+}})_{\e \in \{ 0,1 \}^{K}}$ are understood to be properly ordered, so that, e.g., when $K=2$, $\bX_{(i_1,i_2)} = \mathfrak{f} (U_{\varnothing},U_{ i_1}, U_{ i_2}, U_{\{ i_1,i_2 \}})$ and $\bX_{(i_2,i_1)} = \mathfrak{f} (U_{\varnothing}, U_{i_2}, U_{i_1}, U_{\{ i_1,i_2 \}})$ differ (although they have the identical distribution). 

As in the separately exchangeable case, we consider inference conditional on $U_{\varnothing}$, and in what follows, we will assume that the array $(\bX_{\i})_{\i \in I_{\infty,K}}$ has mean zero (conditional on $U_{\varnothing}$) and is generated by the structure  
\begin{equation}
\bX_{\i} = \mathfrak{g} ( (U_{\{\i \odot \e\}^{+}})_{\e \in \{ 0,1 \}^{K} \setminus \{ \bm{0} \}}), \ \i \in I_{\infty,K},
\label{eq: Aldous-Hoover_JE}
\end{equation}
where $\mathfrak{g}$ is now a map from $[0,1]^{2^{K}-1}$ into $\R^{p}$.

Suppose that we observe $\{ \bX_{\i}:\i\in I_{n,K} \}$ with $n \ge K$ and are interested in distributional approximation of the polyadic sample mean 
\begin{align*}
\bS_n : = \frac{(n-K)!}{n!} \sum_{\i \in I_{n,K}} \bX_{\i}.
\end{align*}
in the high-dimensional setting where the dimension $p$ is allowed to entail $p \gg n$. 

As in Section \ref{sec:multiway_clustered_data}, define $\calE_{k} = \{ \e= (e_1,\dots,e_K)  \in \{ 0,1 \}^{K}: \sum_{k=1}^{K} e_k = k \}$ for $1 \le k \le K$. 
The analysis of the jointly exchangeable array relies on the following decomposition
%$\bX_{\i}$:
%\begin{align*}
%\bX_{\i} &=  \sum_{k=1}^K \E[\bX_{\i} \mid U_{i_k}] +\left(\E[\bX_{\i} \mid U_{i_1},\dots,U_{i_K}]-\sum_{k=1}^K \E[\bX_{\i} \mid U_{i_k}]\right)\\
%&\quad +
%\sum_{k=2}^K\left(\E[\bX_{\i}\mid (U_{\{\i \odot \e\}^+})_{\e \in \cup_{r=1}^k \calE_r} ]-\E[\bX_{\i}\mid (U_{\{\i\odot \e\}^+})_{\e \in \cup_{r=1}^{k-1}\calE_r} ]\right).
%\end{align*}
%This leads to the decomposition
\begin{equation}
\begin{split}
\bS_n &= \frac{1}{n}\sumj \E\left[\frac{(n-K)!}{(n-1)!}\sum_{k=1}^K\sum_{\i\in I_{n,K}: i_k=j} \bX_{\i}~\Big|~U_{j}\right] \\
&\quad +\frac{(n-K)!}{n!} \sum_{\i \in I_{n,K}} \left(\E[\bX_{\i} \mid U_{i_1},\dots,U_{i_K}]-\sum_{k=1}^K \E[\bX_{\i} \mid U_{i_k}]\right) \\
&\quad +\sum_{k=2}^K \frac{(n-K)!}{n!} \sum_{\i \in I_{n,K}}\left(\E[\bX_{\i}\mid (U_{\{\i\odot \e\}^+})_{\e \in \cup_{r=1}^k \calE_r} ]-\E[\bX_{\i}\mid (U_{\{\i\odot \e\}^+})_{\e \in \cup_{r=1}^{k-1}\calE_r} ]\right).
\end{split}
\label{eq: H-decomp polyadic}
\end{equation}
It turns out that the first term on the right-hand side, which we call the the H\'{a}jek projection of $\bS_{n}$, is a dominant term. 
%The second term on the right-hand side of (\ref{eq: H-decomp polyadic}) is a degenerate $U$-statistic and thus negligible compared with the first term under moment conditions (this term can be expanded into $K-1$ terms each of which scales as $O(n^{-k/2})$ if $p$ is fixed for $k=2,\dots,K$).  The analysis of the third term is more complicated but it will be shown that the $k$-th term inside the first summation scales as $O(n^{-k/2})$ if the dimension $p$ is fixed, so that the third term on the right-hand side of (\ref{eq: H-decomp polyadic}) is also negligible compared with the first term. See Appendix \ref{sec:maximal_inequality_for_polyadic_data} for details.  Applying the Hoeffding decomposition to the second term on the right-hand side of (\ref{eq: H-decomp polyadic}), combining it with the third term on the right-hand side of (\ref{eq: H-decomp polyadic}), and aligning the terms according to their orders, we can obtain a Hoeffding-type decomposition for jointly exchangeable arrays. As in the case of separate exchangeability, we call the first term on the right-hand side of (\ref{eq: H-decomp polyadic}) the H\'{a}jek projection of $\bS_{n}$. 
Defining $h_k(u)=\E[\bX_{(1,\dots,K)} \mid U_{k}=u]$ for $k=1,\dots,K$, we can simplify the H\'ajek projection into $n^{-1}\sum_{i=1}^n \bW_j$ where $\bW_j = \sum_{k=1}^K h_k(U_j)$. 
%\begin{align*}
%\bS_n^W=
%\frac{1}{n}\sumj \bW_{j},
%\quad \text{with} \quad
%\bW_j=\sum_{k=1}^K h_k(U_{j}).
%%=\E\left[\frac{(n-K)!}{(n-1)!}\sum_{k=1}^K\sum_{\i\in I_{n,K}: i_k=j} \bX_{\i}\Bigg\mid U_{\{j\}}\right].
%\end{align*}
% (n-1) P (K-1)
%Since $\{ \bW_j \}_{j=1}^n$ are i.i.d., we can expect that $\sqrt{n}\bS_n^W$ can be approximated (in distribution) by $N(\bm{0},\Sigma)$, where $\Sigma =\E\left[\bW_{1}\bW_{1}^{T}\right]$. 
%This suggests the following version of multiplier bootstrap for jointly exchangeable arrays. 
%%%%%%%%%%%%%%%%%%%%%%%%%%%%%%%%%%%%%%%%%%%%%%%%%%%%%%%%%%%%%%%%%%%%%
%\subsection{Multiplier bootstrap for jointly exchangeable arrays}
%%%%%%%%%%%%%%%%%%%%%%%%%%%%%%%%%%%%%%%%%%%%%%%%%%%%%%%%%%%%%%%%%%%%%

%%%%%%%%%%%%%%%%%%%%%%%%%%%%%%%%%%%%%%%%%%%%%%%%%%%%%%%%%%%%%%%%%%%%%
\subsection{High-dimensional CLT for jointly exchangeable arrays}
%%%%%%%%%%%%%%%%%%%%%%%%%%%%%%%%%%%%%%%%%%%%%%%%%%%%%%%%%%%%%%%%%%%%%

We consider to approximate the distribution of $\sqrt {n} \bS_{n}$ by a Gaussian distribution on the set of rectangles 
$\calR$ as defined in Section \ref{sec:multiway_clustered_data}.

Let $D_n \ge 1$ be a given constant that may depend on $n$, and $\underline{\sigma} > 0$ be another given constant independent of $n$. 
We will assume either of the following moment conditions.
\begin{align}
&\max_{1 \le \ell \le p} \| X_{(1,\dots,K)}^{\ell} \|_{\psi_{1}} \le D_{n}, \quad \text{or}  \label{eq:condition1polyadic}\\
&\E[\| \bX_{(1,\dots,K)} \|_{\infty}^{q}] \le D_{n}^{q} \quad \text{for some $q \in (4,\infty)$}.\label{eq:condition1polyadic_poly}
\end{align}
We will also assume the following condition.
\begin{align}
\max_{1 \le \ell \le p} \E[|W_{1}^{\ell}|^{2+k}] \le D_{n}^{k}, \ \kappa=1,2, \quad \text{and} \quad \min_{1 \le \ell \le p} \E[|W_{1}^{\ell}|^{2}] \ge \underline{\sigma}^{2} \label{eq:condition2polyadic}
%&\min_{1 \le \ell \le p} \E[|W_{1}^{\ell}|^{2}] \ge \underline{\sigma}^{2}. \label{eq:condition3polyadic}
\end{align}

The conditions required here are similar to those in the case of separate exchangeability in Section \ref{sec:multiway_clustered_data}. The main difference is that Condition (\ref{eq:condition2polyadic}) is  now imposed on $\bW_1$.

Let $\gamma_{\Sigma} = N(\bm{0},\Sigma)$ with $\Sigma =\E\left[\bW_{1}\bW_{1}^{T}\right]$.
\begin{theorem}[High-dimensional CLT for jointly exchangeable arrays]\label{thm: high-d CLT polyadic}
Suppose that either Condition (\ref{eq:condition1polyadic})  or (\ref{eq:condition1polyadic_poly}) holds, and further Condition (\ref{eq:condition2polyadic}) holds. Then, there exists a constant $C$ such that 
\[
\begin{split}
\sup_{R\in \calR}\left|\Prob(\sqrt{n} \bS_n\in R)-\gamma_{\Sigma}(R) \right|
 \le
\begin{cases}
C \left ( \frac{D_{n}^{2} \log^{7} (pn)}{n} \right )^{1/6} & \text{if (\ref{eq:condition1polyadic}) holds,} \\
C \left [ \left ( \frac{D_{n}^{2} \log^{7} (pn)}{n} \right )^{1/6} + \left ( \frac{D_{n}^{2} \log^{3} (pn)}{n^{1-2/q}} \right )^{1/3} \right ] & \text{if  (\ref{eq:condition1polyadic_poly}) holds,}
\end{cases}
\end{split}
\]
where the constant $C$ depends only on $\underline{\sigma}$ and $K$ if Condition (\ref{eq:condition1polyadic}) holds, while $C$ depends only on $q,\underline{\sigma}$, and $K$ if Condition (\ref{eq:condition1polyadic_poly}) holds. 
\end{theorem}

\begin{remark}[Comparison with \cite{Silverman1976}]
Theorem \ref{thm: high-d CLT polyadic} is a high-dimensional extension of Theorem A in \cite{Silverman1976} that establishes a CLT for jointly exchangeable arrays with fixed $p$. The covariance matrix of the limiting Gaussian distribution in \cite{Silverman1976} has a different expression than our $\Sigma$, but we will verify below that two expressions are indeed the same. The covariance matrix given in Corollary to Theorem A in \cite{Silverman1976} reads as follows: Let $\check{\bX}_{(i_1,\dots,i_K)}$ be the symmetrized version of $\bX_{(i_1,\dots,i_K)}$, i.e., $\check{\bX}_{(i_1,\dots,i_K)} = (K!)^{-1} \sum_{(i_1',\dots,i_K')} \bX_{(i_1',\dots,i_K')}$ where the summation is taken over all permutations of $( i_1,\dots, i_K )$. The covariance matix given in \cite{Silverman1976} is $\Sigma_{S} = K^2\E[\check{\bX}_{(1,\dots,K)} \check{\bX}_{(1,K+1,\dots,2K)}]$. On the other hand, 
$\sum_{k=1}^{K} \E[ \bX_{(1,\dots,K)} \mid U_{k} = u] = \sum_{k=1}^{K} \E[\check{\bX}_{(1,\dots,K)} \mid U_{k}=u] = K \E[ \check{\bX}_{(1,\dots,K)} \mid U_{1} = u]$,
so that \\$\Sigma = K^2 \E\left [ \E[ \check{\bX}_{(1,\dots,K)} \mid U_{1} ] \E[ \check{\bX}_{(1,\dots,K)} \mid U_{1} ] \right ] =  K^2\E[\check{\bX}_{(1,\dots,K)} \check{\bX}_{(1,K+1,\dots,2K)}] = \Sigma_{S}$.
\end{remark}

\subsection{Multiplier bootstrap for jointly exchangeable arrays}
Let $\{ \xi_j \}_{j=1}^{n}$ be independent $N(0,1)$ random variables independent of the data. 
Ideally, we want to make use of the multiplier bootstrap statistic $
n^{-1}\sumj \xi_j(
\bW_{j}-K\bS_n).
$
This is infeasible, however, as the projections $\bW_j$ are unknown. 
As an alternative, we replace each $\bW_j$ by its estimate
\[
\hat{\bW}_j =\frac{(n-K)!}{(n-1)!}\sum_{k=1}^K\sum_{\i\in I_{n,K}: i_k=j} \bX_{\i},
\]
and apply the multiplier bootstrap to $\hat{\bW}_{j}$, i.e., 
\begin{align*}
\bS_n^{MB} :
=\frac{1}{{n}}\sumj \xi_j (\hat{\bW}_j - K\bm{S}_{n})
\end{align*}
When $K=2$ (dyadic), this mulitplier bootstrap 
%simplifies into
%\[
%\bS_n^{MB} =\frac{1}{{n}}\sumj \xi_j \left \{ \frac{1}{(n-1)}\sum_{i' = 1; i' \ne j}^{n}(\bX_{(i',j)}+\bX_{(j,i')}) - 2\bS_n \right \},
%\]
%which 
coincides with the multiplier bootstrap statistic considered in Section 3.2 of \cite{DDG2019}.
However, \cite{DDG2019} do not consider the extension to general $K$ arrays, and focus on the empirical process indexed by a Donsker class, which excludes the high-dimensional sample mean. 
We will study the validity of this multiplier bootstrap for jointly exchangeable arrays.

%%%%%%%%%%%%%%%%%%%%%%%%%%%%%%%%%%%%%%%%%%%%%%%%%%%%%%%%%%%%%%%%%%%%%%%%%%%%%%%%%%%%%%%%%%
Let $\Prob_{|\bX_{I_{n,K}}}$ denote the law conditional on the data $(\bX_{\i})_{\i \in I_{n,K}}$
%Define
%\begin{align*}
%\hat \Delta_{W,1}=\max_{1\le \ell \le p} \frac{1}{n}\sumj (\hat W_{j}^\ell - W_{j}^\ell  )^2.
%\end{align*}
and $\overline \sigma=\max_{1\le \ell \le p}\sqrt{\E[|W_{1}^\ell|^2]}$.

\begin{theorem}[Validity of multiplier bootstrap for jointly exchangeable arrays]\label{thm:bootstrap validity polyadic}
Consider the following two cases.
\begin{enumerate}
\item[(i).] \red{Assume that} Conditions (\ref{eq:condition1polyadic}) and (\ref{eq:condition2polyadic}) hold, and \red{further} there exist constants $C_1$ and $\zeta \in (0,1)$ such that 
\begin{align}
%&\Prob \left ( \overline \sigma^2 \hat \Delta_{W,1}\log^4 p >C_1 n^{-\zeta_2} \right)\le C_1 n^{-1} \quad \text{and}\label{eq:bootstrap_Delta_1_polyadic} \\
\frac{\overline \sigma^2 D_{n}^2 \log^{7} p }{n} \bigvee \frac{D_n^2 (\log^2 n ) \log^5(p n)}{n} \le C_1 n^{-\zeta}. \label{eq:bootstrap_rate_polyadic}
\end{align}
\item[(ii).] \red{Assume that} Conditions (\ref{eq:condition1polyadic_poly}) and (\ref{eq:condition2polyadic}) hold, and \red{further} there exist constants $C_1$ and $\zeta \in (2/q,1)$ such that 
\begin{equation}
\frac{\overline \sigma^2 D_{n}^2 \log^{5} (pn) }{n} \bigvee \left ( \frac{D_{n}^2  \log^3 p}{n^{1-4/q}} \right)^2 \le C_1 n^{-\zeta}.  \label{eq:bootstrap_rate_polyadic_poly}
\end{equation}
\end{enumerate}
Then, under Case (i), for any $\nu \in (1/\zeta,\infty)$, there exists a constant $C$ depending only on $\nu, \underline{\sigma}, K$, and $C_1$ such that
$
\sup_{R\in \calR}\left|\Prob_{|\bX_{I_{n,K}}}(\sqrt{n} \bS_n^{MB}\in R)-\gamma_{\Sigma}(R) \right|\le C n^{-(\zeta-1/\nu)/4}
$
with probability at least $1-Cn^{-1}$. Under Case (ii), the same conclusion holds with $n^{-(\zeta-1/\nu)/4}$ replaced by $n^{-(\zeta-2/q)/4}$, while the constant $C$ depends only on $q, \underline{\sigma}, K$, and $C_1$.
\end{theorem}

%The following proposition provides  primitive sufficient conditions for Condition (\ref{eq:bootstrap_Delta_1_polyadic}) to hold. 
%\begin{proposition}[Primitive sufficient conditions for Condition (\ref{eq:bootstrap_Delta_1_polyadic})]\label{proposition:estimation_errors_polyadic}
%Consider the following two cases.
%\begin{enumerate}
%\item[(i')]
%Conditions (\ref{eq:condition1polyadic}), (\ref{eq:condition2polyadic}), and (\ref{eq:condition3polyadic}) hold, and there exist constants $C_1$ and $\zeta \in (0,1)$ such that 
%\[
%\frac{\overline \sigma^2 D_{n}^2 \log^{7} p }{n}
%\le 
%C_1 n^{-\zeta}.
%\]
%\item[(ii')] Conditions (\ref{eq:condition1polyadic_poly}), (\ref{eq:condition2polyadic}), and (\ref{eq:condition3polyadic}) hold, and there exist constants $C_1$ and $\zeta \in (2/q,1)$ such that
%\[
%\frac{\overline \sigma^2 D_{n}^2 \log^{5} p }{n}
%\le 
%C_1 n^{-\zeta}.
%\]
%\end{enumerate}
%Under Case (i'), for any $\nu \in(1/\zeta,\infty)$, there exists a constant $C$ depending only on $\nu, K$, and $C_1$ such that 
%\[
%\Prob\left(\overline \sigma^2 \hat \Delta_{W,1}\log^4 p >Cn^{-\zeta+1/\nu} \right)\le C n^{-1}.
%\]
%Under Case (ii'), there exists a constant $C$ depending only on $q, K$, and $C_1$ such that 
%\[
%\Prob\left(\overline \sigma^2 \hat \Delta_{W,1}\log^4 p >Cn^{-\zeta+2/q} \right)\le C n^{-1}.
%\]
%\end{proposition}

%%%%%%%%%%%%%%%%%%%%%%%%%%%%%%%%%%%%%%%%%%%%%%%%%
\begin{remark}[Discussion on Conditions (\ref{eq:bootstrap_rate_polyadic}) and (\ref{eq:bootstrap_rate_polyadic_poly})]\label{remark:discussion_rates_polyadic}
Similar to Remark \ref{remark:discussion_polynomial_rates}, if one is interested only in bootstrap consistency,
	Conditions (\ref{eq:bootstrap_rate_polyadic}) and (\ref{eq:bootstrap_rate_polyadic_poly}) can be weakened to $(\overline{\sigma}^2D_{n}^{2} \log^{7}p) \vee (D_{n}^2\log^5 (pn)) = o(n)$ and $(n^{-1}\overline{\sigma}^2D_{n}^{2} \log^{5}(pn)) \vee (n^{-(1-2/q)}D_{n}^{2}\log^{3}p) = o(1)$, respectively. 
%	In addition, to show $\overline{\sigma} \hat{\Delta}_{W}\log^4 p = o_{P}(1)$,
%	the two rate conditions in Proposition \ref{proposition:estimation_errors_polyadic} can be weakened to $\overline \sigma^2 D_n^2 \log^7 p=o(n)$ and $\overline \sigma^2 D_n^2 \log^5 p=o(n^{1-2/q})$, respectively.
\end{remark}
%%%%%%%%%%%%%%%%%%%%%%%%%%%%%%%%%%%%%%%%%%%%%%%%%%

%%%%%%%%%%%%%%%%%%%%%%%%%%%%%%%%%%%%%%%%%%%%%%%%%%%%%%%%%%%%%%%%%%%%%
\section{Applications}\label{sec:applications}
%%%%%%%%%%%%%%%%%%%%%%%%%%%%%%%%%%%%%%%%%%%%%%%%%%%%%%%%%%%%%%%%%%%%%
In this section, we illustrate a couple of applications of our bootstrap methods. 
Section \ref{sec:density} is concerned with construction of confidence bands for densities of flows in dyadic data.
Section  \ref{sec:multiway_lasso} is concerned with  penalty choice for the Lasso and the performance of the corresponding estimate.

%%%%%%%%%%%%%%%%%%%%%%%%%%%%%%%%%%%%%%%%%%%%%%%%%%%%%%%%%%%%%%%%%%%%%
\subsection{Confidence bands for densities of flows in dyadic data}\label{sec:density}
%%%%%%%%%%%%%%%%%%%%%%%%%%%%%%%%%%%%%%%%%%%%%%%%%%%%%%%%%%%%%%%%%%%%%

Researchers are often interested in ``the densities of migration across states, trade across nations, liabilities across banks, or minutes of telephone conversation among individuals'' \citep{Graham2019kernel}.
Densities of these flow measures use dyadic data.
We illustrate an application of our method in Section \ref{sec:polyadic_data} to constructing  confidence bands for such density functions. We refer the reader to \cite{bickel1973, claeskens2003, cck2014density} as references on confidence bands for density estimation with i.i.d. data. 

Following \cite{Graham2019kernel}, suppose that we observe the dyadic data $\{Y_{ij} : 1 \le i \ne j \le n \}$ that admits the structure
\begin{align}
Y_{ij}=\mathfrak g (U_i,U_j,U_{\{i,j\}}) \label{eq:density_DGP}
\end{align}
where $\mathfrak{g}$ is symmetric in the first two arguments and hence $Y_{ij}=Y_{ji}$. 
We are interested in inference on the density of $Y_{ij}$. However, in certain empirical applications, such as international trade \citep[see][]{head2014gravity}, a proportion of the variable of interest is zero. %(\textcolor{red}{example?}). 
Hence we assume that 
$Y_{ij}$ has a probability mass at zero, i.e. 
$Y_{ij}$ is such that $\Prob(Y_{ij}\ne 0)=a\in(0,1]$, and $Y_{ij}\sim f$ when $Y_{ij}\ne 0$, where $f$ is a density function on $\R$. Let $b(y)=af(y)$ denote the scaled density.
We may estimate $f(\cdot) = b(\cdot)/a$ by $\hat f(\cdot) =\hat b(\cdot)/\hat a$, where
$
\hat a={n\choose 2}^{-1}\sum_{1\le i<j\le n} \1(Y_{ij}\ne 0)$ and $\hat b(y)={n\choose 2}^{-1}\sum_{1\le i<j\le n} K_h(y- Y_{ij})\1(Y_{ij}\ne 0).
$
Here $K:\R\to \R$ is a kernel function (a function that integrates to one), $K_h(\cdot):=h^{-1} K(\cdot/h)$,  and  $h=h_n\to 0$ is a bandwidth.

We consider to construct simultaneous confidence intervals (bands) for $f$ over the set of design points $y_1,\dots,y_p$, where $p=p_n \to \infty$ is allowed. 
	Define
	\begin{align*}
	 \tilde X_{ij}^\ell=\left\{ \frac{K_h(y_\ell-Y_{ij})}{\hat a} -\frac{ \hat b(y_\ell)}{\hat a^2}  \right\}\1(Y_{ij}\ne 0), \quad 1 \le i < j \le n, \quad \tilde{X}_{ij}^{\ell} = \tilde{X}_{ji}^{\ell}, \quad 1 \le j  < i \le n, 
	\end{align*}
	for $\ell=1,\dots,p$. 
Then, the multiplier bootstrap statistic is given by
\begin{align*}
&\tilde \bS_n^{MB}=\frac{1}{n}\sum_{i=1}^n \xi_i (\tilde \bW_i-2\tilde \bS_n),  \text{with} \ 
 \tilde \bS_n=\frac{1}{n(n-1)}\sum_{1\le i\ne j \le n}\tilde \bX_{ij}  \ \text{and} \  \tilde \bW_i=\frac{1}{n-1}\sum_{j \ne i} 2\tilde \bX_{ij},
\end{align*}
where $\sum_{j \ne i} = \sum_{j \in \{1,\dots,n\} \setminus \{ i \}}$. 
For a given $\alpha\in (0,1)$, consider the $(1-\alpha)$-simultaneous confidence intervals defined by 
\begin{align*}
\mathcal I(1-\alpha):=\prod_{\ell=1}^p \left[\hat f (y_\ell)\pm \frac{ \tilde c(1-\alpha)}{\sqrt{n}}\right] \quad \text{and} \quad \mathcal I^N(1-\alpha):=\prod_{\ell=1}^p \left[\hat f (y_\ell)\pm \frac{\tilde \sigma_\ell \tilde c^N(1-\alpha)}{\sqrt{n}}\right],
\end{align*} 
where
$\tilde\sigma_\ell^2 = n^{-1}\sum_{i=1}^n( \tilde W_i^\ell - 2\tilde S_n^\ell )^2
$, $\tilde\Lambda=\diag(\tilde\sigma_1^2,\dots,\tilde \sigma_p^2)$, $\tilde c(1-\alpha)$ is the conditional $(1-\alpha)$-quantile of $\|\sqrt{n}\tilde\bS_n^{MB}\|_\infty$, and $\tilde c^N(1-\alpha)$ is the conditional $(1-\alpha)$-quantile of $\|\sqrt{n}\hat\Lambda^{-1/2}\tilde\bS_n^{MB}\|_\infty$. The first method  $\mathcal I(1-\alpha)$ is a constant-length confidence band, while the second method $\mathcal I^N(1-\alpha)$ is a variable-length confidence band based on Studentization.

%%%%%%%%%%%%%%%%%%%%%%%%%%%%%%%%%%%%%%%%%%%%%%%%

The following proposition establishes asymptotic validity of the confidence bands.  
We will assume that there exists a conditional density of $Y_{ij}$ given $U_i$ and $Y_{ij} \ne 0$, denoted by $f_{Y_{12} \mid U_1, Y_{12} \ne 0}(y \mid u)$ (more formally, we assume that the conditional distribution of $Y_{ij}$ given $U_i$ is $\Prob (Y_{ij} \in dy \mid U_i) = \Prob (Y_{ij} = 0 \mid U_i) \delta_{0}(dy) + \Prob (Y_{ij} \ne 0 \mid U_i) f_{Y_{12} \mid U_{1},Y_{12} \ne 0} (y \mid U_i) dy$, where $\delta_{0}$ is the Dirac delta at $0$). 
Let $\overline f_h(y)=\int K_h(y-z)f(z)dz $ and $\overline f_h (y\mid u)=\int K_{h}(y-z) f_{Y_{12}\mid U_1, Y_{12}\ne 0}(z \mid u)dz$ denote the surrogate density and conditional density, respectively. Recall that a kernel $K$ is an $r$-th order kernel for some $r \ge 2$ if $\int y^{t} K(y) dy = 0$ for $t =1,\dots,r-1$ and $\int |y^r K(y)| dy < \infty$. Let $M$, $h_0$, $\sigma_0$, and $a\in(0,1]$ be given positive constants independent of $n$.
%%%%%%%%%%%%%%%%%%%%%%%%%%%%%%%%%%%%%%%%%%%%%%%%
\begin{proposition}\label{proposition:kernel_density_mixture}
	Suppose that: (i) the data is generated following Equation (\ref{eq:density_DGP}) with point mass at zero, $\Prob(Y_{ij}\ne 0)=a$ and $Y_{ij}\sim f$ \red{with probability $a$}; (ii) $\| f \|_{\infty} \le M$ and  $\sup_{y\in \R, u\in [0,1]}|f_{Y_{12}\mid U_1, Y_{12}\ne 0}(y \mid u)|\le M$;  %the conditional density $f_{Y_{12}\mid U_1}$ is bounded by some constant independent of $n$; 
	(iii) for the set of non-zero design points $\{y_1,\dots,y_p\} \subset \R$ and $h\le h_0$, \\
	$\Var\left(\overline f_h(y_\ell\mid U_1) \cdot \Prob(Y_{12}\ne 0\mid U_1)\right) \ge \sigma_0^2 $;
	 %$\min_{1\le \ell \le p} \Var( f_{Y_{12}\mid U_1, Y_{12}\ne 0}(y\mid U_1)\cdot\Prob(Y_{12}\ne 0|U_1))$ 
	  %for the set of non-zero design points $\{y_1,\dots,y_p\} \subset \R$, $\min_{1\le \ell \le p} \Var( f_{Y_{12}\mid U_1} (y_\ell\mid U_1))$ is greater than some positive constant independent of $n$;
	 (iv) the kernel $K$ is a bounded $r$-th order kernel for some $r\ge 2$; (v) the bandwidth satisfies $h\to 0, nh^2 \to \infty$ as $n\to \infty$ and $\log^7 (pn)=o(nh^2)$. Then we have
	\begin{align*}
	\Prob\left(\left(\overline f_h (y_\ell)\right)_{\ell=1}^p\in \mathcal I(1-\alpha)\right) \to (1-\alpha) \quad \text{and} \quad  	\Prob\left(\left (\overline f_h (y_\ell)\right)_{\ell=1}^p\in \mathcal I^N(1-\alpha)\right) \to (1-\alpha).
	\end{align*}
	In addition, if $f$ is $r$-continuously differentiable, $\|f^{(r)}\|_\infty<\infty$, and  $nh^{2r}\log p =o(1)$, then 
	\begin{align*}
		\Prob\left(\left ( f (y_\ell)\right)_{\ell=1}^p\in \mathcal I(1-\alpha) \right)\to (1-\alpha)  \quad \text{and} \quad \Prob\left(\left (f (y_\ell)\right)_{\ell=1}^p\in \mathcal I^N(1-\alpha) \right)\to (1-\alpha).
	\end{align*}
\end{proposition}

Some comments on the proposition are in order. 
%%%%%%%%%%%%%%%%%%%%%%%%%%%%%%%%%%%%%%%%%%%%%%%%
\begin{remark}
(i) The assumption that $\mathfrak g$ in  (\ref{eq:density_DGP}) is symmetric in its first two arguments can in fact be relaxed. In such case, the conclusions in Proposition \ref{proposition:kernel_density_mixture} continue to hold under a few minor modifications to the regularity conditions. 
%the additional assumptions that   $\sup_{y\in \R, u\in [0,1]}|f_{Y_{12}\mid U_2, Y_{12}\ne 0}(y \mid u)|\le M$ and $\Var\left(\int K_{h}(y-z) f_{Y_{12}\mid U_2, Y_{12}\ne 0}(z \mid U_2)dz \cdot \Prob(Y_{12}\ne 0\mid U_2)\right)\ge \sigma_0^2 $. 
Also, when $a=1$ and $r=2$, the proposed dyadic kernel density estimator reduces to the estimator of \cite{GrahamNiuPowell2020kernel}. The proposition complements \cite{GrahamNiuPowell2020kernel} by providing valid simultaneous confidence intervals for their dyadic kernel density estimator.
(ii) In some applications, such as in our empirical illustration in Section \ref{sec:empirical}, the object of interest is $b(\cdot)$. For such case, one can simply omit the estimation of $a$ by setting $\hat a=1$ while keeping $\hat b(\cdot ) $ unaltered. The conclusions in Proposition \ref{proposition:kernel_density_mixture} continue to hold with this modification. 
(iii) 	The proof of Proposition \ref{proposition:kernel_density_mixture} does not follow directly from the results of Section \ref{sec:polyadic_data}, as we have to handle the estimation errors of $\hat a$ and $\hat b(\cdot)$, which involves additional substantial work.
\end{remark}
\subsection{Penalty choice for Lasso under separate exchangeability}\label{sec:multiway_lasso}

Consider a regression model
\begin{align*}
Y_{\bm{i}}= f(\bm{Z}_{\bm{i}}) +\varepsilon_{\bm{i}},  \quad \E[ \varepsilon_{\bm{i}} \mid \bm{Z}_{\bm{i}}] = 0,  \quad \bm{i} = (i_1,\dots,i_K) \in [\bm{N}] = \prod_{k=1}^{K} \{ 1,\dots,N_k \},
\end{align*} 
where $Y_{\bm{i}}$ is a scalar outcome variable, $\bm{Z}_{\bm{i}}\in \R^d$ is a $d$-dimensional vector of covariates, $f:\R^d \to \R $ is an unknown regression function of interest, and $\varepsilon_{\bm{i}}$ is an error term.
We approximate $f$ by a linear combination of technical controls $\bm{X}_{\bm{i}}=P(\bm{Z}_{\bm{i}})$ for some transformation $P:\R^d \to \R^p$, i.e., $f(\bm{Z}_{\bm{i}})=\bm{X}_{\bm{i}}^T\beta_0 + r_{\bm{i}},\, \bm{i}\in [\bm{N}]$, 
where $r_{\bm{i}}$ is a bias term. The dimension $p$ can be much larger than the cluster sizes $\bN$, but we assume that the vector $\beta_0\in \R^p$ is sparse in the sense that $\|\beta_0\|_{0}=s\ll n$ with $n = \min_{1 \le k \le K} N_k$. 
Suppose that the array $\big ( (Y_{\bm{i}},\bm{Z}_{\bm{i}}^T)^T \big)_{\i \in \N^{K}}$ is separately exchangeable and generated as
\[
(Y_{\bm{i}},\bm{Z}_{\bm{i}}^T)^T= \mathfrak{g} ((U_{\bm{i} \odot \bm{e}})_{\bm{e} \in \{ 0,1 \}^{K} \setminus \{ \bm {0} \}}), \ \i \in \N^{K}, \quad \{ U_{\bm{i} \odot \bm{e}} : \bm{i} \in \N^{K}, \bm{e} \in \{ 0,1 \}^{K} \setminus \{ \bm {0} \} \} \stackrel{i.i.d.}{\sim} U[0,1],
\]
for some Borel measurable map $\mathfrak{g}: [0,1]^{2^{K}-1} \to \R^{1+d}$. 

Arguably, one of the most popular estimation methods for such a high-dimensional regression problem is the Lasso  \citep{tibshirani1996}; we refer to \cite{vandegeer2011, giraud2015, wainwright2019} as standard references on high-dimensional statistics. 
Let $N=\prod_{k=1}^K N_k$ denote the total sample size. 
The Lasso estimate for $\beta_0$ is defined by
\[
\hat \beta^\lambda=\argmin_{\beta\in \R^p}\left\{ \frac{1}{N}\sum_{\bm{i}\in [\bm{N}]}(Y_{\bm{i}}-\bm{X}_{\bm{i}}^T\beta)^2+ \lambda\|\beta\|_1\right\},
\]
where $\lambda>0$ is a penalty level. We estimate the vector $\bm{f} = (f_{\i})_{\i \in [\bN]} = (f(\bm{Z}_{\i}))_{\i \in [\bN]}$ by $\hat{\bm{f}}^{\lambda} = (\bX_{\i}^{T}\hat{\beta}^{\lambda})_{\i \in [\bN]}$. Let $\| \bm{t} \|_{N,2}^2=N^{-1}\sum_{\bm{i}\in [\bm{N}]} t_{\bm{i}}^2$ for $\bm{t} = (t_{\bm{i}})_{\bm{i} \in [ \bm{N} ]}$.

In what follows, we discuss the statistical performance of the Lasso estimate. 
Following \cite{bickel2009}, we say that Condition RE$(s,c_0)$ holds (RE refers to ``restricted eigenvalue'') if, for a given positive constant $c_0\ge 1$, the inequality
\[
\kappa(s,c_0)=\min_{\substack{J\subset\{1,\dots,p\}\\1\le |J|\le s}}\inf_{\substack{\theta\in \R^p,\,\theta\ne 0 \\\|\theta_{J^c}\|_1\le c_0\|\theta_{J}\|_1}}  \frac{\sqrt{s N^{-1}\sum_{\bm{i}\in [\bm{N}]}(\theta^T\bm{X}_{\bm{i}})^2}}{\|\theta_{J}\|_1}>0
\]
holds with $J^c= \{1,\dots,p\}\setminus J$. Here for $\theta = (\theta_1,\dots,\theta_p)^{T}$ and $J \subset \{1,\dots,p \}$, $\theta_{J} = (\theta_j)_{j \in J}$. 
%Keep in mind that as the covariates are random, the restricted eigenvalue $\kappa (s,c_0)$ is random as well.

In addition, to guarantee fast rates for the Lasso, it is important to choose the penalty level $\lambda$ in such a way that $\lambda\ge 2c\| \bm{S}_{\bN} \|_\infty$ with $\bm{S}_{\bm{N}}=N^{-1}\sum_{\bm{i}\in [\bm{N}]}\varepsilon_{\i} \bX_{\i}$ for some $c > 1$ \citep{bickel2009,BC2013}. 
%Theorem 1 of \cite{BC2013} implies that if, for a given $c>1$, i) $\lambda\ge 2c\| \bm{S}_{\bN} \|_\infty$ with $\bm{S}_{\bm{N}}=N^{-1}\sum_{\bm{i}\in [\bm{N}]}\varepsilon_{\i} \bX_{\i}$ and  ii) Condition RE$(s,c_0)$ holds with $c_0=(c+1)/(c-1)$, 
%then the following nonasymptotic bounds hold with $\kappa=\kappa(s,c_0)$:
%\[
%\|\hat{\bm{f}}^\lambda - \bm{f}\|_{N,2}\le 3\|\bm{r}\|_{N,2} +\left(1+\frac{1}{c}\right) \frac{\lambda\sqrt{s}}{\kappa}.
%\]
To this end, we shall estimate  the $(1-\eta)$-quantile of $2c \|\bm{S}_{\bN}\|_{\infty}$ for some small $\eta>0$.
We first estimate the error terms $\varepsilon_{\i}$ by pre-estimating $\beta_0$ by the preliminary Lasso estimate $\tilde{\beta} = \hat{\beta}^{\lambda_{0}}$ with penalty $
\lambda^0 =\tau_{n} (n^{-1}\log p)^{1/2}$ for some slowing growing sequence $\tau_{n} \to \infty$. In the following, we take $\tau_n = \log n$ for the sake of simplicity but other choices also work. 
We apply the multiplier bootstrap to
$\tilde{\bm{S}}_{\bN}=N^{-1}\sum_{\i \in [\bN]}\tilde\varepsilon_\i \bX_{\i}$ instead of $\bm{S}_{\bN}$. 

The H\'{a}jek projection to $\bm{S}_{\bN}$ is given by $\sum_{k=1}^{K} N_{k}^{-1}\sum_{k=1}^{N_{k}}\bm{V}_{k,i_{k}}$, where $\bm{V}_{k,i_{k}}$ is given by $\bV_{k,i_k}=\E[ \varepsilon_{(1,\dots,1,i_k,1,\dots,1)}\bX_{(1,\dots,1,i_k,1,\dots,1)}\mid U_{(0,\dots,0,i_k,0,\dots,0)}]$. 
We estimate $\bV_{k,i_k}$ by \\
$\tilde{\bm{V}}_{k,i_k}=\big(\prod_{k'\ne k} N_{k'} \big)^{-1}\sum_{i_1,\dots,i_{k-1},i_{k+1},\dots,i_K} \tilde\varepsilon_\i \bX_{\i}$. Let $\{ \xi_{1,i_{1}} \}_{i_1=1}^{N_{1}}, \dots, \{ \xi_{K,i_{K}} \}_{i_K=1}^{N_K}$ be i.i.d. $N(0,1)$ variables independent of the data, and consider
\begin{align*}
\Lambda_\bN^\xi = \left\|\sum_{k=1}^K\frac{1}{N_k} \sum_{i_k=1}^{N_k} \xi_{k,i_k} (\tilde{\bm{V}}_{k,i_k}-\tilde{\bm{S}}_{\bN})\right\|_\infty.
\end{align*}
We propose to choose $\lambda$ as $\lambda=\lambda(\eta)=2c\Lambda_{\bm{N}}^\xi(1-\eta)$, 
where $\Lambda_{\bm{N}}^\xi(1-\eta)$ denotes the conditional $(1-\eta)$-quantile of $\Lambda_{\bm{N}}^\xi$. 
We allow $\eta$ to decrease with $n$, i.e, $\eta = \eta_n \to 0$. 

The following proposition establishes the asymptotic validity of our choice of $\lambda$ (as $n \to \infty$) under separate exchangeability. In what follows, we understand that $s,p,\bN,\eta$ are functions of $n$ while other parameters such as $c, q, \underline{\kappa}$ are independent of $n$.

\begin{proposition}[Penalty choice for the Lasso under separate exchangeability]\label{proposition:lasso_MB_penalty}
Suppose that: (i) there exist some constants $q \in [4,\infty)$ independent of $n$ and $D_{\bN}$ that may depend on $\bN$ (and thus on $n$) such that $\E[|\varepsilon_{\bm{1}}|^{2q} ]\vee\E[\|\bX_{\bm{1}}\|_{\infty}^{2q} ]\le D_\bN^{q}$ and $ \max_{1\le j\le p}\max_{1\le k\le K}\E[|V_{k,1}^j|^{2+\ell}]\le D_\bN^\ell$ for $\ell=1,2$; (ii) $\E[|V_{k,1}^{j}|^2]$ is bounded and bounded away from zero uniformly in $1 \le j \le p$ and $1 \le k \le K$; (iii) there exists a positive constant $\underline{\kappa}$ independent of $n$ such that $\kappa (s,c_0) \ge \underline{\kappa}$ with probability $1-o(1)$; (iv) as $n \to \infty$, $\| \bm{r} \|_{N,2} = O(\sqrt{(s \log p)/n})$ and $\frac{s \overline N^{1/q}D_\bN^3  \log^7 (p \overline N)}{n}
\bigvee \frac{D_\bN^2  \log^5  (pn)  }{n^{1-2/q}}=o(1)$. 
Then, we have $\lambda \ge 2c \| \bm{S}_{\bN} \|_{\infty}$ with probability $1-\eta - o(1)$. Further, 
we have 
%\[
%\lambda =  O_{P} \left( \sqrt{\frac{\log  p }{n}} \bigvee \sqrt{\frac{\log(1/\eta)}{n}} \right).
%\]
%Consequently, if we take $\eta = \eta_n \to 0$, we have 
$
\| \hat{\bm{f}}^{\lambda} - \bm{f} \|_{N,2} = O_{P} \left( \sqrt{\frac{s\log  p}{n}} \bigvee \sqrt{\frac{s\log(1/\eta)}{n}} \right).
$
\end{proposition}

The proof of Proposition \ref{proposition:lasso_MB_penalty} does not follow directly from the results of Section \ref{sec:multiway_clustered_data}, as we have to take care of the estimation error of the preliminary Lasso estimate $\tilde{\beta}$, which requires  extra work. 
%%%%%%%%%%%%%%%%%%%%%%%%%%%%%%%%%%%%%%%%%%%%%%%%%%%%%%%%%%%%%%%%%%%%%%%%%%%%%%%%%%%%%%%%%

Condition (iii) in the preceding proposition is a high-level condition on the sample gram matrix. 
The following proposition provides primitive sufficient conditions for Condition (iii) to hold for the case of $K=2$. 

%%%%%%%%%%%%%%%%%%%%%%%%%%%%%%%%%%%%%%%%%%%%%%%%%%%%%%%%%%%%%%%%%%%%%%%%%%%%%%%%%%%%%%%%%
\begin{proposition}[RE condition under $K=2$]\label{proposition:lasso_K=2}
Consider $K=2$ and let $B_{\bN} = \sqrt{\E[\max_{\i \in [\bN]}\| \bX_{\i} \|_{\infty}^2]}$. 
Suppose that the eigenvalues of $\E[\bX_{\bm{1}}\bX_{\bm{1}}^T]$ are bounded and bounded away from zero, and $sB_{\bN}^2 \log^4 (p\overline{N}) = o(n)$. Then, there exists a positive constant $\underline{\kappa}$ independent of $n$ such that $\kappa (s,c_0) \ge \underline{\kappa}$ with probability $1-o(1)$. 
\end{proposition}

Under Condition (i) of Proposition \ref{proposition:lasso_MB_penalty}, $B_{\bN} \le \overline{N}^{1/q}D_{\bN}$, so that $sB_{\bN}^2 \log^4 (p\overline{N}) = o(n)$ reduces to $s\overline{N}^{1/q}D_{\bN}\log^{4}(p\overline{N}) = o(n)$, which is implied by Condition (iv) of Proposition \ref{proposition:lasso_MB_penalty}. 

%The proof of Proposition \ref{proposition:lasso_K=2} relies on Lemma 2.7 in \cite{LecueMendelson2017} and an extension of Lemma P.1 in \cite{BCCW2018}, whose proof in turn relies on the techniques in \cite{RudelsonVershynin2008}, from the i.i.d. case to separate exchangeability. 
%%%%%%%%%%%%%%%%%%%%%%%%%%%%%%%%%%%%%%%%%%%%%%%%%%%%%%%%%%%%%%%%%%%%%

%\begin{remark}[Column standardization]
%For intepretability of the Lasso estimate, in practice, we often rescale the penalty by the weighted $\ell_1$-norm (as in \cite{belloni2011} in the quantile regression case)  to make sure that the coefficients are penalized in a comparable manner. All the results in this section continue to hold under this practice as the conditions assumed in Proposition \ref{proposition:lasso_MB_penalty} guarantee the sample second moment of each covariate is consistent uniformly over the coordinates. 
%\end{remark}
%

%%%%%%%%%%%%%%%%%%%%%%%%%%%%%%%%%%%%%%%%%%%%%%%%%%%%%%%%%%%%%%%%%%%%%
\section{Simulation studies}\label{sec:simulation}
%%%%%%%%%%%%%%%%%%%%%%%%%%%%%%%%%%%%%%%%%%%%%%%%%%%%%%%%%%%%%%%%%%%%%

In this section, we present simulation studies to evaluate the finite sample performance of the proposed multiplier bootstrap methods.

We first describe the simulation design for separately exchangeable arrays.
With $\Sigma_{\bm{Z}}$ denoting the $p \times p$ covariance matrix consisting of elements of the form $4^{-\abs{r-c}}$ in its $(r,c)$-th position, separately exchangeable data with $K=2$ indices are generated according to
$
\bX_{\i}
=
\frac{1}{4} \left( \bm{Z}_{(i_1,0)} + \bm{Z}_{(0,i_2)} \right) + \frac{1}{2} \bm{Z}_{(i_1,i_2)},
$
where 
$
\bm{Z}_{\i \odot \e} \sim B N(\bm{0}, \Sigma_{\bm{Z}}) + (1-B)N(\bm{0}, 2\Sigma_{\bm{Z}})
$
and
$
B \sim \text{Bernoulli}(0.5)
$
independently for $\i \in \{(i_1,i_2) \in \N^2: 1 \le i_1 \le N_1, 1 \le i_2 \le N_2\}$ and $\e \in \{0,1\}^2$.
For this data generating design, we run 2,500 Monte Carlo iterations to compute the uniform coverage frequencies of $\E[\bX_{\i}]$ for the nominal probabilities of 90\% and 95\% using our proposed multiplier bootstrap for separately exchangeable arrays with 2,500 bootstrap iterations.

We next describe the simulation design for jointly exchangeable arrays. 
We shall focus on the the most common case in practice, the dyadic data, i.e. $K=2$.
With $\Sigma_{\bm{Z}}$ denoting the $p \times p$ covariance matrix consisting of elements of the form $4^{-\abs{r-c}}$ in its $(r,c)$-th position, dyadic samples are generated 
symmetrically in $i$ and $j$
according to
$
\bX_{i,j}
=
\frac{1}{4} \left( \bm{Z}_{(i,0)} + \bm{Z}_{(j,0)} \right) + \frac{1}{2} \bm{Z}_{(i,j)},
$
where 
$
\bm{Z}_{\i \odot \e} \sim B N(\bm{0}, \Sigma_{\bm{Z}}) + (1-B) N(\bm{0}, 2\Sigma_{\bm{Z}})
$
and
$
B \sim \text{Bernoulli}(0.5)
$
 %independently for $\i \in \{(i,j) \in \N^2: 1 \le i,j \le n, i \ne j \}$ and $\e \in \{1\} \times \{0,1\}$.
independently for $\i \in \{(i,j) \in \N^2: 1 \le i < j \le n \}$ and $\e \in \{1\} \times \{0,1\}$.
We run 2,500 Monte Carlo iterations to compute the uniform coverage frequencies of $\bS_{n}$ for the nominal probabilities of 90\% and 95\% using our proposed multiplier bootstrap with 2,500 bootstrap iterations.

Table \ref{tab:two_way} summarizes simulation results under the separate exchangeability.
The columns consist of the dimension $p$ of $\bX$ and the two-way sample size $(N_1,N_2)$.
The displayed numbers indicate the simulated uniform coverage frequencies for the nominal probabilities of 90\% and 95\%.
For each dimension $p \in \{25,50,100\}$, sample sizes vary as $(N_1,N_2) \in \{(25,25), (50,50),(100,100)\}$.
Table \ref{tab:polyadic} summarizes simulation results under the joint exchangeability.
The columns consist of the dimension $p$ of $\bX$, and the dyadic sample size $N$.
The displayed numbers indicate the simulated uniform coverage frequencies for the nominal probabilities of 90\% and 95\%.
For each dimension $p \in \{25,50,100\}$, sample sizes vary as $n \in \{50, 100, 200\}$.

Observe that, for each simulation design and for each nominal probability, the uniform coverage frequencies approach the nominal probability as the sample size increases.
These results support the theoretical property of our multiplier bootstrap method.
We ran many other sets of simulations with various designs and sample sizes not presented here, but this observed pattern to support our theory remains invariant across all the different sets of simulations -- see Appendix \ref{sec:simulation_gaussian}.
In Appendix \ref{sec:simulation_three_way}, we further experiment with the separate exchangeability with $K=3$ indices.

\begin{table}
	\centering
	  \scalebox{0.85}{
		\begin{tabular}{|r|ccc|ccc|ccc|}
		\hline
			Normalization & \multicolumn{9}{|c|}{No} \\
		\hline
			Dimension of $\bX_{\i}$: $p$ & 25  & 25  & 25  & 50  & 50  & 50  & 100 & 100 & 100\\
		\hline
			Sample Sizes: $N_1,N_2$ & 25  & 50  & 100 & 25  & 50  & 100 & 25  & 50  & 100\\
		\hline
			90\% Coverage      &0.927&0.908&0.905&0.942&0.931&0.919&0.943&0.910&0.917\\
			95\% Coverage      &0.967&0.954&0.956&0.976&0.968&0.960&0.973&0.957&0.962\\
		\hline
			Normalization & \multicolumn{9}{|c|}{Yes} \\
		\hline
			Dimension of $\bX_{\i}$: $p$ & 25  & 25  & 25  & 50  & 50  & 50  & 100 & 100 & 100\\
		\hline
			Sample Sizes: $N_1,N_2$ & 25  & 50  & 100 & 25  & 50  & 100 & 25  & 50  & 100\\
		\hline
			90\% Coverage      &0.884&0.892&0.905&0.885&0.885&0.900&0.857&0.878&0.901\\
			95\% Coverage      &0.936&0.938&0.949&0.930&0.938&0.942&0.921&0.936&0.952\\
		\hline
		\end{tabular}
		}
		\medskip
	\caption{Simulation results for separately exchangeable data with $K=2$ indices. Displayed are the dimension $p$ of $\bX$, the two-way sample size $(N_1,N_2)$ with $N_1=N_2$, and the simulated uniform  coverage frequencies for the nominal probabilities of 90\% and 95\%.}
	\label{tab:two_way}
\end{table}

\begin{table}
	\centering
		\scalebox{0.85}{
		\begin{tabular}{|r|ccc|ccc|ccc|}
		\hline
			Normalization & \multicolumn{9}{|c|}{No} \\
		\hline
			Dimension of $\bX_{i,j}$: $p$ & 25  & 25  & 25  & 50  & 50  & 50  & 100 & 100 & 100\\
		\hline
			Sample Size: $n$ & 50  & 100 & 200 & 50  & 100 & 200 & 50  & 100 & 200\\
		\hline
			90\% Coverage      &0.902&0.896&0.891&0.912&0.914&0.908&0.904&0.915&0.893\\
			95\% Coverage      &0.960&0.953&0.945&0.956&0.963&0.951&0.953&0.961&0.952\\
		\hline
			Normalization & \multicolumn{9}{|c|}{Yes} \\
		\hline
			Dimension of $\bX_{i,j}$: $p$ & 25  & 25  & 25  & 50  & 50  & 50  & 100 & 100 & 100\\
		\hline
			Sample Size: $n$ & 50  & 100 & 200 & 50  & 100 & 200 & 50  & 100 & 200\\
		\hline
			90\% Coverage      &0.851&0.854&0.887&0.819&0.865&0.884&0.802&0.870&0.864\\
			95\% Coverage      &0.921&0.924&0.938&0.890&0.936&0.943&0.882&0.927&0.925\\
		\hline
		\end{tabular}
		}
		\medskip
	\caption{Simulation results for dyadic data. Displayed are the dimension $p$ of $\bX$, the dyadic sample size $n$, and the simulated uniform  coverage frequencies for the nominal probabilities of 90\% and 95\%.}
	\label{tab:polyadic}
\end{table}

%%%%%%%%%%%%%%%%%%%%%%%%%%%%%%%%%%%%%%%%%%%%%%%%%%%%%%%%%%%%%%%%%%%%%
\section{Real data analysis}\label{sec:empirical}
%%%%%%%%%%%%%%%%%%%%%%%%%%%%%%%%%%%%%%%%%%%%%%%%%%%%%%%%%%%%%%%%%%%%%

In this section, we present an empirical application of the method proposed in Section \ref{sec:density} to constructing uniform confidence bands for the density functions of bilateral trade volumes in the international trade, with a similar motivation to that stated in \citet{Graham2019kernel,GrahamNiuPowell2020kernel}.
Recall that our method extends those by \citet{Graham2019kernel} in that we can draw uniform confidence bands as opposed to point-wise confidence intervals.
From this analysis, we can learn about the evolution of the distributions of international trade volumes over time.

We employ the international trade data used in \citet{head2014gravity}, that come from the Direction of Trade Statistics (DoTS).
This data set contains information about bilateral trade flows among 208 economies for 59 years from 1948 to 2006.
In this analysis, we will focus on the relatively recent years, 1990, 1995, 2000 and 2005.
Our measure of the bilateral trade volume $Y_{ij}$ is defined as the logarithm of the sum of the trade flow from economy $i$ to economy $j$ and the trade flow from economy $j$ to economy $i$.
We perform simulation studies on confidence bands for densities in Appendix \ref{sec:simulation_density}, confirm that the method works as desired, and thus use the same software code here to draw confidence bands of the probability density function of $Y_{ij}$.
Since there is a probability mass at zero in the international trade volumes, what we estimate is precisely the Lebesgue-Radon-Nikodym derivative of the continuous part of the distribution, rather than the probability density function.
Specifically, we use $\hat b(y)$ defined in Section \ref{sec:density} for estimation, and confidence bands are constructed by setting $\hat a = 1$.
That said, we shall call it a density for conciseness.

Figure \ref{fig:trade_1990} illustrates estimates and confidence bands of the density functions of $Y_{ij}$ in each of the years 1990, 1995, 2000 and 2005.
Each panel of the figure displays the kernel density estimates in a solid curve and the 95\% uniform confidence bands in a gray shade.
In addition, we also display the proportion of zero bilateral trade volumes to the left of the kernel density plots so we can get an idea of the complementary proportion that consists the density of the continuously distributed part of the distribution.
Although we treat $Y_{ij}$ as the logarithm of the bilateral trade volumes in estimation and inference, we use the original scale on the horizontal axis for ease of reading the graphs.

\begin{figure}[t]
	\centering
	\begin{tabular}{cc}
		\includegraphics[width=0.5\textwidth]{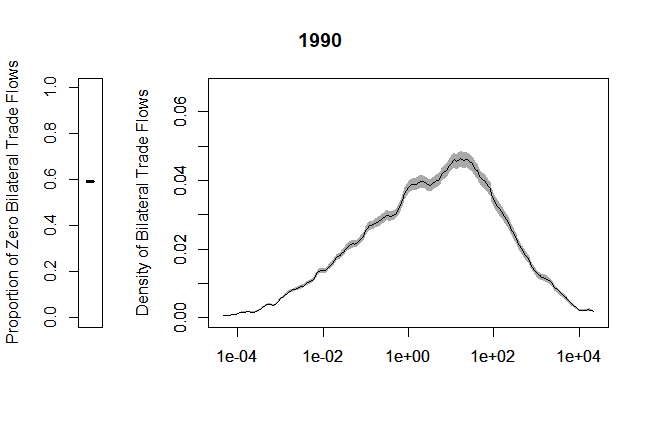}&
		\includegraphics[width=0.5\textwidth]{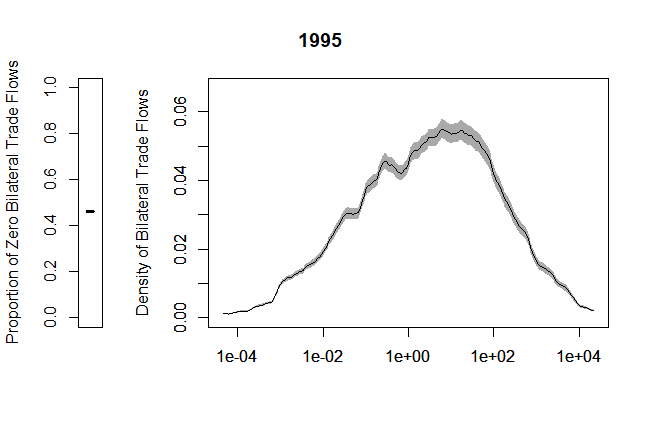}\\
		\includegraphics[width=0.5\textwidth]{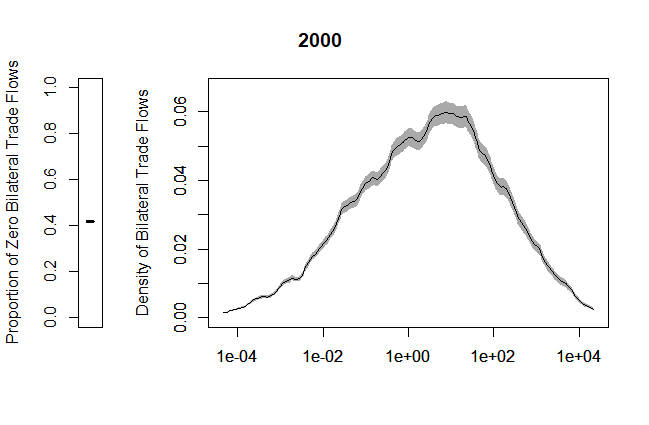}&
		\includegraphics[width=0.5\textwidth]{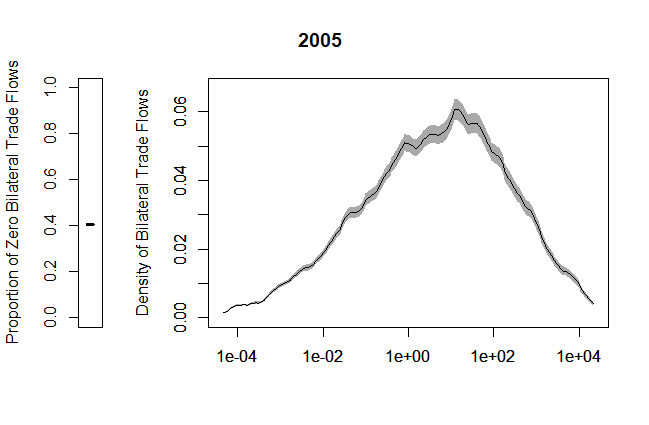}
	\end{tabular}
	\caption{The kernel density estimates (solid curve) and the 95\% uniform confidence bands (gray shade) of the bilateral trade volumes in 1990, 1995, 2000, and 2005.}
	\label{fig:trade_1990}
\end{figure}

Observe that the proportion of the zero trade volume is decreasing over time, and the density function is accordingly moving upward over time.
Despite this pattern of the changes over time, the shapes of the density functions are rather similar across time in the middle of the distribution.
This observation entails a high level of confidence given the reasonably tight confidence bands.
On the other hand, notice that the right tail of the distribution becomes fatter as time progresses, implying that there is an increasing number of bilateral trading pairs with very large trade volumes.

%%%%%%%%%%%%%%%%%%%%%%%%%%%%%%%%%%%%%%%%%%%%%%%%%%%%%%%%%%%%%%%%%%%%%
\section{Summary}\label{sec:summary}
%%%%%%%%%%%%%%%%%%%%%%%%%%%%%%%%%%%%%%%%%%%%%%%%%%%%%%%%%%%%%%%%%%%%%

In this paper, we have developed methods and theories for inference about  high-dimensional parameters with separately/jointly exchangeable arrays. 
Building on the high-dimensional CLTs over the rectangles, we have proposed bootstrap methods and established their finite sample validity for both notions of exchangeability. 
Simulation studies support the theoretical properties of the methods.
We have illustrated a couple of applications of the bootstrap methods.
First, extending \citet{Graham2019kernel}, we have applied our method to construction of uniform confidence bands for  density functions of dyadic data.
Second, we have demonstrated an application of our method to penalty choice for $\ell_1$-penalized regression under the separate exchangeability.
As such, the results in the present paper  pave the way for a variety of applications to analyses of separately and jointly exchangeable arrays.

\clearpage

\appendix
%%%%%%%%%%%%%%%%%%%%%%%%%%%%%%%%%%%%%%%%%%%%%%%%%%%%%%%%%%%%%%%%%%%%%
\section*{Appendix}

\section{Additional results}
\label{sec:additional results}

\subsection{Additional results for Section \ref{sec:multiway_clustered_data}}
\label{sec: normalized sample mean SE}

In practice, we often normalize the coordinates of the sample mean by estimates of the standard deviations, so that each coordinate is approximately distributed as $N(0,1)$. In view of the high-dimensional CLT, the approximate variance of the $j$-th coordinate of $\sqrt{n}\bS_{\bN}$ is given by $\sigma_j^2=\Sigma_{j,j}$, where $\Sigma_{j,j}$ is the $(j,j)$-th component of $\Sigma$. This can be estimated by 
\begin{align*}
\hat \sigma_j^2=  \sum_{k=1}^K\frac{n}{N_k^2} \sum_{i_k=1}^{N_k} ( \overline X_{k,i_k}^j 
- S_{\bN}^j )^2.
\end{align*}
Let $\Lambda = \diag \{ \sigma_1^2,\dots,\sigma_p^2 \}$ and $\hat{\Lambda}= \diag\{\hat\sigma_1^2,\dots,\hat\sigma_p^2\}$. We consider to approximate the distribution of $\sqrt{n}\hat{\Lambda}^{-1/2}\bS_{\bN}$ by $\sqrt{n}\hat{\Lambda}^{-1/2}\bS_{\bN}^{MB}$.

\begin{corollary}\label{cor:normalized_multiway_MB}
Consider Cases (i) and (ii) in Theorem \ref{thm:bootsrap_validity}.
In Case (i), assume further that 
\[
\frac{D_{\bN}^2 \log^7(p\overline{N})}{n} \le C_1 n^{-\zeta},
\]
while in Case (ii) assume further that 
\[
\frac{D_{\bN}^2 \log^7(p\overline{N})}{n} \bigvee \left ( \frac{D_{\bN}^2 \log^3 (p\overline{N})}{n^{1-2/q}}\right )^2 \le C_1 n^{-\zeta}.
\]
Then, under Case (i), there exists a constant $C$ depending only on $\underline{\sigma}, K$, and $C_1$ such that for $\bm{Y}\sim N(\bm{0},\Sigma)$,
\begin{align*}
&\sup_{R\in \calR}\left|\Prob(\sqrt{n} \hat{\Lambda}^{-1/2}\bS_{\bN}\in R)-\Prob( \Lambda^{-1/2} \bm{Y}\in R)\right|\le C n^{-\zeta/6} \quad \text{and} \\
&\Prob \left \{ \sup_{R\in \calR}\left|\Prob_{|\bX_{[\bN]}}(\sqrt{n} \hat{\Lambda}^{-1/2}\bS_{\bN}^{MB}\in R)-\Prob( \Lambda^{-1/2} \bm{Y}\in R)\right|\le C n^{-\zeta/6} \right \}  \ge 1-Cn^{-1}. 
\end{align*}
Under Case (ii), the same conclusion holds with $n^{-\zeta/6}$ replaced by $\max \{ n^{-\zeta/6}, n^{-(\zeta-2/q)/2} \sqrt{\log n} \}$, while the constant $C$ depends only on $q,\underline{\sigma}, K$, and $C_1$. 
\end{corollary}

\noindent
The proof of this corollary is deferred to Appendix \ref{sec:proofs_multiway_clustered_data}. 

We may alternatively use Bessel's correction
\begin{align*}
\tilde \sigma_j^2=  \sum_{k=1}^K\frac{n}{N_k (N_k-1)} \sum_{i_k=1}^{N_k} ( \overline X_{k,i_k}^j 
- S_{\bN}^j )^2.
\end{align*}
to improve finite sample performances.
We employ this finite sample adjustment in our numerical examples.
Specifically, for $\alpha \in (0,1)$, let $\widehat{cv}^B(\alpha)$ denote the $\lfloor \alpha B \rfloor$-th order statistic (i.e., the approximate $1-\alpha$ quantile) of 
$
\|\sqrt{n}\hat{\Lambda}^{-1/2}\bS_{\bN}^{MB}\|_\infty
$
in $B$ replications of the multiplier bootstrap draws.
We obtain the $1-\alpha$ level uniform confidence band by
\begin{align*}
\prod_{j=1}^p \left [ S_{\bN}^j \pm \frac{\widehat{cv}^B(\alpha) \tilde \sigma_{j}}{\sqrt{n}} \right ].
\end{align*}

\subsection{Additional results for Section \ref{sec:polyadic_data}}
\label{sec: normalized sample mean JE}

We consider normalized sample means for jointly exchangeable arrays. In light of the high-dimensional CLT for jointly exchangeable arrays, the approximate variance of the $\ell$-th coordinate of $\sqrt{n}\bS_{n}$ is given by
$\sigma^2_\ell=\text{Var}(W_{1}^{\ell})$, which can be estimated by
\begin{align*}
\hat \sigma^2_\ell= \frac{1}{n}\sumk ( \hat W_{k}^\ell - KS_n^\ell  )^2.
\end{align*}
Let $\Lambda=\diag\{\sigma_1^2,\dots,\sigma_p^2\}$ and $\hat{\Lambda}= \diag\{\hat\sigma_1^2,\dots,\hat\sigma_p^2\}$. We consider to approximate the distribution of $\sqrt{n}\hat{\Lambda}^{-1/2}\bS_{n}$ by $\sqrt{n}\hat{\Lambda}^{-1/2}\bS_{n}^{MB}$.

\begin{corollary}\label{cor:normalized_polyadic_MB}
Consider Cases (i) and (ii) in Theorem \ref{thm:bootstrap validity polyadic}.
In Case (i), assume further that 
\[
\frac{D_{n}^2 \log^7(pn)}{n} \le C_1 n^{-\zeta},
\]
while in Case (ii) assume further that 
\[
\frac{D_{n}^2 \log^7(pn)}{n} \bigvee \left ( \frac{D_{n}^2 \log^3 (pn)}{n^{1-2/q}}\right )^2 \le C_1 n^{-\zeta}.
\]
Then, under Case (i),  there exists a constant $C$ depending only on $\underline{\sigma}, K$, and $C_1$ such that for $\bm{Y} \sim N(\bm{0},\Sigma)$,
\begin{align*}
&\sup_{R\in \calR}\left|\Prob(\sqrt{n} \hat{\Lambda}^{-1/2}\bS_{n}\in R)-\Prob( \Lambda^{-1/2} \bm{Y}\in R)\right|\le C n^{-\zeta/6} \quad \text{and} \\
&\Prob \left \{ \sup_{R\in \calR}\left|\Prob_{|\bX_{I_{n,K}}}(\sqrt{n} \hat{\Lambda}^{-1/2}\bS_{n}^{MB}\in R)-\Prob( \Lambda^{-1/2} \bm{Y}\in R)\right|\le C  n^{-\zeta/6} \right \} \ge 1-Cn^{-1}.
\end{align*}
Under Case (ii), the same conclusion holds with $n^{-\zeta/6}$ replaced by $\max \{ n^{-\zeta/6}, n^{-(\zeta-2/q)/2} \sqrt{\log n} \}$, while the constant $C$ depends only on $q,\underline{\sigma}, K$, and $C_1$. 
\end{corollary}

\noindent
The proof is analogous to Corollary \ref{cor:normalized_multiway_MB} and thus omitted.

We may alternatively use  Bessel's correction
\begin{align*}
\tilde \sigma^2_\ell= \frac{1}{n-1}\sumk ( \hat W_{k}^\ell - KS_n^\ell  )^2.
\end{align*}
to improve finite sample performances.
%Specifically, let $\tilde{\Lambda}= \diag\{\tilde\sigma_1^2,\dots,\tilde\sigma_p^2\}$, and consider $\sqrt{n} \tilde{\Lambda}^{-1/2}\bS_{n}$ in place of $\sqrt{n} \hat{\Lambda}^{-1/2}\bS_{n}$.
%We employ this finite sample adjustment in our simulation studies.
We employ this finite sample adjustment in our numerical examples.
Specifically, for $\alpha \in (0,1)$, let $\widehat{cv}^B(\alpha)$ denote the $\lfloor \alpha B \rfloor$-th order statistic (i.e., the approximate $1-\alpha$ quantile) of 
$
\| \sqrt{n}\hat{\Lambda}^{-1/2}\bS_{n}^{MB}\|_\infty
$
in $B$ replications of the multiplier bootstrap draws.
We obtain the $1-\alpha$ level uniform confidence band by
\begin{align*}
\prod_{j=1}^p \left [S^j_{n} \pm \frac{\widehat{cv}^B(\alpha) \tilde \sigma_j}{\sqrt{n}} \right].
\end{align*}

%%%%%%%%%%%%%%%%%%%%%%%%%%%%%%%%%%%%%%%%%%%%%%%%%%%%%%%%%%%%%%%%%%%%%
\section{Maximal inequalities for separately exchangeable arrays}\label{sec:maximal_inequality_for_multiway_data} 
%%%%%%%%%%%%%%%%%%%%%%%%%%%%%%%%%%%%%%%%%%%%%%%%%%%%%%%%%%%%%%%%%%%%%
In this section, we shall develop maximal inequalities for separately exchangeable arrays. 
As in Section \ref{sec:multiway_clustered_data}, let  $(\bX_{\i})_{\i \in \N^{K}}$ be a $K$-array consisting of random vectors in $\R^{p}$ with mean zero generated by the structure (\ref{eq: Aldous-Hoover_SE}), i.e., $\bX_{\i} = \mathfrak{g}((U_{\i \odot \e})_{\e \in \{ 0,1 \}^{K} \setminus \{ \bm{0} \}})$ for $\i \in \N^K$. We will follow the notations used in Section \ref{sec:multiway_clustered_data}.  
The following theorem is fundamental. 

\begin{theorem}
\label{thm: maximal inequality}
Pick any $1 \le k \le K$ and $\e \in \mathcal{E}_{k}$. Then, for any $q \in [1,\infty)$, we have
\[
\left ( \E \left [ \left \| \sum_{\i \in I_{\e}([\bN])} \hat{\bX}_{\i}  \right \|_{\infty}^{q} \right ] \right)^{1/q} \le C (\log p)^{k/2} \left ( \E \left [ \max_{1 \le j \le p} \left (\sum_{\i \in I_{\e}([\bN])} |\hat{X}_{\i}^{j}|^{2} \right)^{q/2}  \right ] \right)^{1/q},
\]
where $C$ is a constant that depends only on $q$ and $K$. 
\end{theorem}

The following corollary is immediate from Jensen's inequality. 

\begin{corollary}[Global maximal inequality]
\label{cor: maximal inequality}
For any $1 \le k \le K, \e \in \mathcal{E}_{k}$, and $q \in [1,\infty)$, we have
\begin{equation}
\left (\E \left [ \left \| \sum_{\i \in I_{\e}([\bN])} \hat{\bX}_{\i}  \right \|_{\infty}^{q} \right ] \right)^{1/q} \le C (\log p)^{k/2} \sqrt{\prod_{k' \in \supp (\e)} N_{k'}} (\E[ \| \hat{\bX}_{\bm{1}\odot \e} \|_{\infty}^{q \vee 2}])^{1/(q \vee 2)},
\label{eq: local max}
\end{equation}
where $C$ is a constant that depends only on $q$ and $K$. 
\end{corollary}
\red{
\begin{proof}[Proof of Corollary \ref{cor: maximal inequality}]
We begin with observing that
\[
\E \left [ \max_{1 \le j \le p} \left (\sum_{\i \in I_{\e}([\bN])} |\hat{X}_{\i}^{j}|^{2} \right)^{q/2}  \right ] \le \E \left [  \left (\sum_{\i \in I_{\e}([\bN])} \|\hat{\bX}_{\i}\|_{\infty}^{2} \right)^{q/2}  \right ]
\]
If $q \le 2$, then by Jensen's inequality, the right-hand side is bounded by 
\[
 \left (  \E \left [ \sum_{\i \in I_{\e}([\bN])} \|\hat{\bX}_{\i}\|_{\infty}^{2}  \right ]\right)^{q/2} = | I_{\e}([\bN])|^{q/2} (\E[ \| \hat{\bX}_{\bm{1}\odot \e} \|_{\infty}^{2}])^{q/2}.
\]
If $q > 2$, then by Jensen's inequality,
\[
\begin{split}
\left (\sum_{\i \in I_{\e}([\bN])} \|\hat{\bX}_{\i}\|_{\infty}^{2} \right)^{q/2} &= |I_{\e}([\bN])|^{q/2} \left ( \frac{1}{|I_{\e}([\bN])|} \sum_{\i \in I_{\e}([\bN])} \|\hat{\bX}_{\i}\|_{\infty}^{2} \right)^{q/2} \\
&\le |I_{\e}([\bN])|^{q/2} \times \frac{1}{|I_{\e}([\bN])|} \sum_{\i \in I_{\e}([\bN])} \|\hat{\bX}_{\i}\|_{\infty}^{q}.
\end{split}
\]
The expectation of the right-hand side is $ | I_{\e}([\bN])|^{q/2}\E[ \|\hat{\bX}_{\bm{1} \odot \e}\|_{\infty}^{q} ]$. 
\end{proof}
}

\begin{remark}
By Jensen's inequality, $\E[ \| \hat{\bX}_{\bm{1}\odot \e} \|_{\infty}^{q \vee 2}]$ on the right-hand side of (\ref{eq: local max}) can be replaced by $\E[\| \bX_{\bm{1}} \|_{\infty}^{q \vee 2}]$ by adjusting the constant $C$. 
\end{remark}

The proof of Theorem \ref{thm: maximal inequality} relies on the following symmetrization inequality. Recall that a Rademacher random variable is a random variable taking $\pm 1$ with equal probability.

\begin{lemma}[Symmetrization]
\label{lem: symmetrization}
Pick any $1 \le k \le K$.
Let $\{ \epsilon_{1,i_{1}} \},\dots, \{ \epsilon_{k,i_{k}} \}$ be independent Rademacher random variables independent of the $U$-variables.
Then, for any nondecreasing convex function $\Phi:[0,\infty) \to [0,\infty)$, we have 
\[
\E \left [ \Phi \left (\left \| \sum_{i_{1},\dots,i_{k}} \hat{\bX}_{(i_{1},\dots,i_{k},0,\dots,0)} \right \|_{\infty}\right) \right ] \le \E \left [ \Phi \left ( 2^{k} \left \| \sum_{i_{1},\dots,i_{k}} \epsilon_{1,i_{1}} \cdots \epsilon_{k,i_{k}} \hat{\bX}_{(i_{1},\dots,i_{k},0,\dots,0)} \right \|_{\infty} \right)\right ].
\]
\end{lemma}

The proof of Lemma \ref{lem: symmetrization} in turn  relies on the following result.

\begin{lemma}
\label{lem: conditional zero mean}
Let $\i \in \N^{K}$. Pick any $1 \le k \le K$ and let $\e \in \mathcal{E}_{k}$. 
Then, for any $\ell \in \supp (\e)$, conditionally on $(U_{\i \odot \e'})_{ \e' \le \e - \e_{\ell}}$, the vector $\hat{\bX}_{\i \odot \e}$ has mean zero. 
\end{lemma}

\begin{proof}[Proof of Lemma \ref{lem: conditional zero mean}]
For illustration, consider first the $K=3$ case and $\e = (1,1,1)$. Then 
\[
\hat{\bX}_{\i} = \bX_{\i} - \hat{\bX}_{(i_{1},i_{2},0)} - \hat{\bX}_{(0,i_{2},i_{3})} - \hat{\bX}_{(i_{1},0,i_{3})} - \hat{\bX}_{(i_{1},0,0)}  - \hat{\bX}_{(0,i_{2},0)} - \hat{\bX}_{(0,0,i_{3})}. 
\]
Given $(U_{(i_{1},0,0)}, U_{(0,i_{2},0)}, U_{(i_{1},i_{2},0)})$, we have 
\[
\begin{split}
\E[ \hat{\bX}_{(0,i_{2},i_{3})}  \mid U_{(i_{1},0,0)}, U_{(0,i_{2},0)}, U_{(i_{1},i_{2},0)}] &= \E[\bX_{\i} \mid U_{(0,i_{2},0)}] - \E[\bX_{\i} \mid U_{(0,i_{2},0)}] = 0 ,\\
\E[ \hat{\bX}_{(i_{1},0,i_{3})}  \mid U_{(i_{1},0,0)}, U_{(0,i_{2},0)}, U_{(i_{1},i_{2},0)}] &=\E[\bX_{\i} \mid U_{(i_{1},0,0)}] - \E[\bX_{\i} \mid U_{(i_{1},0,0)}] = 0.
\end{split}
\]
Conclude that 
\[
\begin{split}
\E[\hat{\bX}_{\i} \mid U_{(i_{1},0,0)}, U_{(0,i_{2},0)}, U_{(i_{1},i_{2},0)}] &= \E[\bX_{\i} \mid U_{(i_{1},0,0)}, U_{(0,i_{2},0)}, U_{(i_{1},i_{2},0)}] \\
&\quad - (\hat{\bX}_{(i_{1},i_{2},0)} + \hat{\bX}_{(i_{1},0,0)}  + \hat{\bX}_{(0,i_{2},0)}) \\
&=0.
\end{split}
\]

The proof for the general case is by induction on $k$. The conclusion is trivial when $k=1$. Suppose that the lemma is true up to $k-1$. 
Then,
\[
\begin{split}
&\E[\hat{\bX}_{\i \odot \e} \mid (U_{\i \odot \e'})_{ \e' \le \e - \e_{\ell}}] \\
&\quad =\E[ \bX_{\i} \mid (U_{\i \odot \e'})_{ \e' \le \e - \e_{\ell}}] - 
\hat{\bX}_{\i \odot (\e - \e_{\ell})}  \\
&\qquad - \sum_{\substack{ \e' \le \e \\ \e' \neq \e, \e - \e_{\ell}}} \E[\hat{\bX}_{\i\odot \e'} \mid (U_{\i \odot \e''})_{ \e'' \le \e - \e_{\ell}}]   \quad (\text{by the definition of $\hat{\bX}_{\i \odot \e}$})\\
&\quad =\sum_{\substack{\e' \le \e-\e_{\ell} \\ \e' \ne \e - \e_{\ell}}} \hat{\bX}_{\i \odot \e'}  - \sum_{\substack{ \e' \le \e \\ \e' \neq \e, \e - \e_{\ell}}} \E[\hat{\bX}_{\i\odot \e'} \mid (U_{\i \odot \e''})_{ \e'' \le \e - \e_{\ell}}] \quad (\text{by plugging in the expansion of $\hat{\bX}_{\i \odot (\e - \e_{\ell})}$}) \\
&\quad =\sum_{\substack{\e' \le \e-\e_{\ell} \\ \e' \ne \e - \e_{\ell}}} \E[\hat{\bX}_{\i \odot \e'}\mid(U_{\i \odot \e''})_{ \e'' \le \e - \e_{\ell}} ] - \sum_{\substack{ \e' \le \e \\ \e' \neq \e, \e - \e_{\ell}}} \E[\hat{\bX}_{\i\odot \e'} \mid (U_{\i \odot \e''})_{ \e'' \le \e - \e_{\ell}}]\\
&\quad = - \sum_{\substack{ \e' \le \e - \e_{\ell'}, \ell' \ne \ell \\ \ell \in \supp (\e'), \ell' \in \supp (\e)}} \E\left [ \hat{\bX}_{\i\odot \e'} \mid (U_{\i \odot \e''})_{ \e'' \le \e-\e_{\ell}} \right ].
\end{split}
\]
Here, we have used the fact that $\hat{\bX}_{\i \odot \e'}$ is $\sigma ((U_{\i \odot \e''})_{\e'' \le \e'})$-measurable, so that $\E[\hat{\bX}_{\i\odot \e'} \mid (U_{\i \odot \e''})_{ \e'' \le \e - \e_{\ell}}] = \hat{\bX}_{\i \odot \e'}$ as long as $\supp (\e') \subset \supp (\e - \e_{\ell})$. 
For any $ \e' \le \e - \e_{\ell'}$ with $\ell' \neq \ell, \ell \in \supp (\e')$, and $\ell' \in \supp(\e)$,  we have
\[
\E\left [ \hat{\bX}_{\i\odot \e'} \mid (U_{\i \odot \e''})_{ \e'' \le \e-\e_{\ell}} \right ] = \E\left [ \hat{\bX}_{\i\odot \e'} \mid (U_{\i \odot \e''})_{ \e'' \le \e'-\e_{\ell}} \right ] = 0
\]
by the induction hypothesis. 
Conclude that $\E[\hat{\bX}_{\i \odot \e} \mid (U_{\i \odot \e'})_{ \e' \le \e - \e_{\ell}}] = 0$.
\end{proof}

\begin{proof}[Proof of Lemma \ref{lem: symmetrization}]
Let $\e = (\underbrace{1,\dots,1}_{k},0,\dots,0)$.
Given $( U_{\i \odot \e'} )_{\i \in [\bN], \e' \le \e - \e_{1}}$, $\{ \sum_{i_{2},\dots,i_{k}} \hat{\bX}_{(i_{1},i_{2}\dots,i_{k},0,\dots,0)}  : i_{1} = 1,\dots, N_{1} \}$ are independent with mean zero (the latter follows from Lemma \ref{lem: conditional zero mean}). 
Hence, applying the symmetrization inequality (\cite{vdVW1996}, Lemma 2.3.6) conditionally on $( U_{\i \odot \e'})_{\i \in [\bN],\e' \le \e - \e_{1}}$, we have 
\[
\begin{split}
&\E \left [ \Phi \left ( \left \| \sum_{i_{1},\dots,i_{k}} \hat{\bX}_{(i_{1},\dots,i_{k},0,\dots,0)} \right \|_{\infty} \right ) \mid ( U_{\i \odot \e'} )_{\i \in [\bN],\e' \le \e - \e_{1}} \right ] \\
&= \E \left [ \Phi \left ( \left \| \sum_{i_{1}} \left (\sum_{i_{2}, \dots,i_{k}}\hat{\bX}_{(i_{1},\dots,i_{k},0,\dots,0)} \right ) \right \|_{\infty}  \right )\mid ( U_{\i \odot \e'} )_{\i \in [\bN],\e' \le \e - \e_{1}} \right ] \\
&\le \E \left [ \Phi \left ( 2\left \| \sum_{i_{1}} \epsilon_{1,i_{1}} \left (\sum_{i_{2}, \dots,i_{k}}\hat{\bX}_{(i_{1},\dots,i_{k},0,\dots,0)} \right ) \right \|_{\infty}\right )  \mid ( U_{\i \odot \e'} )_{\i \in [\bN],\e' \le \e - \e_{1}} \right ] \\
&=\E \left [ \Phi \left(2 \left \| \sum_{i_{1}, \dots,i_{k}} \epsilon_{1,i_{1}}\hat{\bX}_{(i_{1},\dots,i_{k},0,\dots,0)}  \right \|_{\infty} \right ) \mid ( U_{\i \odot \e'} )_{\i \in [\bN],\e' \le \e - \e_{1}} \right ] 
\end{split}
\]
By Fubini's theorem, we have 
\[
\E \left [ \Phi \left ( \left \| \sum_{i_{1},\dots,i_{k}} \hat{\bX}_{(i_{1},\dots,i_{k},0,\dots,0)} \right \|_{\infty}\right ) \right ] \le 
\E \left [\Phi \left ( 2\left \| \sum_{i_{1},\dots,i_{k}} \epsilon_{1,i_{1}}\hat{\bX}_{(i_{1},\dots,i_{k},0,\dots,0)}  \right \|_{\infty} \right) \right ].
\]
Next, given $\{ \epsilon_{1,i_{1}} \} \cup \{ U_{\i \odot \e'} \}_{\i \in [\bN],\e' \le \e-\e_{2}}$, $\{ \sum_{i_{1},i_{3},\dots,i_{K}} \epsilon_{1,i_{1}} \hat{\bX}_{(i_{1},i_{2}\dots,i_{K},0,\dots,0)}  : i_{2}=1,\dots, N_{2} \}$ are independent with mean zero, so that by the symmetrization inequality and Fubini's theorem, we have
\[
\begin{split}
&\E \left [\Phi \left ( 2\left \| \sum_{i_{1},\dots,i_{k}} \epsilon_{1,i_{1}}\hat{\bX}_{(i_{1},\dots,i_{k},0,\dots,0)}  \right \|_{\infty} \right ) \right ] \\
&=\E \left [\Phi \left( 2 \left \| \sum_{i_{2}} \left (\sum_{i_{1},i_{3} \dots,i_{k}} \epsilon_{1,i_{1}}\hat{\bX}_{(i_{1},\dots,i_{k},0,\dots,0)}  \right )  \right \|_{\infty}\right) \right ] \\
&\le \E \left [\Phi \left(4 \left \| \sum_{i_{2}} \epsilon_{2,i_{2}}  \left (\sum_{i_{1},i_{3} \dots,i_{k}} \epsilon_{1,i_{1}}\hat{\bX}_{(i_{1},\dots,i_{k},0,\dots,0)}  \right )  \right \|_{\infty}\right )  \right ] \\
&=\E \left [\Phi \left(4 \left \| \sum_{i_{1},\dots,i_{k}} \epsilon_{1,i_{1}} \epsilon_{2,i_{2}} \hat{\bX}_{(i_{1},\dots,i_{k},0,\dots,0)}  \right \|_{\infty} \right ) \right ].
\end{split}
\]
The conclusion of the lemma follows from repeating this procedure. 
\end{proof}

We are now in position to prove Theorem \ref{thm: maximal inequality}.

\begin{proof}[Proof of Theorem \ref{thm: maximal inequality}]
In this proof, the notation $\lesssim$ means that the left-hand side is less than the right-hand side up to a constant that depends only on $q$ and $K$.
We may assume without loss of generality  $\e = (\underbrace{1,\dots,1}_{k},0,\dots,0)$. 
In view of Lemma \ref{lem: symmetrization}, it suffices to show that 
\[
\E \left [ \left \| \sum_{i_{1},\dots,i_{k}} \epsilon_{1,i_{1}} \cdots \epsilon_{k,i_{k}} \hat{\bX}_{(i_{1},\dots,i_{k},0,\dots,0)} \right \|_{\infty}^{q} \right ] \lesssim (\log p)^{qk/2} \E \left [ \max_{1 \le j \le p} \left (\sum_{i_{1},\dots,i_{k}} |\hat{X}_{(i_{1},\dots,i_{k},0,\dots,0)}^{j}|^{2}\right)^{q/2} \right ]. 
\]
By conditioning and Lemma 2.2.2 in \cite{vdVW1996}, together with the fact that that the $L^{q}$-norm is bounded from above by the $\psi_{2/k}$-norm up to some constant that depends only on $(q,k)$ \red{(cf. Lemma \ref{lem: Orlicz norms} ahead)}, the problem boils down to proving that, for any constants $a_{i_{1},\dots,i_{k}}$,
\[
\left \| \sum_{i_{1},\dots,i_{k}} \epsilon_{1,i_{1}} \cdots \epsilon_{k,i_{k}} a_{i_{1},\dots,i_{k}} \right \|_{\psi_{2/k}} \lesssim \sqrt{\sum_{i_{1},\dots,i_{k}} a_{i_{1},\dots,i_{k}}^{2}},
\]
but this follows from Corollary 3.2.6 in \cite{delaPenaGine1999}. Indeed, let 
\[
(\epsilon_{1}',\epsilon_{2}',\dots) = (\epsilon_{1,1},\dots,\epsilon_{1,N_{1}},\epsilon_{2,1},\dots,\epsilon_{K,N_{K}}),
\]
and define correspondingly 
\[
b_{j_{1}\dots j_{K}}
=
\begin{cases}
a_{i_1 \dots i_{K}} & \text{if} \ j_1 = i_{1}, j_{2} = N_{1} + i_{2}, \dots, j_{K}=\prod_{k=1}^{K-1} N_{k} + i_{K}, \\
0 & \text{otherwise}
\end{cases}
\]
for $i_{k} =1,\dots,N_k, k=1,\dots,K$. Then, 
\[
\sum_{i_{1},\dots,i_{K}} \epsilon_{1,i_{1}} \cdots \epsilon_{K,i_{K}} a_{i_{1} \dots i_{K}} = \sum_{j_{1}<\dots < j_{K}}\epsilon_{j_{1}}' \dots \epsilon_{j_{K}}' b_{j_{1}\dots j_{K}}.
\]
Corollary 3.2.6 in \cite{delaPenaGine1999} implies that the $\psi_{2/k}$-norm of the right-hand side is $\lesssim \sqrt{\sum_{j_{1} < \cdots < j_{K}} b_{j_{1} \dots j_{K}}^{2}} = \sqrt{\sum_{i_{1},\dots,i_{K}} a_{i_{1} \dots i_{K}}^{2}}$. 
\end{proof}

\red{We shall prove the following technical result used in the proof of Theorem \ref{thm: maximal inequality}.

\begin{lemma}
\label{lem: Orlicz norms}
Let $0 < \beta < \infty$ and $1 \le q < \infty$ be given, and let $m = m(\beta,q)$ be the smallest positive integer satisfying $m\beta \ge q$. Then for every real-valued random variable $\xi$, we have $(\E[|\xi|^q])^{1/q} \le (m!)^{1/( m\beta)} \| \xi \|_{\psi_\beta}$. 
\end{lemma}

\begin{proof}[Proof of Lemma \ref{lem: Orlicz norms}]
By Taylor expansion, we have $\psi_{\beta} (x) = e^{x^{\beta}} - 1 = \sum_{\ell=1}^{\infty} \frac{x^{\ell \beta}}{\ell !} \ge \frac{x^{m \beta}}{m!}$ for $x \ge 0$. 
Choose $C = \| \xi \|_{\psi_{\beta}}$, so that $\E[\psi_{\beta}(|\xi/C|)] \le 1$ (by the monotone convergence theorem the infimum in the definition of Orlicz norm is attained).  Then, $1 \ge \E[|\xi/C|^{m\beta}]/m!$, so $(\E[|\xi|^q])^{1/q} \le (\E[|\xi|^{m\beta}])^{1/(m\beta)} \le C (m!)^{1/(m\beta)}$. 
\end{proof}
}

\begin{remark}[Comparison with \cite{DDG2019}]
\label{rem: comparison with DDG}
Lemma S2 of \cite{DDG2019} derives a symmetrization inequality for the empirical process of an separately exchangeable array. Their symmetrization inequality is substantially different from the maximal inequalities developed in this section, in the sense that their symmetrization inequality is applied to the whole sample mean and does not lead to correct orders to degenerate components of the Hoeffding decomposition. Indeed,  \cite{DDG2019} do not derive a Hoeffding-type decomposition for separately exchangeable arrays. 
\end{remark}

%%%%%%%%%%%%%%%%%%%%%%%%%%%%%%%%%%%%%%%%%%%%%%%%%%%%%%%%%%%%%%%%%%%%%
\section{Proofs for Section \ref{sec:multiway_clustered_data}}\label{sec:proofs_multiway_clustered_data}
%%%%%%%%%%%%%%%%%%%%%%%%%%%%%%%%%%%%%%%%%%%%%%%%%%%%%%%%%%%%%%%%%%%%%

\subsection{Proof of Lemma \ref{lem: H decomposition}}\label{sec:proof of Lemma Hoeffding multiway}
%%%%%%%%%%%%%%%%%%%%%%%%%%%%%%%%%%%%%%%%%%%%%%%%%%%%%%%%%%%%%%%%%%%%%
%\begin{proof}%[Proof of Lemma \ref{lem: H decomposition}] 
The lemma follows from the fact that $\E[ \bX_{\i} \mid (U_{\i \odot \e})_{ \e \le \bm{1}}] = \bX_{\i}$, so that $\bX_{\i} = \hat{\bX}_{\i} + \sum_{ \e \le \bm{1}, \e \ne \bm{1}} \hat{\bX}_{\i\odot \e} =  \sum_{\e \in \{ 0,1 \}^{K} \setminus \{ \bm {0} \}} \hat{\bX}_{\i \odot \e}$.
%\end{proof}
\qed

%%%%%%%%%%%%%%%%%%%%%%%%%%%%%%%%%%%%%%%%%%%%%%%%%%%%%%%%%%%%%%%%%%%%%
\subsection{Proof of Theorem \ref{thm: high-d CLT}}\label{sec:proof of theorem HDCLT multiway}
%%%%%%%%%%%%%%%%%%%%%%%%%%%%%%%%%%%%%%%%%%%%%%%%%%%%%%%%%%%%%%%%%%%%%
We will assume Condition (\ref{eq:condition1}). The proof under Condition (\ref{eq:condition1_poly}) is similar and thus omitted. 
In this proof, let $C$ denote a generic constant that depends only on $\underline{\sigma}$ and $K$. 
\red{Further, we may assume without loss of generality that 
 \begin{equation}
 \frac{D_{\bm{N}}^2 \log^7 (p\overline{N})}{n} \le 1
 \label{eq: rate assumption}
 \end{equation}
 since otherwise the conclusion would be trivial by taking $C$ in the statement of the theorem to be greater than $1$.}
We divide the proof into two steps. 

\underline{Step 1}. We first prove the following bound for the H\'{a}jek projection
\[
\sup_{R \in \mathcal{R}} | \Prob (\sqrt{n}\bS_{\bN}^{W} \in R) -\gamma_{\Sigma}(R) | \le C \left ( \frac{D_{\bN}^{2} \log^{7} (p\overline{N})}{n} \right )^{1/6},
\]
where $\bS_{\bN}^{W} =\sum_{k=1}^{K} N_{k}^{-1}\sum_{i_{k}=1}^{N_{k}} \bW_{k,i_{k}}$. 

For the notational convenience, we assume $K=2$; the proof for the general case is completely analogous. 
Let $\overline{\bW}_{k} = N_{k}^{-1} \sum_{i_{k}} \bW_{k,i_{k}}$. 
\red{We will apply  Proposition 2.1 in \cite{CCK2017AoP} to $\sqrt{N_k}\overline{\bW}_{k}$. Condition (\ref{eq:condition2}) ensures Conditions (M.1) and (M.2) in \cite{CCK2017AoP} to hold, and Condition (\ref{eq:condition1}) ensures Condition (E.1) in \cite{CCK2017AoP} to hold. Conclude from  Proposition 2.1 in \cite{CCK2017AoP} that}
\[
\sup_{R \in \calR} | \Prob( \sqrt{N_{k}} \overline{\bW}_{k}  \in R) - \gamma_{\Sigma_{W_{k}}} (R) | \le C \left ( \frac{D_{\bN}^{2} \log^{7} (p\overline{N})}{n} \right)^{1/6}, \ k=1,2. 
\]
For any rectangle $R = \prod_{j=1}^{p} [a_{j},b_{j}]$, vector $\bm{w} = (w_1,\dots,w_p)^{T} \in \R^{P}$, and scalar $c > 0$, we use the notation $[ c R + \bm{w}] = \prod_{j=1}^{p}[ca_{j}+w_{j},cb_{j}+w_{j}]$, which is still a rectangle. With this in mind, observe that for any rectangle $R \in \calR$, 
\[
\Prob (\sqrt{n}(\overline{\bW}_{1}  + \overline{\bW}_{2} ) \in R) = \E \left [ \Prob \left (\sqrt{N_{1}} \overline{\bW}_{1} \in [\sqrt{N_{1}/n} R - \sqrt{N_{1}}\overline{\bW}_{2}] \mid \overline{\bW}_{2} \right)  \right ]
\]
Since $\overline{\bW}_{1}$ and $\overline{\bW}_{2}$ are independent, the right-hand side is bounded by 
\[
\E \left [  \gamma_{\Sigma_{W_{1}}}( [\sqrt{N_{1}/n} R - \sqrt{N_{1}}\overline{\bW}_{2}]) \right ] + C\left ( \frac{D_{\bN}^{2} \log^{7} (p\overline{N})}{n} \right)^{1/6}.
\]
For $\bm{Y}_{1} \sim N(\bm{0},\Sigma_{W_{1}})$ independent of $\overline{\bW}_{2}$, we have 
\[
 \gamma_{\Sigma_{W_{1}}}( [\sqrt{N_{1}/n} R - \sqrt{N_{1}}\overline{\bW}_{2}])= \Prob (\bm{Y}_{1} \in[\sqrt{N_{1}/n} R - \sqrt{N_{1}}\overline{\bW}_{2}] \mid \overline{\bW}_{2}),
\]
so that 
\[
\begin{split}
&\E \left [  \gamma_{\Sigma_{W_{1}}}( [\sqrt{N_{1}/n} R - \sqrt{N_{1}}\overline{\bW}_{2}]) \right ]  =\Prob(\bm{Y}_{1} \in[\sqrt{N_{1}/n} R - \sqrt{N_{1}}\overline{\bW}_{2}]) \\
&\quad = \Prob (\sqrt{N_{2}}\overline{\bW}_{2} \in [ \sqrt{N_{2}/n} R - \sqrt{N_{2}/N_{1}}\bm{Y}_{1}]) \\
&\quad = \E\left [ \Prob (\sqrt{\red{N_{2}}}\overline{\bW}_{2} \in [ \sqrt{N_{2}/n} R - \sqrt{N_{2}/N_{1}}\bm{Y}_{1}] \mid \bm{Y}_{1}) \right ]. 
\end{split}
\]
Since $\bm{Y}_{1}$ and $\overline{\bW}_{2}$ are independent, the far right-hand side is bounded by
\[
 \E\left [ \gamma_{\Sigma_{W_{2}}} ([ \sqrt{N_{2}/n} R - \sqrt{N_{2}/N_{1}}\bm{Y}_{1}]) \right ] + C\left ( \frac{D_{\bN}^{2} \log^{7} (p\overline{N})}{n} \right)^{1/6}.
 \]
For $\bm{Y}_{2} \sim N(\bm{0},\Sigma_{W_{2}})$ independent of $\bm{Y}_{1}$, the first term can be written as $\Prob (\sqrt{n/N_{1}} \bm{Y}_{1} + \sqrt{n/N_{2}} \bm{Y}_{2} \in R) = \gamma_{\Sigma}(R)$. Conclude that 
\[
\Prob (\sqrt{n}(\overline{\bW}_{1}  + \overline{\bW}_{2} ) \in R)  \le \gamma_{\Sigma}(R) + C\left ( \frac{D_{\bN}^{2} \log^{7} (p\overline{N})}{n} \right)^{1/6}.
\]
\red{The reverse inequality,
\[
\Prob (\sqrt{n}(\overline{\bW}_{1}  + \overline{\bW}_{2} ) \in R)  \ge \gamma_{\Sigma}(R)- C\left ( \frac{D_{\bN}^{2} \log^{7} (p\overline{N})}{n} \right)^{1/6}
\]
 follows similarly. }

\underline{Step 2}. We will prove the conclusion of the theorem. 
Recall the decomposition:
\[
\bS_{\bN} = \bS_{\bN}^{W} + \bm{R}_{\bN} \quad \text{with} \quad \bm{R}_{\bN} = \sum_{k=2}^{K} \sum_{\e \in \mathcal{E}_{k}} \frac{1}{\prod_{k' \in \supp (\e)} N_{k'}} \sum_{\i \in I_{\e}([\bN])} \hat{\bX}_{\i}.
\]
\red{
By Corollary \ref{cor: maximal inequality} applied with $q=1$ (see also the remark after the corollary), we have 
\[
\begin{split}
\E[ \| \bm{R}_{\bN} \|_{\infty}] &\le  \sum_{k=2}^{K} \sum_{\e \in \mathcal{E}_{k}} \frac{1}{\prod_{k' \in \supp (\e)} N_{k'}} \E \left [ \left \| \sum_{\i \in I_{\e}([\bN])} \hat{\bX}_{\i} \right \|_{\infty} \right ] \\
&\le C \sum_{k=2}^{K} (\log p)^{k/2}  \sum_{\e \in \mathcal{E}_{k}} \frac{1}{\sqrt{\prod_{k' \in \supp (\e)} N_{k'}}} \sqrt{\E[\| \bX_{\bm{1}} \|_{\infty}^2]} \\
&\le C \sum_{k=2}^{K} n^{-k/2}(\log p)^{k/2}  \sqrt{\E[\| \bX_{\bm{1}} \|_{\infty}^2]}. 
\end{split}
%&\le C n^{-1}  D_{\bN} (\log p)^{2},
%\end{split}
\]
By Lemma 2.2.2 in \cite{vdVW1996}, we have 
\[
\begin{split}
\sqrt{\E[\| \bX_{\bm{1}} \|_{\infty}^2]} &\le C \| \| \bX_{\bm{1}} \|_{\infty} \|_{\psi_1} \\
&\le C (\log p) \max_{1 \le j \le p} \| \bX_{\bm{1}}^{j} \|_{\psi_1} \\
&\le C (\log p) D_{\bm{N}}.
\end{split}
\]
Using the assumption (\ref{eq: rate assumption}), we conclude that 
\[
\E[ \| \bm{R}_{\bN} \|_{\infty}] \le C \sum_{k=2}^{K} n^{-k/2}(\log p)^{k/2+1} D_{\bm{N}}  \le Cn^{-1} (\log p)^2 D_{\bm{N}}. 
\]
To be precise, the second inequality follows from the following argument. By (\ref{eq: rate assumption}), we know that $n^{-1} \log p \le (\log 2)^{-6}$ (as $D_{\bm{N}} \ge 1$ and $p \ge 2$), so that 
\[
\begin{split}
\sum_{k=2}^K n^{-k/2}(\log p)^{k/2+1} &= n^{-1} (\log p)^2 \sum_{k=2}^K (n^{-1} \log p)^{(k-2)/2} \\
&\le (K-1) (\log 2)^{-3(K-2)} n^{-1} (\log p)^2.
\end{split}
\]
}

For $R = \prod_{j=1}^{p} [a_{j},b_{j}]$ with $\bm{a}=(a_1,\dots,a_p)^{T}$ and $\bm{b} = (b_1,\dots,b_p)^{T}$, we have 
\[
\begin{split}
&\Prob (\sqrt{n}\bS_{\bN} \in R) = \Prob (\{ - \sqrt{n}\bS_{\bN} \le -\bm{a} \} \cap \{ \sqrt{n}\bS_{\bN} \le \bm{b} \}) \\
&\le \Prob (\{ - \sqrt{n}\bS_{\bN} \le -\bm{a} \} \cap \{ \sqrt{n}\bS_{\bN} \le \bm{b} \} \cap \{ \| \sqrt{n} \bm{R}_{N} \|_{\infty} \le t \}) + \Prob (\| \sqrt{n}\bm{R}_{N} \|_{\infty} > t) \\
&\le \Prob (\{ - \sqrt{n}\bS_{\bN}^{W} \le -\bm{a} - t\} \cap \{ \sqrt{n}\bS_{\bN}^{W} \le \bm{b}+t \}) +C  t^{-1} n^{-1/2} D_{\bN} (\log p)^{2} \\
&\le \gamma_{\Sigma} (\{ \bm{ y} \in \R^{p} : -\bm{y} \le -\bm{a} \red{-} t, \bm{y} \le \bm{b}+t \}) \\
&\qquad + C \left ( \frac{D_{\bN}^{2} \log^{7} (p\overline{N})}{n} \right)^{1/6} + C  t^{-1} n^{-1/2} D_{\bN} (\log p)^{2} \\
&\le \gamma_{\Sigma} (R) + C t \sqrt{\log p} + C \left ( \frac{D_{\bN}^{2} \log^{7} (p\overline{N})}{n} \right)^{1/6} + C  t^{-1} n^{-1/2} D_{\bN} (\log p)^{2},
\end{split}
\]
where the last line follows from Nazarov's inequality (see Lemma \ref{lem: AC} in Appendix \ref{sec:auxiliary_lemmas}) \red{together with the fact that the smallest diagonal element of $\Sigma$ is bounded from below by $\underline{\sigma}^2$, which is guaranteed from the second part of Condition (\ref{eq:condition2})}. 
Choosing $t = n^{-1/4} D_{\bN}^{1/2} (\log^{3} p)^{1/4}$, we have 
\[
\begin{split}
\Prob (\sqrt{n}\bS_{\bN} \in R) &\le \gamma_{\Sigma} (R) + C \left ( \frac{D_{\bN}^{2} \log^{7} (p\overline{N})}{n} \right)^{1/6}  + C\left ( \frac{D_{\bN}^{2} \log^{5}p}{n} \right )^{1/4} \\
&\le \gamma_{\Sigma}(R) + C \left ( \frac{D_{\bN}^{2} \log^{7} (p\overline{N})}{n} \right)^{1/6},
\end{split}
\]
\red{where we used the assumption (\ref{eq: rate assumption}) to derive the final inequality.}
\red{The reverse inequality,
\[
\Prob (\sqrt{n}\bS_{\bN} \in R) \ge \gamma_{\Sigma}(R) - C \left ( \frac{D_{\bN}^{2} \log^{7} (p\overline{N})}{n} \right)^{1/6}
\]
follows similarly.} \qed

%%%%%%%%%%%%%%%%%%%%%%%%%%%%%%%%%%%%%%%%%%%%%%%%%%%%%%%%%%%%%%%%%%%%%
\subsection{Proof of Theorem \ref{thm:bootsrap_validity}}\label{sec:proof of multiway multiplier bootstrap}
We separately prove the theorem under Cases (i) and (ii). 

\underline{Case (i)}. 
Let $C$ denote a generic constant that depends only on $\nu, \underline{\sigma}, K$, and $C_1$. Also the notation $\lesssim$ means that the left-hand side is bounded by the right-hand side up to a constant that depends only on $\nu, \underline{\sigma}, K$, and $C_1$. We divide the proof into two steps. 

\underline{Step 1}. 
Define
\[
\hat \Delta_{W}=\max_{1\le j \le p;1 \le k \le K} \frac{1}{N_k} \sum_{i_k=1}^{N_k} ( \overline X_{k,i_k}^j 
- W_{k,i_k}^j )^2.
\]
We will show that $\Prob  ( \overline \sigma^2 \hat \Delta_{W}\log^4 p > n^{-(\zeta-1/\nu)} ) \lesssim n^{-1}$.
It suffices to show that
$\Prob(\overline \sigma^2 \hat \Delta_{W,1,1}\log^4 p >  n^{-(\zeta-1/\nu)})\lesssim n^{-1}$, where
\begin{align*}
 \hat \Delta_{W,1,1}=\max_{1\le \ell\le p}\frac{1}{N_1} \sum_{i_1=1}^{N_1}
( \overline{X}_{1,i_{1}}^\ell -W_{1,i_1}^\ell)^2,
\end{align*}
as similar bounds hold for $\max_{1\le \ell\le p}N_k^{-1} \sum_{i_k=1}^{N_k}
( \overline{X}_{k,i_{k}}^\ell -W_{k,i_k}^\ell)^2$ with $k\in \{2,\dots,K\}$.

We first note that 
\begin{align*}
\hat \Delta_{W,1,1}
=
\max_{1\le \ell \le p} \frac{1}{N_1} \sum_{i_1=1}^{N_1} ( \overline X_{1,i_1}^\ell - W_{1,i_1}^\ell )^2
\le
 \frac{1}{N_1} \sum_{i_1=1}^{N_1}\|\overline{\bX}_{1,i_1} - \bW_{1,i_1}\|_{\infty}^2.
\end{align*}
Pick any $i_1 \in \N$. 
For each $\i_{-1}=(i_2,\dots,i_K)\in \N^{K-1}$ and $\e \in \{0,1\}^{K-1}$, define the vector
\begin{align*}
V_{\i_{-1}\odot \e}=(U_{(0,\i_{-1} \odot \e)},U_{(i_1,\i_{-1} \odot \e)}).
\end{align*}
%Letting $\mathfrak{g}_{U_{(i_1,0,\dots,0)}}(\cdot)=\mathfrak{g}(U_{(i_1,0,\dots,0)},\cdot)$ conditionally on $U_{(i_1,0,\dots,0)}$, 
With this notation, we can rewrite $\bX_{\i}$ with $\i = (i_1,\i_{-1})$  as 
\begin{align*}
\bX_{\i}=\mathfrak{g}\big(U_{(i_1,0,\dots,0)},(V_{\i_{-1}\odot \e})_{\e\in \{0,1\}^{K-1}\setminus \{ \bm{0} \}}\big).
\end{align*}
From this expression, we see that, conditionally on $U_{(i_1,0,\dots,0)}$, the $(K-1)$-array $(\bX_{(i_1,\i_{-1})})_{\i_{-1} \in \N^{K-1}}$ is separately exchangeable with mean vector $\bW_{1,i_1}$ generated by $\{ V_{\i_{-1}\odot \e} : \i_{K-1} \in \N^{K-1}, \e \in \{0,1\}^{K-1} \setminus \{ \bm{0} \} \}$. 
Applying Corollary \ref{cor: maximal inequality} conditionally on $U_{(i_1,0,\dots,0)}$ (the fact that $U_{\i \odot \e}$ are uniform on $[0,1]$ is not crucial in the proof of Corollary \ref{cor: maximal inequality}) combined with Jensen's inequality, we have 
\[
\E[\|\overline{\bX}_{1,i_1} - \bW_{1,i_1}\|_{\infty}^{2\nu} \mid U_{(i_1,0\dots,0)}] \\
\lesssim \underbrace{\left ( \sum_{k=1}^{K-1}n^{-k/2}  (\log p)^{k/2} \right )^{2\nu}}_{\lesssim  (n^{-1} \log p)^{\nu}} \E[\|\bX_{(i_1,1,\dots,1)}\|_{\infty}^{2\nu} \mid U_{(i_1,\dots,0)}],
\]
so that by Fubini's theorem
\[
\E[\|\overline{\bX}_{1,i_1} - \bW_{1,i_1}\|_{\infty}^{2\nu}] \lesssim (n^{-1} \log p)^{\nu} \E[\|\bX_{(i_1,1,\dots,1)}\|_{\infty}^{2\nu}] \lesssim (n^{-1} D_{\bN}^2 \log^3 p)^{\nu}.
\]
This implies that  $\E[(\overline \sigma^2 \hat \Delta_{W,1,1} \log^4 p)^{\nu}]\lesssim n^{-\zeta \nu}$ under our assumption. By Markov's inequality, we conclude that 
\begin{align*}
\Prob\left(\overline \sigma^2 \hat \Delta_{W,1,1} \log^4 p > n^{-\zeta +1/\nu}\right)\lesssim n^{-1}.
\end{align*}
This completes Step 1.

\underline{Step 2}. 
Conditionally on $\bX_{[\bN]}$, we have $\sqrt{n}\bS_{\bN}^{MB} \sim N(\bm{0},\hat{\Sigma})$,
where 
\[
\hat{\Sigma} =\sum_{k=1}^K \frac{n}{N_k^2} \sum_{i_k=1}^{N_k}(\overline{\bX}_{k,i_{k}}  - \bS_{\bN})(\overline{\bX}_{k,i_{k}}  -  \bS_{\bN})^{T} .
\]
Hence, to obtain a bound on $\sup_{R\in \calR}|\Prob_{|\bX_{[\bN]}}(\sqrt{n}\bS_{\bN}^{MB}\in R)-\gamma_{\Sigma}(R) |$, it suffices to bound $\| \hat{\Sigma} - \Sigma \|_{\infty}$ in view of Lemma \ref{lem:Gaussian_comparison} in Appendix \ref{sec:auxiliary_lemmas}. We note that
\begin{align*}
\|\hat \Sigma - \Sigma\|_{\infty}
\le&
\sum_{k=1}^K
 \underbrace{\max_{1\le j,\ell\le p}\Big|\frac{n}{N_k^2} \sum_{i_k=1}^{N_k} \overline X_{k,i_{k}}^j \overline{X}_{k,i_{k}}^\ell- \frac{n}{N_k} S_{\bN}^j S_{\bN}^\ell  - \frac{n}{N_k}\E[W_{k,1}^j W_{k,1}^\ell]\Big|}_{=:\hat \Delta_{W,k}}. 
\end{align*}
We will focus on bounding $\hat \Delta_{W,1}$ as similar bounds hold for $\hat \Delta_{W,k}$ with $ k\in\{2,\dots, K\}$. 

Observe that
\[
\begin{split}
\frac{n}{N_1^2} \sum_{i_1=1}^{N_1}  \overline{X}_{1,i_{1}}^j  \overline{X}_{1,i_{1}}^\ell
=&
\frac{n}{N_1^2} \sum_{i_1=1}^{N_1}( \overline{X}_{1,i_{1}}^j -W_{1,i_1}^j)( \overline{X}_{1,i_{1}}^\ell -W_{1,i_1}^\ell)
+
\frac{n}{N_1^2} \sum_{i_1=1}^{N_1}( \overline{X}_{1,i_{1}}^j -W_{1,i_1}^j)W_{1,i_1}^\ell \\
& +\frac{n}{N_1^2} \sum_{i_1=1}^{N_1}W_{1,i_1}^j( \overline{X}_{1,i_{1}}^\ell -W_{1,i_1}^\ell)
+\frac{n}{N_1^2}\sum_{i_1=1}^{N_1} W_{1,i_1}^j W_{1,i_1}^\ell.
\end{split}
\]
By the Cauchy-Schwarz inequality and the definition of $n$, we obtain
\begin{equation}
\begin{split}
\hat\Delta_{W,1}
\le&
 \underbrace{\max_{1\le \ell\le p}\frac{1}{N_1} \sum_{i_1=1}^{N_1}
( \overline{X}_{1,i_{1}}^\ell -W_{1,i_1}^\ell)^2}_{=\hat \Delta_{W,1,1}}
+
2 \hat \Delta_{W,1,1}^{1/2}\sqrt{
\max_{1\le \ell \le p}\frac{1}{N_1}\sum_{i_1=1}^{N_1}|W_{1,i_1}^\ell|^2 
} \\
&+
\underbrace{\max_{1\le j,\ell\le p}\Big | \frac{1}{N_1}\sum_{i_1}^{N_1}(W_{1,i_1}^j W_{1,i_1}^\ell - \E[W_{1,1}^j W_{1,1}^\ell] )\Big| }_{=:\hat \Delta_{W,1,2}}
+
\max_{1\le \ell \le p}|S_{\bN}^\ell|^2.
\end{split}
\label{eq:thm:bootstrap:1}
\end{equation}
For the second term on the right-hand side, we have
\begin{align}
\frac{1}{N_1}\sum_{i_1=1}^{N_1}|W_{1,i_1}^\ell|^2 
\le
\E[|W_{1,i_1}^\ell|^2] +
\frac{1}{N_1}\sum_{i_1=1}^{N_1}(|W_{1,i_1}^\ell|^2 -\E[|W_{1,i_1}^\ell|^2]) \le \overline \sigma^2+\hat \Delta_{W,1,2}.
\label{eq:thm:bootstrap:2}
\end{align}
Further, since $S_{\bN}^\ell = N_1^{-1} \sum_{i_1=1}^{N_1}(\overline{ X}_{1,i_1}^\ell - W_{1,i_1}^\ell )+
N_1^{-1} \sum_{i_1=1}^{N_1} W_{1,i_1}^\ell $, we have
\begin{align}
\max_{1\le \ell \le p}|S_{\bN}^\ell|^2
\lesssim
\hat \Delta_{W,1,1} + \hat \Delta_{W,1,3}^2,\label{eq:thm:bootstrap:3}
\end{align}
where  $\hat \Delta_{W,1,3}=\max_{1\le \ell \le p}|N_1^{-1}\sum_{i_1=1}^{N_1} W_{1,i_1}^\ell |$. 
Combining (\ref{eq:thm:bootstrap:1})--(\ref{eq:thm:bootstrap:3}), we have
\begin{align*}
\hat \Delta_{W,1}
\lesssim
\hat \Delta_{W,1,1}+
 \overline \sigma \hat \Delta_{W,1,1}^{1/2} + \hat \Delta_{W,1,2} + \hat \Delta_{W,1,3}^2.
\end{align*}
It remains to find bounds on the four terms on the right-hand side.

First, by Step 1, we have
$\hat\Delta_{W,1,1} \log^ 2 p \le C n^{-(\zeta-1/\nu)}$ and
$\overline \sigma \hat \Delta_{W,1,1}^{1/2}\log^2 p\le C n^{-(\zeta-1/\nu)/2} $ 
with probability at least $1-Cn^{-1}$.
Second, we note that
\[
\E\left[\max_{1\le i_1 \le N_1}\max_{1\le \ell\le p }|W_{1,i_1}^\ell|^4 \right] 
\lesssim 
(\log p \overline N)^4 \max_{1\le \ell\le p }\||W_{1,1}^\ell|^4\|_{\psi_{1/4}} =(\log p \overline N)^4 \underbrace{\max_{1\le \ell \le p}\|W_{1,1}^\ell\|_{\psi_{1}}^4}_{\le D_{\bN}^{4}}.  
\]
Applying Lemma 8 in \cite{CCK2015PTRF}, we have
\[
\begin{split}
\E[\hat \Delta_{W,1,2}]
&\lesssim N_{1}^{-1}\sqrt{(\log p) \max_{1\le j,\ell\le p} \sum_{i_1=1}^{N_1} \E[|W_{1,i_1}^j W_{1,i_1}^\ell|^2]}  + N_1^{-1}\sqrt{\E\left[\max_{1\le i_1 \le N_1}\max_{1\le \ell \le p} |W_{1,i_1}^\ell|^4 \right] }\log p \\
&\lesssim  N_1^{-1/2} D_{\bN} \log^{1/2} p + N_1^{-1}D_{\bN}^2 (\log p)\log^2 (p \overline N)\\
&\lesssim  n^{-1/2} D_{\bN} \log^{1/2} p + n^{-1}D_{\bN}^2  \log^3 (p \overline N).
\end{split}
\]
Now, applying Lemma E.2 in \cite{CCK2017AoP} with $\eta=1$ and $\beta=1/2$, together with the fact that 
\[
\left\|\max_{1\le i_1 \le N_1}\max_{1\le \ell \le p} |W_{1,i_1}^\ell|^2 \right\|_{\psi_{1/2}}=\left\|\max_{1\le i_1 \le N_1}\max_{1\le \ell \le p} |W_{1,i_1}^\ell | \right \|_{\psi_{1}}^2\lesssim  (\log pN_{1})^{2} D_{\bN}^{2},
\]
we have
\[
\Prob \left(\hat \Delta_{W,1,2} \ge 2 \E[\hat \Delta_{W,1,2}]+t \right) \le
\exp\left(-\frac{n t^2}{3 D_{\bN}^2}\right) + 
3\exp\left\{-\left(\frac{n t}{CD_{\bN}^2 \log^2(p \overline N)}\right)^{1/2}\right\}.
\]
Setting  $t=\{C n^{-1} D_{\bN}^2 \log n\}^{1/2}\vee \{C n^{-1} D_{\bN}^2 (\log^2 n)\log^3(p \overline N) \}$, %obtained by solving for t that makes the right-hand side prob be Cn^{-1}.
we conclude that
\[
\Prob\left(\hat \Delta_{W,1,2} \ge  C\{(n^{-1} D_{\bN}^2 \log^{1/2} (pn)  + n^{-1}D_{\bN}^2 (\log n)^2\log^3 (p \overline N )\}\right) \le Cn^{-1}.
\]
Condition (\ref{eq:bootstrap_rate}) then guarantees that $\hat \Delta_{W,1,2} \log^{2} p \le Cn^{-\zeta/2}$ with probability at least $1-Cn^{-1}$.

Finally, since $\overline \sigma^2 \le (\max_{1 \le \ell \le p}\E[|W_{1,1}^\ell|^3])^{2/3} \le 1 + \max_{1\le \ell\le p} \E[|W_{1,1}^\ell|^3]\lesssim D_{\bN}$, using Lemma 8 in \cite{CCK2015PTRF}, we have
\[
\E[\hat \Delta_{W,1,3}]
\lesssim
(n^{-1} D_{\bN} \log p)^{1/2} + n^{-1} D_{\bN} \log (p \overline N).
\]
Applying Lemma E.2 in \cite{CCK2017AoP} with $\eta=1$ and $\beta=1$, we have
\[
\hat \Delta_{W,1,3}^2 \log^2 p \le
 \underbrace{C\{n^{-1} D_{\bN} (\log^2 p) \log (pn) + n^{-2} D_{\bN}^2 (\log^2 n) (\log^2 p) \log^{2} (p \overline N)\}}_{\le Cn^{-\zeta}}
\]
with probability at least $1-C n^{-1}$.
Conclude that $\hat\Delta_{W,1} \log^2 p \le C n^{-(\zeta - 1/\nu)/2}$
with probability at least $1-C n^{-1}$. The desired result then follows from Lemma \ref{lem:Gaussian_comparison} in Appendix \ref{sec:auxiliary_lemmas}.

\underline{Case (ii)}. The proof is similar to the previous case. We only point out required modifications. Let $C$ denote a generic constant that depends only on $q, \underline{\sigma}, K$, and $C_1$. The similar modification applies to $\lesssim$.

 Set $\nu = q/2$ in the previous case. Under Case (ii), we have 
\[
\E[\| \overline{\bX}_{1,i_1} - \bW_{1,i_1}\|_{\infty}^{q}] \lesssim (n^{-1} \log p)^{\nu}\underbrace{\E[\| \bX_{(i_1,1,\dots,1)} \|_{\infty}^q]}_{\le D_{\bN}^q},
\]
which implies that $\E[(\overline{\sigma}^2 \hat{\Delta}_{W,1,1}\log^4p)^{q/2}] \lesssim n^{-\zeta q/2}$. Markov's inequality yields that 
\[
\Prob\left(\overline \sigma^2 \hat \Delta_{W,1,1}\log^4 p >n^{-\zeta+2/q} \right)\lesssim n^{-1}.
\]
 In view of the previous case, it remains to find bounds on $\hat{\Delta}_{W,1,2}$ and $\hat{\Delta}_{W,1,3}$.

Applying Lemma 8 in \cite{CCK2015PTRF}, we have
\begin{align*}
\E[\hat \Delta_{W,1,2}]
&\lesssim N_{1}^{-1}\sqrt{(\log p) \max_{1\le j,\ell\le p} \sum_{i_1=1}^{N_1} \E[|W_{1,i_1}^j W_{1,i_1}^\ell|^2]}  + N_1^{-1}\sqrt{\E\left[\max_{1\le i_1 \le N_1}\max_{1\le \ell \le p} |W_{1,i_1}^\ell|^4 \right] }\log p \\
&\lesssim  N_1^{-1/2} D_{\bN} \log^{1/2} p + N_1^{-1+2/q}D_{\bN}^2 \log p \\
&\lesssim  n^{-1/2} D_{\bN} \log^{1/2} p + n^{-1+2/q}D_{\bN}^2 \log p.
\end{align*}
Applying  the Fuk-Nagaev inequality (Lemma E.2 in \cite{CCK2017AoP}) with $s=q/2$, we have
\begin{align*}
\Prob \left(\hat \Delta_{W,1,2} \ge 2 \E[\hat \Delta_{W,1,2}]+t \right) &\le
\exp\left(-\frac{N_1 t^2}{3 D_{\bN}^2}\right) 
+ 
\frac{C N_1 D_{\bN}^q}{ N_1^{q/2} t^{q/2}} \\
&\le
\exp\left(-\frac{n t^2}{3 D_{\bN}^2}\right) 
+ 
\frac{C D_{\bN}^q}{ n^{q/2-1} t^{q/2}}.
\end{align*}
Setting $t=(C n^{-1} D_{\bN}^2 \log n)^{1/2}\bigvee ( Cn^{-1+4/q} D_{\bN}^2 )$, we have
\begin{align*}
\Prob \left(\hat \Delta_{W,1,2} \ge C\{ (n^{-1} D_{\bN}^2 \log (pn) )^{1/2} + n^{-1+4/q}   D_{\bN}^2 \log p \} \right)\le Cn^{-1}.
\end{align*}
Condition (\ref{eq:bootstrap_rate_poly}) then guarantees that $\hat \Delta_{W,1,2} \log^{2}p \le Cn^{-\zeta/2}$
with probability at least $1-Cn^{-1}$.
A bound for $\hat \Delta_{W,1,3}$ can be obtained similarly. Using Lemma 8 in \cite{CCK2015PTRF}, we have 
\begin{align*}
\E[\hat \Delta_{W,1,3}]\lesssim&
(n^{-1} D_{\bN} \log p)^{1/2} + n^{-1+1/q} D_{\bN} \log p.
\end{align*}
Applying Lemma E.2 in \cite{CCK2017AoP} with $s=q$, we have 
\begin{align*}
\Prob \left(\hat \Delta_{W,1,3} \ge 2 \E[\hat \Delta_{W,1,3}]+t \right) \le&
\exp\left(-\frac{n t^2}{3 D_{\bN}}\right) 
+ 
\frac{C D_{\bN}^q}{ n^{q-1} t^{q}}.
\end{align*}
Setting $t=(Cn^{-1}D_\bN \log n)^{1/2}\bigvee (C n^{-1+2/q}D_\bN)$, we conclude that 
\begin{align*}
\hat \Delta_{W,1,3}^2 \log^2 p \le
\underbrace{C\{n^{-1} D_{\bN} (\log^{2}p)\log (p n) + n^{-2+4/q} \log^4 p\}}_{\le Cn^{-\zeta}}
\end{align*}
with probability at least $1-Cn^{-1}$.
\qed

%%%%%%%%%%%%%%%%%%%%%%%%%%%%%%%%%%%%%%%%%%%%%

%%%
%%%%%%%%%%%%%%%%%%%%%%%%
\subsection{Proof of Corollary \ref{cor:normalized_multiway_MB}}\label{sec:proof of corollary normalized multiway}
We only prove the corollary under Case (i). The proof for Case (ii) is similar. 
Let $C$ denote a generic constant that depends only on $\underline{\sigma}, K$, and $C_1$. 
We first note that from the proof of Theorem \ref{thm:bootsrap_validity}, we have 
\begin{align*}
\max_{1\le j \le p}\left|\sigma_j^2/\hat\sigma_j^2 -1\right|\le C n^{-\zeta/3}/\log^2 p
\end{align*}
with probability at least $1-Cn^{-1}$ (choose $\nu = 3/\zeta$ in Theorem \ref{thm:bootsrap_validity}). 
By Theorem \ref{thm: high-d CLT}, we have 
\begin{align*}
\sup_{R\in \calR}\left|\Prob(\sqrt{n} {\Lambda}^{-1/2}\bS_{\bN}\in R)-\Prob( \Lambda^{-1/2} \bm{Y}\in R)\right|\le C n^{-\zeta/6}.
\end{align*} 
By the Borell-Sudakov-Tsirel'son inequality and the fact $\E[\|\Lambda^{-1/2}\bm{Y}\|_{\infty}]\le C \sqrt{\log p}$, which is implied by the Gaussianity of $\Lambda^{-1/2}\bm{Y}$, we have 
\begin{align*}
\Prob \left(\|\Lambda^{-1/2}\bm{Y}\|_{\infty} > C \sqrt{\log (pn) } \right)\le n^{-1}.
\end{align*}
Combining the high-dimensional CLT, we see that 
\begin{align*}
\Prob \left(\|\sqrt{n} {\Lambda}^{-1/2}\bS_{\bN} \|_{\infty}> C \sqrt{\log (pn) } \right)\le C n^{-\zeta/6}.
\end{align*}
Since $\frac{n^{-\zeta/3}}{\log^2 p}\times \sqrt{\log (p n)}\le \frac{C n^{-\zeta/6}}{\log^{3/2} p}$, we have
\begin{align*}
\Prob \left(\|\sqrt{n} ({\hat{\Lambda}}^{-1/2}-\Lambda^{-1/2})\bS_{\bN} \|_{\infty}> t_n \right)\le C n^{-\zeta/6}
\end{align*}
with $t_n=\frac{C n^{-\zeta/6}}{\log^{3/2} p}$.

Now, for $R=\prod_{j=1}^p [a_j,b_j]$ with $\bm{a}=(a_1,\dots,a_p)^{T}$ and $\bm{b} = (b_1,\dots,b_p)^{T}$, we have 
\begin{align*}
\Prob \left(\sqrt{n} \hat{{\Lambda}}^{-1/2}\bS_{\bN} \in R\right)
&\le
\Prob \left(\{-\sqrt{n} \Lambda^{-1/2}\bS_{\bN} \le -\bm{a}- t_n \}\cap \{\sqrt{n}\Lambda^{-1/2}\bS_{\bN} \le \bm{b} + t_n \}\right)\\
&\quad +\Prob \left(\|\sqrt{n} ({\hat{\Lambda}}^{-1/2}-\Lambda^{-1/2})\bS_{\bN} \|_{\infty}> t_n \right)\\
&\le
\Prob \left(\{- {\Lambda}^{-1/2}\bm{Y} \le -\bm{a} - t_n \}\cap \{{\Lambda}^{-1/2}\bm{Y} \le \bm{b}+ t_n \}\right)
+
C n^{-\zeta/6}\\
&\le
\Prob (\Lambda^{-1/2} \bm{Y} \in R)+ C n^{-\zeta/6},
\end{align*}
where the last inequality follows from Lemma \ref{lem: AC} together with the fact that
\[
t_n \sqrt{\log p}\le C n^{-\zeta/6}/\log p \le C n^{-\zeta/6}.
\]
Thus, we have
\begin{align*}
\Prob(\sqrt{n} \hat{\Lambda}^{-1/2}\bS_{\bN}\in R)\le 
\Prob( \Lambda^{-1/2} \bm{Y}\in R)+C n^{-\zeta/6}.
\end{align*}
Likewise, we have
\begin{align*}
\Prob(\sqrt{n} \hat{\Lambda}^{-1/2}\bS_{\bN}\in R)\ge 
\Prob( \Lambda^{-1/2} \bm{Y}\in R)- C n^{-\zeta/6}.
\end{align*}
Conclude that 
\begin{align*}
\sup_{R\in \calR}\left|\Prob(\sqrt{n} \hat{\Lambda}^{-1/2}\bS_{\bN}\in R)-\Prob( \Lambda^{-1/2} \bm{Y}\in R)\right|\le C n^{-\zeta/6}.
\end{align*}
Similarly, using Theorem \ref{thm:bootsrap_validity} and following similar arguments, we conclude that
\begin{align*}
\sup_{R\in \calR}\left|\Prob_{|\bX_{[\bN]}}(\sqrt{n} \hat{\Lambda}^{-1/2}\bS_{\bN}^{MB}\in R)-\Prob( \Lambda^{-1/2} \bm{Y}\in R)\right|\le C n^{-\zeta/6}
\end{align*}
with probability at least $1-Cn^{-1}$. 
%\end{proof}
\qed

%%%%%%%%%%%%%%%%%%%%%%%%%%%%%%%%%%%%%%%%%%%%%%%%%%%%%%%%%%%%%%%%%%%%%%%%%5

\section{Maximal inequalities for jointly exchangeable arrays}\label{sec:maximal_inequality_for_polyadic_data}
%%%%%%%%%%%%%%%%%%%%%%%%%%%%%%%%%%%%%%%%%%%%%%%%%%%%%%%%%%%%%%%%%%%%%
In this section, we shall develop maximal inequalities for jointly exchangeable arrays. 
As in Section \ref{sec:polyadic_data}, let $( \bX_{\i})_{\i \in I_{\infty,K}}$ be a $K$-array consisting of random vectors in $\R^{p}$ with mean zero generated by the structure (\ref{eq: Aldous-Hoover_JE}), i.e., $\bX_{\i} = \mathfrak{g}((U_{\{ \i \odot \e \}^{+}})_{\e \in \{ 0,1 \}^{K} \setminus \{ \bm{0} \}})$. We will follow the notations used in Section \ref{sec:polyadic_data}.  Recall that $I_{n,K} = \{ (i_1,\dots,i_K) : 1 \le i_1,\dots,i_K \le n \ \text{and} \ i_1,\dots,i_K \ \text{are distinct} \}$.

We first point out that when analyzing the sample mean $\bS_n$, it is without loss of generality to assume that $\bX_{\i}$ is symmetric in the components of $\i$, i.e., 
\begin{equation}
\bX_{(i_1,\dots,i_K)} = \bX_{(i_1',\dots,i_K')} \label{eq: symmetry}
\end{equation}
for any permutation $(i_1',\dots,i_K')$ of $(i_1,\dots,i_K)$. This is because even if $\bX_{\i}$ is not symmetric in the components of $\i$, we can instead work with its symmetrized version
\[
\check{\bX}_{(i_1,\dots,i_K)} = \frac{1}{K!} \sum_{(i_1',\dots,i_K')} \bX_{(i_1',\dots,i_K')},
\]
where the summation is taken over all permutations of $(i_1,\dots,i_K)$. It is not difficult to see that the array $(\check{\bX}_{\i})_{\i \in I_{\infty,K}}$ continues to be jointly exchangeable and satisfies that 
\[
\bS_{n} = \frac{(n-K)!}{n!} \sum_{\i \in I_{n,K}} \check{\bX}_{\i} = \binom{n}{K}^{-1} \sum_{1 \le i_1 < \cdots < i_{K} \le n} \check{\bX}_{\i}.
\]
Henceforth, in this section, we will maintain Condition (\ref{eq: symmetry}). 

In the decomposition (\ref{eq: H-decomp polyadic}), the second term on the right-hand side 
\[
\mathbb{U}_{n} =\binom{n}{K}^{-1} \sum_{1 \le i_1 < \cdots < i_{K} \le n} \left( \E[\bX_{\i} \mid U_{i_1},\dots,U_{i_{K}}] - \sum_{k=1}^{K} \E[\bX_{\i} \mid U_{i_k}] \right )
\]
is a degenerate $U$-statistic (with a symmetric kernel) of degree $K$. Indeed, if we define $\mathfrak{t} (u_1,\dots,u_K) = \E[\bX_{(1,\dots,K)} \mid U_1 = u_1,\dots,U_K=u_K] - \sum_{k=1}^{K} \E[\bX_{(1,\dots,K)} \mid U_k = u_k]$, then $\mathfrak{t}$ is symmetric and 
\[
\mathbb{U}_{n} = \binom{n}{K}^{-1} \sum_{1 \le i_1 < \cdots < i_K \le n} \mathfrak{t}(U_{i_1},\dots,U_{i_K}).
\]
The kernel $\mathfrak{t}$ is degenerate as 
\[
\E[ \mathfrak{t}(u,U_2,\dots,U_{K})] = \E[\bX_{(1.\dots,K)} \mid U_1= u]-\E[\bX_{(1,\dots,K)} \mid U_1 = u] = 0.
\]
Applying Corollary 5.6 in \cite{ChenKato2019b}, we obtain the following lemma. 

\begin{lemma}
\label{lem: max_ineq_poly1}
For any $q \in [1,\infty)$, we have 
\[
(\E[\| \mathbb{U}_{n} \|_{\infty}^{q}])^{1/q} \le C \sum_{k=2}^{K} n^{-k/2} (\log p)^{k/2} (\E[\| \bX_{(1,\dots,K)} \|_{\infty}^{q \vee 2}])^{1/(q \vee 2)},
\]
where $C$ is a constant that depends only on $q$ and $K$.
\end{lemma}

We turn to the analysis of the third term on the right-hand side of (\ref{eq: H-decomp polyadic})
\[
\begin{split}
& \sum_{k=2}^K \frac{(n-K)!}{n!} \sum_{\i \in I_{n,K}}\left(\E[\bX_{\i}\mid (U_{\{\i\odot \e\}^+})_{\e \in \cup_{r=1}^k \calE_r} ]-\E[\bX_{\i}\mid (U_{\{\i\odot \e\}^+})_{\e \in \cup_{r=1}^{k-1}\calE_r} ]\right) \\
 &\quad =\sum_{k=2}^K \binom{n}{K}^{-1} \sum_{1 \le i_1 < \cdots < i_K \le n}\left(\E[\bX_{\i}\mid (U_{\{\i\odot \e\}^+})_{\e \in \cup_{r=1}^k \calE_r} ]-\E[\bX_{\i}\mid (U_{\{\i\odot \e\}^+})_{\e \in \cup_{r=1}^{k-1}\calE_r} ]\right)
 \end{split}
 \]
 where the quality follows from Condition (\ref{eq: symmetry}).
 
\begin{lemma}
\label{lem: max_ineq_poly2}
For any $k=2,\dots,K$ and $q \in [1,\infty)$, we have 
\[
\begin{split}
&\left ( \E\left [ \left \| \binom{n}{K}^{-1} \sum_{1 \le i_1 < \cdots < i_K \le n}\left(\E[\bX_{\i}\mid (U_{\{\i\odot \e\}^+})_{\e \in \cup_{r=1}^k \calE_r} ]-\E[\bX_{\i}\mid (U_{\{\i\odot \e\}^+})_{\e \in \cup_{r=1}^{k-1}\calE_r} ]\right) \right \|_{\infty}^{q} \right ] \right )^{1/q}\\
&\quad \le 
C n^{-k/2} (\log p)^{1/2} (\E[\| \bX_{(1,\dots,K)} \|_{\infty}^{q \vee 2}])^{1/(q \vee 2)},
\end{split}
\]
where $C$ is a constant that depends only on $q$ and $K$. 
\end{lemma}

Before the formal proof of Lemma \ref{lem: max_ineq_poly2}, which is somewhat involved, we shall look at the case with $k=K=2$ to understand the bound. If $k=K=2$, then the term in question is 
\[
\binom{n}{2}^{-1} \sum_{1 \le i < j \le n} (\E[\bX_{(i,j)} \mid U_{i},U_{j},U_{\{i,j\}}] - \E[\bX_{(i,j)} \mid U_{i},U_{j}] ).
\]
Conditionally on $U_{i}$'s, this is the sum of independent random vectors with mean zero, so the bound in the lemma can be deduced from applying the symmetrization inequality \citep[][Lemma 2.3.6]{vdVW1996} conditionally on $U_{i}$'s and then Lemma 2.2.2 in \cite{vdVW1996} to the weighted sum of Rademacher variables conditionally on all $U$-variables. 
The general case is more involved and we will apply the decoupling inequality for $U$-statistics with index-dependent kernels (cf. Theorem 3.1.1 in \cite{delaPenaGine1999}) and adapt the telescoping sum technique used in the proof of Lemma A.1 in \cite{DDG2019}.\footnote{\red{We are indebted to an anonymous referee who pointed out a mistake in the initial proof of the lemma.}}

\red{
\begin{proof}[Proof of Lemma \ref{lem: max_ineq_poly2}]
In this proof, the notation $\lesssim$ means that the left-hand side is bounded by the right-hand side up to a constant that depends only on $q$ and $K$. Fix any $k=2,\dots,K$. 
Conditionally on $\calU_{k-1} = \{ U_{\{\i\odot \e\}^+} : \e \in \cup_{r=1}^{k-1}\calE_r, \i \in I_{\infty,K} \}$, the component 
\[
\E[\bX_{\i}\mid (U_{\{\i\odot \e\}^+})_{\e \in \cup_{r=1}^k \calE_r} ]-\E[\bX_{\i}\mid (U_{\{\i\odot \e\}^+})_{\e \in \cup_{r=1}^{k-1}\calE_r} ]
\]
is a function of $(U_{\{\i\odot \e\}^+})_{\e \in \calE_k}$ with mean zero
\[
\E[\bX_{\i}\mid (U_{\{\i\odot \e\}^+})_{\e \in \cup_{r=1}^k \calE_r} ]-\E[\bX_{\i}\mid (U_{\{\i\odot \e\}^+})_{\e \in \cup_{r=1}^{k-1}\calE_r} ] = \mathfrak{h}_{(\{\i\odot \e\}^+)_{\e \in \calE_{k}}} ((U_{\{ \i \odot \e \}^{+}})_{\e \in \calE_{k}}). 
\]
The function $\mathfrak{h}_{(\{\i\odot \e\}^+)_{\e \in \calE_{k}}}$ implicitly depends on $(U_{\{\i\odot \e\}^+})_{\e \in \cup_{r=1}^{k-1}\calE_r}$, so that it is indexed by $(\{\i\odot \e\}^+)_{\e \in \calE_{k}}$ (the vector $({\{\i\odot \e\}^+})_{\e \cup_{r=1}^{k-1}\calE_r}$ is uniquely determined by $(\{\i\odot \e\}^+)_{\e \in \calE_{k}}$ so it is enough to index the function by $(\{\i\odot \e\}^+)_{\e \in \calE_{k}}$). 
Define 
\[
\calJ_{n,k} = \{  (\{ \i \odot \e \}^+)_{\e \in \calE_k} : 1 \le i_1 < \cdots < i_K \le n \}.
\]
This is a collection of vectors of sets where each vector contains $m_k = \binom{K}{k}$ sets.
We denote a generic element of $\calJ_{n,k}$ by $\vec{J} = (J_1,\dots,J_{m_k})$ by ordering the elements of $\calE_k$. We will also write $U_{\vec{J}} = (U_{J_1},\dots,U_{J_{m_k}})$. Then we arrive at the expression
\[
\sum_{1 \le i_1 < \cdots < i_K \le n}\left(\E[\bX_{\i}\mid (U_{\{\i\odot \e\}^+})_{\e \in \cup_{r=1}^k \calE_r} ]-\E[\bX_{\i}\mid (U_{\{\i\odot \e\}^+})_{\e \in \cup_{r=1}^{k-1}\calE_r} ]\right) = \sum_{\vec{J} \in \calJ_{n,k}} \mathfrak{h}_{\vec{J}} (U_{\vec{J}}). 
\]

Let $\calI_{n,k} = \{ \{ i_1,\dots,i_k \} : 1 \le i_1,\dots,i_k \le n \ \text{are distinct} \}$, and let $\{ V_{J} :J \in \calI_{n,k} \}$ be i.i.d.  $U[0,1]$ random variables independent of the $U$-variables. Conditionally on $\calU_{k-1}$, we have $\E[ \mathfrak{h}_{\vec{J}}(V_{\vec{J}}) \mid \calU_{k-1} ] = \E[ \mathfrak{h}_{\vec{J}}(U_{\vec{J}}) \mid \calU_{k-1} ] = 0$ (with $V_{\vec{J}} = (V_{J_1},\dots,V_{J_{m_k}}))$ for $\vec{J} \in \calJ_{n,k}$, so that by Jensen's inequality, 
\[
\begin{split}
\E \left [ \left \|  \sum_{\vec{J} \in \calJ_{n,k}} \mathfrak{h}_{\vec{J}} (U_{\vec{J}}) \right \|_{\infty}^{q} \mid \calU_{k-1} \right] &= \E \left [ \left \|  \sum_{\vec{J} \in \calJ_{n,k}} \{ \mathfrak{h}_{\vec{J}} (U_{\vec{J}}) - \E[\mathfrak{h}_{\vec{J}}(V_{\vec{J}}) \mid \calU_{k-1}] \} \right \|_{\infty}^{q} \mid \calU_{k-1} \right] \\
&\le \E \left [ \left \|  \sum_{\vec{J} \in \calJ_{n,k}} \{ \mathfrak{h}_{\vec{J}} (U_{\vec{J}}) - \mathfrak{h}_{\vec{J}}(V_{\vec{J}}) \} \right \|_{\infty}^{q} \mid \calU_{k-1} \right].
\end{split}
\]
Conditionally on $\calU_{k-1}$, let $\overline{\mathfrak{h}}_{\vec{J}}(W_{\vec{J}}) = \mathfrak{h}_{\vec{J}} (U_{\vec{J}}) - \mathfrak{h}_{\vec{J}}(V_{\vec{J}}) $ with $W_{\vec{J}} = (W_{J_1},\dots,W_{J_{m_k}}) = ((U_{J_1},V_{J_1}),\dots,(U_{J_{m_k}},V_{J_{m_k}}))$. \
Conditionally on $\calU_{k-1}$, $\sum_{\vec{J} \in \calJ_{n,k}} \overline{\mathfrak{h}}_{\vec{J}}(W_{\vec{J}})$ can be seen as a $U$-statistic with index-dependent kernels by adding zero kernels. Namely, if we extend $\overline{\mathfrak{h}}_{J_1,\dots,J_{m_k}}$ as 
\[
\overline{\mathfrak{h}}^{(e)}_{J_1,\dots,J_{m_k}} (w_1,\dots,w_{m_{k}}) = 
\begin{cases}
\overline{\mathfrak{h}}_{J_1,\dots,J_{m_k}} (w_1,\dots,w_{m_{k}}) & \text{if $(J_1,\dots,J_{m_k}) \in \calJ_{n,k}$} \\
0 & \text{if $(J_1,\dots,J_{m_k}) \not \in  \calJ_{n,k}$}
\end{cases}
\]
for all distinct $J_1,\dots,J_{m_k} \in \calI_{n,k}$, then we see that 
\[
\sum_{\vec{J} \in \calJ_{n,k}} \overline{\mathfrak{h}}_{\vec{J}}(W_{\vec{J}}) = \sum_{\substack{J_1,\dots,J_{m_k} \in \calI_{n,k} \\ J_1,\dots,J_{m_k}:  \text{distinct}}} \overline{\mathfrak{h}}^{(e)}_{J_1,\dots,J_{m_k}}(W_{J_1},\dots,W_{J_{m_k}}).
\]
Thus, by the first part of Theorem 3.1.1 in \cite{delaPenaGine1999} (after introducing a proper ordering in $\calI_{n,k}$), we have 
\[
\begin{split}
\E \left [ \left \|  \sum_{\vec{J} \in \calJ_{n,k}}\overline{\mathfrak{h}}_{\vec{J}}(W_{\vec{J}})  \right \|_{\infty}^{q} \mid \calU_{k-1} \right] &=\E \left [ \left \|   \sum_{\substack{J_1,\dots,J_{m_k} \in \calI_{n,k} \\ J_1,\dots,J_{m_k}:  \text{distinct}}} \overline{\mathfrak{h}}^{(e)}_{J_1,\dots,J_{m_k}}(W_{J_1},\dots,W_{J_{m_k}}) \right \|_{\infty}^{q} \mid \calU_{k-1} \right] \\
& \lesssim \E \left [ \left \|   \sum_{\substack{J_1,\dots,J_{m_k} \in \calI_{n,k} \\ J_1,\dots,J_{m_k}:  \text{distinct}}} \overline{\mathfrak{h}}^{(e)}_{J_1,\dots,J_{m_k}}(W^1_{J_1},\dots,W^{m_k}_{J_{m_k}}) \right \|_{\infty}^{q} \mid \calU_{k-1} \right] \\
&= \E \left [ \left \|  \sum_{\vec{J} \in \calJ_{n,k}}\overline{\mathfrak{h}}_{\vec{J}}(W_{J_1}^{1},\dots,W_{J_{m_k}}^{m_k})  \right \|_{\infty}^{q} \mid \calU_{k-1} \right],
\end{split}
\]
where $\{ W_{J}^1 \}_{J \in \calI_{n,k}},\dots, \{ W_{J}^{m_k} \}_{J \in \calI_{n,k}}$ are independent copies of $\{ W_{J} \}_{J \in \calI_{n,k}}$ (with $W_{J} = (U_J,V_J)$) independent of $\calU_{k-1}$. We note here that kernels $\overline{\mathfrak{h}}_{\vec{J}}$ need not be symmetric in the sense of (3.1.2) in \cite{delaPenaGine1999}, but the first part of Theorem 3.1.1 in \cite{delaPenaGine1999} does not require the symmetry of kernels.

Decompose $\mathfrak{h}_{\vec{J}} (U_{J_1}^1,\dots,U_{J_{m_k}}^{m_k}) - \mathfrak{h}_{\vec{J}} (V_{J_1}^1,\dots,V_{J_{m_k}}^{m_k})$ as the following telescoping sum:
\[
\begin{split}
&\mathfrak{h}_{\vec{J}} (U_{J_1}^1,\dots,U_{J_{m_k}}^{m_k}) - \mathfrak{h}_{\vec{J}} (V_{J_1}^1,\dots,V_{J_{m_k}}^{m_k}) \\
&= 
\mathfrak{h}_{\vec{J}} (U_{J_1}^1,\dots,U_{J_{m_k}}^{m_k}) - \mathfrak{h}_{\vec{J}} (V_{J_1}^1,U_{J_2}^2, \dots,U_{J_{m_k}}^{m_k}) \\
&\quad +\mathfrak{h}_{\vec{J}} (V_{J_1}^1,U_{J_2}^2, \dots,U_{J_{m_k}}^{m_k}) - \mathfrak{h}_{\vec{J}} (V_{J_1}^1,V_{J_2}^2, U_{J_3}^3, \dots,U_{J_{m_k}}^{m_k}) \\
&\quad \quad \vdots \\
&\quad + \mathfrak{h}_{\vec{J}} (V_{J_1}^1,\dots,V_{J_{m_k-1}}^{m_k-1},U_{J_{m_k}}^{m_k}) - \mathfrak{h}_{\vec{J}} (V_{J_1}^1, \dots,V_{J_{m_k}}^{m_k}).
\end{split}
\]
Accordingly, we have
\begin{equation}
\label{eq: telescoping sum}
\begin{split}
&\E \left [ \left \|  \sum_{\vec{J} \in \calJ_{n,k}} \mathfrak{h}_{\vec{J}} (U_{\vec{J}}) \right \|_{\infty}^{q} \mid \calU_{k-1} \right]  \\
&\lesssim \E \left [ \left \|  \sum_{\vec{J} \in \calJ_{n,k}} \{ \mathfrak{h}_{\vec{J}} (U_{J_1}^1,\dots,U_{J_{m_k}}^{m_k}) - \mathfrak{h}_{\vec{J}} (V_{J_1}^1,U_{J_2}^2, \dots,U_{J_{m_k}}^{m_k})  \}  \right \|_{\infty}^{q} \mid \calU_{k-1} \right]  \\
&\quad + \E \left [ \left \|  \sum_{\vec{J} \in \calJ_{n,k}} \{ \mathfrak{h}_{\vec{J}} (V_{J_1}^1,U_{J_2}^2, \dots,U_{J_{m_k}}^{m_k}) - \mathfrak{h}_{\vec{J}} (V_{J_1}^1,V_{J_2}^2, U_{J_3}^3, \dots,U_{J_{m_k}}^{m_k}) \}  \right \|_{\infty}^{q} \mid \calU_{k-1}
 \right] \\
 &\quad \quad \vdots \\
 &\quad + \E \left [ \left \|  \sum_{\vec{J} \in \calJ_{n,k}} \{  \mathfrak{h}_{\vec{J}} (V_{J_1}^1,\dots,V_{J_{m_k-1}}^{m_k-1},U_{J_{m_k}}^{m_k}) - \mathfrak{h}_{\vec{J}} (V_{J_1}^1, \dots,V_{J_{m_k}}^{m_k}) \}  \right \|_{\infty}^{q} \mid \calU_{k-1} \right ].
\end{split}
\end{equation}
We note that there are $m_k = \binom{K}{k}$ terms on the right-hand side. 
We will focus on bounding the first term on the right-hand side, since bounding other terms is similar.  Observe that 
\[
\begin{split}
&\sum_{\vec{J} \in \calJ_{n,k}} \{ \mathfrak{h}_{\vec{J}} (U_{J_1}^1,\dots,U_{J_{m_k}}^{m_k}) - \mathfrak{h}_{\vec{J}} (V_{J_1}^1,U_{J_2}^2, \dots,U_{J_{m_k}}^{m_k})  \} \\
&= \sum_{J_1} \left [ \sum_{J_2,\dots,J_{m_k}}\{ \mathfrak{h}_{\vec{J}} (U_{J_1}^1,\dots,U_{J_{m_k}}^{m_k}) - \mathfrak{h}_{\vec{J}} (V_{J_1}^1,U_{J_2}^2, \dots,U_{J_{m_k}}^{m_k})  \} \right ].
\end{split}
\]
Here the summation $\sum_{J_1} \sum_{J_2,\dots,J_{m_k}}$ is understood as 
\[
\sum_{\substack{J_1: \exists (J_2,\dots,J_{m_k}) \\ \text{such that} \ (J_1,J_2,\dots,J_{m_k}) \in \calJ_{n,k}}} \sum_{(J_2,\dots,J_{m_k}): (J_1,J_2,\dots,J_{m_k}) \in \calJ_{n,k}}.
\]
Let $\calW_{-1} = \{ W_{J}^{\ell} : J \in \calI_{n,k}, \ell =2,\dots,m_k \}$. Since $\{ W_{J}^1 \}_{J \in \calI_{n,k}},\dots, \{ W_{J}^{m_k} \}_{J \in \calI_{n,k}}$ are independent, conditionally on $\calW_{-1}$ and $\mathcal U_{k-1}$, the terms 
\[
\sum_{J_2,\dots,J_{m_k}}\{ \mathfrak{h}_{\vec{J}} (U_{J_1}^1,\dots,U_{J_{m_k}}^{m_k}) - \mathfrak{h}_{\vec{J}} (V_{J_1}^1,U_{J_2}^2, \dots,U_{J_{m_k}}^{m_k})  \} 
 \]
are independent across different $J_1$'s. Further, they have conditional mean $0$ given $\calW_{-1}$ and $\mathcal U_{k-1}$ (as $U_{J_1}^1 \stackrel{d}{=} V_{J_1}^1$ given $\calW_{-1}$ and $\mathcal U_{k-1}$). Thus, applying the symmetrization inequality (\cite{vdVW1996}, Lemma 2.3.6), we have 
\[
\begin{split}
&\E \left [ \left \| \sum_{J_1} \left [ \sum_{J_2,\dots,J_{m_k}}\{ \mathfrak{h}_{\vec{J}} (U_{J_1}^1,\dots,U_{J_{m_k}}^{m_k}) - \mathfrak{h}_{\vec{J}} (V_{J_1}^1,U_{J_2}^2, \dots,U_{J_{m_k}}^{m_k})  \} \right ] \right \|_{\infty}^{q}\mid \calW_{-1},\mathcal U_{k-1} \right] \\
&\lesssim \E \left [ \left \| \sum_{J_1} \epsilon_{J_1} \left [ \sum_{J_2,\dots,J_{m_k}}\{ \mathfrak{h}_{\vec{J}} (U_{J_1}^1,\dots,U_{J_{m_k}}^{m_k}) - \mathfrak{h}_{\vec{J}} (V_{J_1}^1,U_{J_2}^2, \dots,U_{J_{m_k}}^{m_k})  \} \right ] \right \|_{\infty}^{q}  \mid \calW_{-1},\mathcal U_{k-1}\right],
\end{split}
\]
where $\{ \epsilon_{J} : J \in \calI_{n,k} \}$ are independent Rademacher random variables independent of everything else. By Fubini, together with the elementary inequality $(a+b)^q \le 2^{q-1}(a^q+b^q)$ for $a,b \ge 0$ and the fact that $\{ W_{J}^1 \}_{J \in \calI_{n,k}},\dots, \{ W_{J}^{m_k} \}_{J \in \calI_{n,k}}$ are independent and $U_{J_1}^1 \stackrel{d}{=} V_{J_1}^1$, we have
\[
\begin{split}
&\E \left [ \left \| \sum_{J_1} \left [ \sum_{J_2,\dots,J_{m_k}}\{ \mathfrak{h}_{\vec{J}} (U_{J_1}^1,\dots,U_{J_{m_k}}^{m_k}) - \mathfrak{h}_{\vec{J}} (V_{J_1}^1,U_{J_2}^2, \dots,U_{J_{m_k}}^{m_k})  \} \right ] \right \|_{\infty}^{q}\mid \calU_{k-1} \right] \\
&\lesssim \E \left [ \left \| \sum_{J_1} \epsilon_{J_1} \left [ \sum_{J_2,\dots,J_{m_k}} \mathfrak{h}_{\vec{J}} (U_{J_1}^1,\dots,U_{J_{m_k}}^{m_k})  \right ] \right \|_{\infty}^{q}  \mid  \calU_{k-1}\right].
\end{split}
\]

Conditionally on the $U$-variables, each variable $ \sum_{J_1} \epsilon_{J_1} \left [ \sum_{J_2,\dots,J_{m_k}} \mathfrak{h}_{\vec{J}}^\ell (U_{J_1}^1,\dots,U_{J_{m_k}}^{m_k}) \right]$ (with $\mathfrak{h}_{\vec{J}} = (\mathfrak{h}_{\vec{J}}^1,\dots,\mathfrak{h}_{\vec{J}}^p)$) is a weighted sum of independent Rademacher random variables and thus sub-Gaussian whose (conditional) $\psi_2$-norm is 
\[
\lesssim \sqrt{\sum_{J_1} \left [  \sum_{J_2,\dots,J_{m_k}} \mathfrak{h}_{\vec{J}}^{\ell} (U_{J_1}^1,\dots,U_{J_{m_k}}^{m_k}) \right ]^2} \le \sqrt{\sum_{J_1}\left \|  \sum_{J_2,\dots,J_{m_k}} \mathfrak{h}_{\vec{J}} (U_{J_1}^1,\dots,U_{J_{m_k}}^{m_k})   \right \|_{\infty}^2}
\]
by e.g. Corollary 3.2.6 in \cite{delaPenaGine1999}.
Applying Lemma 2.2.2 in \cite{vdVW1996} and noting that $\E[|\xi|^q] \lesssim \| \xi \|_{\psi_2}^q$ by Lemma \ref{lem: Orlicz norms}, we have 
\[
\begin{split}
&\E \left [ \left \| \sum_{J_1} \epsilon_{J_1} \left [ \sum_{J_2,\dots,J_{m_k}} \mathfrak{h}_{\vec{J}} (U_{J_1}^1,\dots,U_{J_{m_k}}^{m_k})  \right ] \right \|_{\infty}^{q}  \mid \calU_{k-1}\right] \\
&\lesssim (\log p)^{q/2} \E\left [ \left ( \sum_{J_1}\left \|  \sum_{J_2,\dots,J_{m_k}} \mathfrak{h}_{\vec{J}} (U_{J_1}^1,\dots,U_{J_{m_k}}^{m_k})   \right \|_{\infty}^2 \right)^{q/2} \mid \calU_{k-1} \right ].
\end{split}
\]
Observe that given $J_1$, the number of $(J_2,\dots,J_{m_k})$ such that $(J_1,J_2,\dots,J_{m_k}) \in \calJ_{n,k}$ is 
\[
\binom{n-k}{K-k} = O(n^{K-k}).
\]
To see this, observe that $J = (J_1,\dots,J_{m_k}) \in \calJ_{n,k}$ is of the form $J = (\{ \i \odot \e \}^+)_{\e \in \calE_k}$ for some $(i_1,\dots,i_K)$ such that $1 \le i_1 < \dots < i_K \le n$. Fixing $J_1$ corresponds to fixing $k$ elements of $i_1,\dots,i_K$, so the number of possible $(J_2,\dots,J_{m_k})$ coincides with the number of ways to choose remaining $K-k$ elements from $n-k$ integers.

Thus, by the Cauchy-Schwarz inequality, we have 
\[
\sum_{J_1}\left \|  \sum_{J_2,\dots,J_{m_k}} \mathfrak{h}_{\vec{J}} (U_{J_1}^1,\dots,U_{J_{m_k}}^{m_k})   \right \|_{\infty}^2 \lesssim  n^{K-k} \sum_{\vec{J}} \|  \mathfrak{h}_{\vec{J}} (U_{J_1}^1,\dots,U_{J_{m_k}}^{m_k}) \|_{\infty}^{2}. 
\]
Combining the fact that the size of $\calJ_{n,k}$ is $\binom{n}{K} = O(n^{K})$, we have 
\[
\begin{split}
&\E \left [ \left \| \sum_{J_1} \left [ \sum_{J_2,\dots,J_{m_k}}\{ \mathfrak{h}_{\vec{J}} (U_{J_1}^1,\dots,U_{J_{m_k}}^{m_k}) - \mathfrak{h}_{\vec{J}} (V_{J_1}^1,U_{J_2}^2, \dots,U_{J_{m_k}}^{m_k})  \} \right ] \right \|_{\infty}^{q} \mid \calU_{k-1} \right]  \\
&\lesssim  n^{(K-k/2)q}(\log p)^{q/2} \E\left [ \left( | \calJ_{n,k} |^{-1} \sum_{\vec{J}} \|  \mathfrak{h}_{\vec{J}} (U_{J_1}^1,\dots,U_{J_{m_k}}^{m_k})\|_{\infty}^{2} \right)^{q/2} \mid \calU_{k-1} \right].
\end{split}
\]
Using Jensen's inequality, we have
\[
\begin{split}
&\E\left [ \left( | \calJ_{n,k} |^{-1} \sum_{\vec{J}} \|  \mathfrak{h}_{\vec{J}} (U_{J_1}^1,\dots,U_{J_{m_k}}^{m_k})\|_{\infty}^{2} \right)^{q/2} \mid \calU_{k-1} \right] \\
&\le
\begin{cases}
(|\calJ_{n,k} |^{-1} \sum_{\vec{J}}\E[  \|  \mathfrak{h}_{\vec{J}} (U_{J_1}^1,\dots,U_{J_{m_k}}^{m_k})\|_{\infty}^{2} \mid \calU_{k-1}])^{q/2} & \text{if $q \le 2$} \\
 |\calJ_{n,k} |^{-1} \sum_{\vec{J}} \E[  \|  \mathfrak{h}_{\vec{J}} (U_{J_1}^1,\dots,U_{J_{m_k}}^{m_k})\|_{\infty}^{q} \mid \calU_{k-1}] & \text{if $q > 2$}
\end{cases}
.
\end{split}
\]
Since, conditionally on $\calU_{k-1}$, $\mathfrak{h}_{\vec{J}} (U_{J_1}^1,\dots,U_{J_{m_k}}^{m_k}) \stackrel{d}{=} \mathfrak{h}_{\vec{J}}(U_{\vec{J}})$, combining Fubini and Jensen, and the definition of $\mathfrak{h}_J$, we conclude that 
\begin{equation}
\begin{split}
&\E \left [ \left \| \sum_{J_1} \left [ \sum_{J_2,\dots,J_{m_k}}\{ \mathfrak{h}_{\vec{J}} (U_{J_1}^1,\dots,U_{J_{m_k}}^{m_k}) - \mathfrak{h}_{\vec{J}} (V_{J_1}^1,U_{J_2}^2, \dots,U_{J_{m_k}}^{m_k})  \} \right ] \right \|_{\infty}^{q}  \right] \\
&\lesssim 
\begin{cases}
n^{(K-k/2)q} (\log p)^{q/2} (\E[\| \bX_{(1,\dots,K)} \|_{\infty}^{2}])^{q/2} & \text{if $q \le 2$} \\ 
n^{(K-k/2)q} (\log p)^{q/2} \E[\| \bX_{(1,\dots,K)} \|_{\infty}^{q}] & \text{if $q > 2$}
\end{cases}
.
\label{eq: final bound}
\end{split}
\end{equation}
Indeed, if $q \le 2$, then 
\[
\begin{split}
&\E \left [ \left(|\calJ_{n,k} |^{-1} \sum_{\vec{J}}\E[  \|  \mathfrak{h}_{\vec{J}} (U_{\vec{J}})\|_{\infty}^{2} \mid \calU_{k-1}]\right)^{q/2}\right ]\\
& \le \left ( |\calJ_{n,k} |^{-1} \sum_{\vec{J}}\E[  \|  \mathfrak{h}_{\vec{J}}(U_{\vec{J}})\|_{\infty}^{2}] \right )^{q/2} \\
&\lesssim \left ( n^{-K} \sum_{1 \le i_1 < \cdots < i_{K} \le n}\E\left [  \| \E[\bX_{\i}\mid (U_{\{\i\odot \e\}^+})_{\e \in \cup_{r=1}^k \calE_r} ]-\E[\bX_{\i}\mid (U_{\{\i\odot \e\}^+})_{\e \in \cup_{r=1}^{k-1}\calE_r} ]\|_{\infty}^{2} \right] \right )^{q/2} \\
&\lesssim \left ( n^{-K} \sum_{1 \le i_1 < \cdots < i_{K} \le n}\E\left [  \| \E[\bX_{\i}\mid (U_{\{\i\odot \e\}^+})_{\e \in \cup_{r=1}^k \calE_r} ]\|_{\infty}^{2} + \| \E[\bX_{\i}\mid (U_{\{\i\odot \e\}^+})_{\e \in \cup_{r=1}^{k-1}\calE_r} ]\|_{\infty}^{2} \right] \right )^{q/2} \\
&\lesssim \left ( n^{-K}\sum_{1 \le i_1 < \cdots < i_{K} \le n}\E[  \| \bX_{\i} \|_{\infty}^2 ] \right )^{q/2} \lesssim  (\E[\| \bX_{(1,\dots,K)} \|_{\infty}^{2}])^{q/2}. 
\end{split}
\]
Likewise, if $q > 2$, then 
\[
\E\left [ \E[  \|  \mathfrak{h}_{\vec{J}} (U_{\vec{J}})\|_{\infty}^{q} \mid \calU_{k-1}] \right] \lesssim  \E[\| \bX_{(1,\dots,K)} \|_{\infty}^{q}]. 
\]
Thus we obtain (\ref{eq: final bound}). 
Similar bounds hold for other terms in the decomposition (\ref{eq: telescoping sum}). This completes the proof. 
\end{proof}
}

\begin{remark}[Comparison with \cite{DDG2019}]
\label{rem: comparison with DDG 2}
Lemma A.1 in \cite{DDG2019} derives a symmetrization inequality for the empirical process of a jointly exchangeable array. Essentially, the same comparison made in Remark \ref{rem: comparison with DDG} applies to the comparison of their Lemma A.1 with the maximal inequalities developed in this section. Lemma S3 in \cite{DDG2019} covers the degenerate case but focuses only on the $K=2$ case. As seen in the proof of Lemma \ref{lem: max_ineq_poly2} above, handling the degenerate components in $K>2$ cases is highly nontrivial \red{(however, it would be fair to point out that the proof of Lemma \ref{lem: max_ineq_poly2} is partly inspired by the proof of Lemma A.1 in \cite{DDG2019})}. 
\end{remark}

%%%%%%%%%%%%%%%%%%%%%%%%%%%%%%%%%%%%%%%%%%%%%%%%%%%%%%%%%%%%%%%%%%%%%%%%%%%%%%%
\section{Proofs for Section \ref{sec:polyadic_data}}\label{sec:proofs_polyadic_data}

\subsection{Proof of Theorem \ref{thm: high-d CLT polyadic}}\label{sec:proof of theorem HDCLT polyadic}
Given Lemmas \ref{lem: max_ineq_poly1} and \ref{lem: max_ineq_poly2}, the proof is almost identical to that of Theorem \ref{thm: high-d CLT}. We omit the details for brevity. 
\qed

\subsection{Proof of Theorem \ref{thm:bootstrap validity polyadic}}\label{sec:proof of theorem polyadic mutliplier bootstrap}

We only prove the proposition under Case (i). The proof for Case (ii) is similar.

\underline{Step 1}. Define
\begin{align*}
\hat \Delta_{W,1}=\max_{1\le \ell \le p} \frac{1}{n}\sumj (\hat W_{j}^\ell - W_{j}^\ell  )^2.
\end{align*}
We will show that 
\[
\Prob \left ( \overline{\sigma}^2 \hat{\Delta}_{W,1}\log^{4}p > n^{-\zeta+1/\nu} \right) \lesssim n^{-1}.
\]
Here the notation $\lesssim$ means that the left-hand side is bounded by the right-hand side up to a constant that depends only on $\nu,\underline{\sigma},K$, and $C_1$. 

Recall that $\bW_j$ can can be written as
\begin{align*}
\bW_j=\E\left[\frac{(n-K)!}{(n-1)!}\sum_{k=1}^K\sum_{\i \in I_{n,K}: i_k=j} \bm{X}_{\i}~\Big|~U_{j}\right].
\end{align*}
We have 
\begin{align*}
\hat \Delta_{W,1}&=\max_{1\le \ell \le p} \frac{1}{n}\sumj (\hat W_{j}^\ell -W_{j}^\ell  )^2
\le
 \frac{1}{n}\sumj \|\hat \bW_{j}- \bW_{j} \|_\infty^2\\
&\lesssim
 \sum_{k=1}^K\frac{1}{n}\sumj \left\|\frac{(n-K)!}{(n-1)!}\sum_{\i \in I_{n,K}:i_k=j}( \bX_{\i} - \E[\bX_{\i} \mid U_{j}])\right\|^2_\infty.
\end{align*}
Consider the $k=1$ term. 
Pick any $j  \in \N$. Let $I_{\infty,K-1}^{-j} = \{ (i_2,\dots,i_K) \in (\N \setminus \{ j \})^{K-1} : \text{$i_2,\dots,i_K$ are distinct} \}$.
Given $U_{j}$, for each $\bm{i}_{-1}=(i_2,\dots,i_K) \in I_{\infty,K}^{-j}$ and $\e \in \{0,1\}^{K-1}$, define the vector
\begin{align*}
V_{\{\bm{i}_{-1}\odot \e\}^{+}}=(U_{\{ \bm{i}_{-1} \odot \e\}^{+}},U_{\{(j,\bm{i}_{-1} \odot \e)\}^{+}}).
\end{align*}
With this notation, we can rewrite $\bX_{\i}$ with $\i = (j,\i_{-1})$ as 
\begin{align*}
\bm{X}_{\bm{i}} =\mathfrak{g}\big(U_{j},(V_{\{\bm{i}_{-1}\odot \e\}^{+}})_{\e\in \{0,1\}^{K-1}\setminus \{ \bm{0} \}}\big).
\end{align*}
From this expression, we see that, conditionally on $U_{j}$, the array $(\bX_{(j,\i_{-1})})_{\i_{-1} \in I_{\infty,K-1}^{-j}}$ is 
jointly exchangeable with mean vector $\E[\bX_{\i} \mid U_{j}]$. Applying Lemmas \ref{lem: max_ineq_poly1} and \ref{lem: max_ineq_poly2} conditionally on $U_j$ (the fact that $U$-variables are uniform on $(0,1)$ is not crucial in the proofs), we have 
\[
\begin{split}
&\E \left [ \left\|\frac{(n-K)!}{(n-1)!}\sum_{\i \in I_{n,K}:i_k=j}( \bX_{\i} - \E[\bX_{\i} \mid U_{j}])\right\|^{2\nu}_\infty \mid U_{j} \right ] \\
&\lesssim \underbrace{\left (\sum_{k=1}^{K-1} n^{-k/2} (\log p)^{k/2} \right)^{2\nu} }_{\lesssim (n^{-1}\log p)^{\nu}} \E[ \| \bX_{(j,\i_{-1})} \|_{\infty}^{2\nu} \mid U_{j}],
\end{split}
\]
where $\i_{-1} \in I_{\infty,K-1}^{-j}$ is arbitrary. By Fubini's theorem, the expectation of the left-hand side can be bounded as 
\[
\lesssim (n^{-1}\log p)^{\nu} \E[ \| \bX_{(j,\i_{-1})} \|_{\infty}^{2\nu}] \lesssim (n^{-1}D_n^2 \log^3 p)^{\nu}.
\]
Similar bounds hold for other $k$. Conclude that $\E[ (\overline{\sigma}^2 \hat{\Delta}_{W,1}\log^{4}p)^{\nu}] \lesssim n^{-\zeta \nu}$ under our assumption. Together with Markov's inequality, we obtain the conclusion of Step 2. 

\underline{Step 2}. 
Conditionally on $(\bX_{\i})_{\i \in I_{n,K}}$, we have $\sqrt{n} \bS_{n}^{MB} \sim N(\bm{0},\hat{\Sigma})$,
where  
\[
\hat{\Sigma} = \frac{1}{n} \sumj 
(\hat{\bW}_{j} - K\bm{S}_{n})(\hat{\bW}_{j} - K\bm{S}_{n})^T
\]
As in the proof of Theorem \ref{thm:bootsrap_validity}, the desired result follows from bounding $\hat{\Delta}_{W}=\| \hat{\Sigma} - \Sigma \|_{\infty}$. 

We first note that 
\begin{align*}
\hat{\Delta}_{W} =
\max_{1\le \ell,\ell' \le p}\left| \frac{1}{n} \sumj (\hat W_j^{\ell}-K S_n^{\ell} )(\hat W_j^{\ell'}-K S_n^{\ell'} ) -\E[W_1^{\ell} W_1^{\ell'}] \right| .
\end{align*}
For every $\ell, \ell' \in \{1,\dots,p\}$, %$\frac{1}{n}\sum_{j=1}^n \hat \bW_j=K\bS_n$, one has
\begin{align*}
&\frac{1}{n} \sumj (\hat W_j^{\ell}-K S_n^{\ell} )(\hat W_j^{\ell'}-K S_n^{\ell'} )
=
\frac{1}{n} \sumj \hat W_j^{\ell}\hat W_j^{\ell'}- K^2 S_n^{\ell}  S_n^{\ell'}  \\
&\quad =
\frac{1}{n}\sumj (\hat W_j^{\ell}- W_j^{\ell} )(\hat W_j^{\ell'} - W_j^{\ell'} )
+
\frac{1}{n} \sumj  (\hat W_j^{\ell}- W_j^{\ell})W_j^{\ell'} \\
&\quad \quad +
\frac{1}{n} \sumj   W_j^{\ell}(\hat W_j^{\ell'}- W_j^{\ell'})+ \frac{1}{n} \sumk W_j^{\ell} W_j^{\ell'} - K^2S_n^{\ell} S_n^{\ell'}. 
\end{align*}
Using the Cauchy-Schwarz inequality, we have 
\begin{align*}
\hat \Delta_W 
\le&
\underbrace{ \max_{1\le \ell\le p}\frac{1}{n} \sumj (\hat W_j^{\ell} -W_j^{\ell} )^2 }_{= \Delta_{W,1}} 
+ 
2\hat\Delta_{W,1}^{1/2} \sqrt{\max_{1\le \ell\le p} \frac{1}{n}\sumj |W_j^{\ell}|^2}\\
&+
\underbrace{ 
\max_{1\le \ell,\ell' \le p} \left| \frac{1}{n} \sumj (W_j^{\ell} W_j^{\ell'}  - \E[ W_1^{\ell} W_1^{\ell'} ] ) \right|
 }_{= \hat \Delta_{W,2}} + K^2 \max_{1\le \ell \le p} |S_n^{\ell}|^2.
\end{align*}
For the second term on the right-hand side, we have
\begin{align*}
\frac{1}{n}\sumk |W_j^{\ell}|^2 
\le 
\E\left[\frac{1}{n} \sumj |W_j^\ell|^2\right] +
\frac{1}{n}\sumj ( |W_j^\ell|^2-\E[|W_1^\ell|^2])
\le 
\overline \sigma^2 +\hat \Delta_{W,2}.
\end{align*}
Further, since $KS_n^\ell = n^{-1} \sumj(\hat W_j^\ell - W_j^\ell )+n^{-1}\sumj W_j^\ell $, we have
\begin{align*}
K^2\max_{1\le \ell \le p}|S_n^\ell|^2 \le 2\hat \Delta_{W,1} +2\hat \Delta_{W,3}^2,
\end{align*}
where $\hat \Delta_{W,3}=\max_{1\le \ell \le p}|n^{-1}\sumj W_j^\ell|$. 
Conclude that
\begin{align*}
\hat \Delta_{W} \lesssim \hat \Delta_{W,1} + \overline{\sigma} \hat \Delta_{W,1}^{1/2} + \hat \Delta_{W,2} + \hat \Delta_{W,3}
\end{align*}
up to a universal constant. 
The rest is completely analogous to the latter part of the proof of Theorem \ref{thm:bootsrap_validity}. We omit the details for brevity. \qed

%\subsection{Proof of Corollary \ref{cor:normalized_polyadic_MB}}

%%%%%%%%%%%%%%%%%%%%%%%%%%%%%%%%%%%%%%%%%%%%%%%%%%%%%%%%%%%%%%%%%%%%%%%%%%%%%%%%%%%%%%%%%%%%

%%%%%%%%%%%%%%%%%%%%%%%%%%%%%%%%%%%%%%%%%%%%%

\newcommand{\diam}{\text{diam}}
\section{Proof for Section \ref{sec:applications}}\label{sec:proofs_for_applications}

\subsection{Proof of Proposition \ref{proposition:kernel_density_mixture}} 
%%%%%%%%%%%%%%%%%%%%%%%%%%%%%%%%%%%%%%%%%%%%%%%%
	In this proof, the notation $\lesssim$ means that the left-hand side is bounded by the right-hand side up to a constant independent of $n$. Also, $\sum_{j\ne i} $ is understood as $\sum_{j\in \{1,\dots,n\}\setminus\{i\}}$.  We will first consider the non-normalized statistic.
	
Define the infeasible sample mean and its H\'ajek projection,
\begin{align*}
 &\bS_n=\frac{1}{n(n-1)} \sum_{1\le i\ne j\le n}\bX_{ij},\quad \frac{1}{n}\sum_{i=1}^n \bW_i, \text{ where $\overline b_h(\cdot)=a \overline f_h(\cdot)$,}\\ 
 &\bX_{ij}=\left(\left\{ \frac{K_h(y_\ell-Y_{ij})}{a} -\frac{\overline b_h(y_\ell)}{a^2}\right\}\1(Y_{ij}\ne 0)\right)_{1\le \ell\le p},
 \\
  &\bW_i=\left(2\left\{\frac{\overline f_h (y_\ell\mid U_i)}{a}-\frac{\overline b_h(y_\ell)}{a^2} \right\}\Prob(Y_{ij} \ne 0 \mid U_i)\right)_{1\le \ell \le p}.
\end{align*}

We will show that 
\begin{align*}
\sqrt{n} (\hat f(y_\ell)-\overline f_h(y_\ell))=\sqrt{n} S_n^\ell +O_P\left(\frac{\log p}{n^{1/2}}\right)
\end{align*}
uniformly over $\ell \in \{ 1,\dots,p \}$.
Under Conditions (ii) and (iv) of the proposition, for any $y\in \{y_1,...,y_p\}$,
\begin{align*}
\left | \E[K_h(y- Y_{ij})\1(Y_{ij}\ne 0)\mid U_i] \right |&=\left | \Prob(Y_{ij}\ne 0\mid U_i) \int K(z)f_{Y_{12} \mid U_1,Y_{12}\ne 0}( y+zh \mid U_i) dz \right | \\
&\le M \int |K(z)| dz,\\
| \E[K_h(y- Y_{ij})\1(Y_{ij}\ne 0)\mid U_i,U_j] | &\le h^{-1} \| K \|_{\infty},\\
E[\hat b(y)]&=\E[K_h(y- Y_{ij})\1(Y_{ij}\ne 0)]\\
&=a \int K_h(y-z)f(z)dz =\overline b_h(y),\\
\E\left[\max_{1\le \ell\le p}\left|K_h(y_\ell- Y_{ij})\1(Y_{ij}\ne 0)\right|^2\right] &\le h^{-2} \| K \|_{\infty}^2.
\end{align*}	
Using these results, Condition (v) of the proposition, the Hoeffding-type decomposition (\ref{eq: H-decomp polyadic}), Lemmas \ref{lem: max_ineq_poly1} and \ref{lem: max_ineq_poly2} together give $|\hat a-a|=O_P(n^{-1/2})$ and
\[
\max_{1\le \ell \le p } |\hat b(y_\ell) -\overline  b_h(y_\ell)| 
=O_P\left(\sqrt{n^{-1}\log p} + (nh)^{-1}\log p + (nh)^{-1} \log^{1/2} p \right)=O_P\left(\sqrt{n^{-1}\log p}\right). 
\]
Then, linearization of the estimator yields the following representation
\begin{align*}
&\sqrt{n} (\hat f(y_\ell)-\overline f_h(y_\ell))=\sqrt{n}\left[\frac{\hat b(y_\ell) - \overline b_h(y_\ell)}{a}-\frac{(\hat a -a)\overline b_h(y_\ell)}{a^2}\right]+ O_P\left(\frac{\log p}{n^{1/2}}\right)
\end{align*}
uniformly over $y\in\{y_1,\dots,y_p\} $.
Using the fact that $Y_{ij}=Y_{ji}$, we can rewrite the right-hand side as
\begin{align*}
\frac{\sqrt{n}}{n(n-1)} \sum_{1\le i\ne j\le n} X_{ij}^\ell+O_P\left(\frac{\log p}{n^{1/2}}\right)=\sqrt{n} S_n^\ell +O_P\left(\frac{\log p}{n^{1/2}}\right)
\end{align*}
uniformly over $\ell \in \{ 1,\dots,p \}$. 

Next, we apply Theorem \ref{thm:bootstrap validity polyadic} under Condition (\ref{eq:condition1polyadic}) and Remark \ref{remark:discussion_rates_polyadic} to the infeasible bootstrap  statistic
\begin{align*}
\bS_n^{MB}=\frac{1}{n}\sum_{i=1}^n \xi_i ( \hat\bW_i-2\bS_n),\quad \hat \bW_i= \frac{1}{n-1}\sum_{j\ne i} 2\bX_{ij}.
\end{align*}
To verify Condition (\ref{eq:condition1polyadic}) and the first part of Condition (\ref{eq:condition2polyadic}), observe that $\|\bX_{ij}\|_\infty \le 2a^{-1}h^{-1} \int |K(z)|dz = O(h^{-1})$. Thus, for $D_n=Ch^{-1}$ with some appropriate constant $C$, we have $\max_{1\le \ell \le p}\|X_{12}^\ell\|_{\psi_1}\le D_n$ and $\max_{1\le \ell\le p}\E[|W_{1}^\ell|^{2+\kappa}]\le D_n^\kappa$ for $\kappa=1,2$ for $n$ large enough. Condition (iii) of the proposition guarantees that $\min_{1\le \ell \le p}\E[|W_1^\ell|^2] \ge 4a^{-2} \sigma_{0}^2$. To verify the rate conditions in Remark \ref{remark:discussion_rates_polyadic}, note that $\overline \sigma^2=\max_{1\le \ell \le p}\E[|W_1^\ell|^2]=O(1)$ under Conditions (ii) and (iv) of the proposition.
Thus Condition (v) of the proposition implies that $D_n^{2}\log^7 (pn)=h^{-2}\log^7 (pn)=o(n)$. We have verified the conditions of Theorem \ref{thm:bootstrap validity polyadic} under Remark \ref{remark:discussion_rates_polyadic}. Theorem \ref{thm:bootstrap validity polyadic} then yields the distributional approximation 
\begin{align*}
\sup_{R\in\calR}\left|\Prob_{|\bX_{I_{n,2}}}\left(\sqrt{n}\bS_n^{MB}\in R\right) - \gamma_{\Sigma}(R)\right|=o_P(1),
\end{align*}
where $\gamma_{\Sigma} = N(\bm{0},\Sigma)$ and $\Sigma=\E[\bW_1\bW_1^T]$.

 Coming back to the statistic $\tilde \bS_n^{MB}$, we note that conditionally on $\bX_{I_{n,2}}$,
\begin{align*}
\sqrt{n}\tilde \bS_n^{MB}\sim N(\bm{0},\tilde \Sigma),\text{ where } \tilde \Sigma=\frac{1}{n}\sumi (\tilde\bW_i -2\tilde\bS_n)(\tilde \bW_i -2\tilde\bS_n)^T.
\end{align*}
We will show that
$
\|\tilde \Sigma - \hat\Sigma\|_\infty  = o_P((\log p)^{-2})
$,
where
 $\hat\Sigma=n^{-1}\sum_{i=1}^n(\hat\bW_i -2 \bS_n)(\hat\bW_i -2 \bS_n)^T$. In view of Lemma \ref{lem:Gaussian_comparison} and Step 2 in the proof of Theorem \ref{thm:bootstrap validity polyadic},  this claim and the distributional approximation of the infeasible bootstrap  statistic imply
 \begin{align*}
 \sup_{R\in\calR}\left|\Prob_{|\bX_{I_{n,2}}}\left(\sqrt{n}\tilde\bS_n^{MB}\in R\right) - \gamma_{\Sigma}(R)\right|=o_P(1),
 \end{align*} 
which in turn implies the first statement of the proposition for the non-normalized statistic. 

 Following similar calculation to those in Step 2 in the proof of Theorem \ref{thm:bootstrap validity polyadic}, %since $n^{-1}\sum_{i=1}^n\tilde \bW_i=2\tilde \bS_n$ and $n^{-1}\sum_{i=1}^n\hat \bW_i=2\bS_n$, 
 we have
 \begin{align*}
 &\max_{1\le \ell,\ell'\le p}\left|\frac{1}{n}\sumi \left\{(\tilde W_i^\ell -2\tilde S_n^\ell)(\tilde W_i^{\ell'} -2 \tilde S_n^{\ell'} )- ( \hat W_i^{\ell}  -2 S_n^{\ell} )(\hat W_i^{\ell'}  -2 S_n^{\ell'} )\right\}\right|\\
 &\quad \le\underbrace{\max_{1\le \ell,\ell'\le p}\left|\frac{1}{n}\sumi (\tilde W_i^\ell \tilde W_i^{\ell'}-\hat W_i^{\ell} \hat W_i^{\ell'} )\right|}_{=\Delta_I} +4\underbrace{\max_{1\le \ell,\ell'\le p}\left|(\tilde S_n^\ell \tilde S_n^{\ell'}-S_n^\ell S_n^{\ell'})\right|}_{=\Delta_{II}}.
 \end{align*}
We shall first bound $\Delta_{I}$. Since
\[
 \tilde W_i^\ell \tilde W_i^{\ell'}-\hat W_i^{\ell} \hat W_i^{\ell'} 
 = (\tilde W_i^\ell-\hat W_i^{\ell}) (\tilde W_i^{\ell'}-\hat W_i^{\ell'}) + (\tilde W_i^\ell-\hat W_i^{\ell}) \hat W_i^{\ell'} + \hat W_i^{\ell} (\tilde W_i^{\ell'}-\hat W_i^{\ell'}),
\]
 we have
 \begin{align*}
 \Delta_I\le \underbrace{\max_{1\le \ell\le p} \frac{1}{n}\sumi (\tilde W_i^\ell-\hat W_i^{\ell})^2}_{=\Delta_{III}} + 2\underbrace{ \Delta_{III}^{1/2} \sqrt{\max_{1\le \ell \le p}\frac{1}{n} \sumi |\hat W_i^\ell|^2 } }_{=\Delta_{IV}}.
 \end{align*}
Consider $\Delta_{III}$.  We see that, with probability $1-o(1)$,
 \begin{align*}
 |\tilde W_i^\ell -\hat W_i^\ell|^2=&
 \left|\frac{2}{n-1}\sum_{j\ne i}\left\{\frac{a-\hat a}{a\hat a}K_h(y_\ell-Y_{ij})+\frac{\hat a^2\overline b_h(y_\ell)-a^2\hat b(y_\ell)}{a^2\hat a^2}\right\}\1(Y_{ij}\ne 0)\right|^2\\
 \lesssim& (\hat a-a)^2\left|\frac{1}{n-1}\sum_{j\ne i}K_h(y_\ell-Y_{ij})\1(Y_{ij}\ne 0)\right|^2 + (\hat a-a)^2 \vee |\hat b(y_\ell)-\overline b_h(y_\ell)|^2
 \end{align*}
 up to a constant independent of $n$ and $i$ (the $o(1)$ term is uniform in $1 \le i \le n$). 
% \1(Y_{ij}\ne 0)
Decompose the first term on the right-hand side as
 \begin{align*}
 &\max_{1\le \ell \le p}\left|\frac{1}{n-1}\sum_{j\ne i}K_h(y_\ell-Y_{ij})\1(Y_{ij}\ne 0)\right|^2\\
&\le 2 \max_{1\le \ell \le p}\left|\frac{1}{n-1}\sum_{j\ne i}\left\{K_h(y_\ell-Y_{ij})\1(Y_{ij}\ne 0)-\E[K_h(y_\ell-Y_{ij})\1(Y_{ij}\ne 0)\mid U_i]\right\}\right|^2\\
 &\quad + 2\max_{1\le \ell \le p}\left|\E[K_h(y_\ell-Y_{ij})\1(Y_{ij}\ne 0)\mid U_i]\right|^2.
 \end{align*}
Conditionally on $U_i$, we apply Theorem 2.14.1 in \cite{vdVW1996} with $p=2$ to the first term on the right-hand side to deduce that 
\begin{align*}
&\E\left[\max_{1\le \ell \le p}\left|\frac{1}{n-1}\sum_{j\ne i}\left\{K_h(y_\ell-Y_{ij})\1(Y_{ij}\ne 0)-\E[K_h(y_\ell-Y_{ij})\1(Y_{ij}\ne 0)\mid U_i]\right\}\right|^2\mid U_i\right]\\
&\quad  \lesssim \frac{\log p}{n} \cdot\E\left[\max_{1\le \ell \le p} \left(K_h(y_\ell-Y_{ij})\right)^2\1(Y_{ij}\ne 0)\mid U_i\right] \\
&\quad \le \frac{\| K \|_{\infty}^2\log p}{nh^2}
\end{align*}
 up to a constant independent of $n$ and $i$. 
Thus by Fubini, we have
\begin{align*}
&\E\left[\max_{1\le \ell \le p}\left|\frac{1}{n-1}\sum_{j\ne i}\left\{K_h(y_\ell-Y_{ij})\1(Y_{ij}\ne 0)-\E[K_h(y_\ell-Y_{ij})\1(Y_{ij}\ne 0)\mid U_i]\right\}\right|^2\right] \\
&=O\left( \frac{\log p}{nh^2} \right).
\end{align*}
Recalling that $\max_{1\le \ell \le p}\left|\E[K_h(y_\ell-Y_{ij})\1(Y_{ij}\ne 0)\mid U_i]\right| \le M \int |K(z)|dz$, we have
\begin{align*}
\Delta_{III}=O_P\left(|\hat a- a|^2 \vee \max_{1\le \ell\le p}|\hat b(y_\ell)-\overline b_h(y_\ell)|^2\right)=O_P\left(\frac{\log p}{n}\right).
\end{align*}

Next consider $\Delta_{IV}$. We note that
\begin{align*}
\max_{1\le \ell \le p}\frac{1}{n} \sumi |\hat W_i^\ell|^2 \le \frac{2}{n} \sumi \max_{1\le \ell \le p}\left|\hat W_i^\ell-\E[\hat W_i^\ell\mid U_i ]\right|^2 +\frac{2}{n} \sumi \max_{1\le \ell \le p}\left|\E[\hat W_i^\ell\mid U_i ]\right|^2.
\end{align*}
Since for any $i, j \ne i$ and $\ell$, $\E[\hat W_{i}^\ell\mid U_i]=2\E[X_{ij}^\ell \mid U_i] \le 4a^{-1} M \int |K(z)| dz$, we have
\begin{align*}
\E\left[\frac{1}{n} \sumi \max_{1\le \ell \le p}\left|\E[\hat W_i^\ell\mid U_i] \right|^2\right]= &\E\left[ \max_{1\le \ell \le p}\left|2\E[ X_{12}^\ell\mid U_1] \right|^2\right]= O(1).
\end{align*}
Further, for any $i$ and $j\ne i$, conditionally on $U_i$, we have
\begin{align*}
\E\left[\max_{1\le \ell \le p} \left| \hat W_i^\ell -\E[\hat W_i^\ell\mid U_i]\right|^2\mid U_i\right]=
\E\left[\max_{1\le \ell \le p} \left| \hat W_i^\ell -2\E[X_{ij}^\ell\mid U_i]\right|^2\mid U_i\right].
\end{align*}
Conditionally on $U_i$, we apply Theorem 2.14.1 in \cite{vdVW1996} with $p=2$ to deduce that
\begin{align*}
\E\left[\max_{1\le \ell \le p} \left| \hat W_i^\ell -2\E[X_{ij}^\ell\mid U_i]\right|^2\mid U_i\right] &\lesssim
\frac{\log p}{n}\cdot\E[\|\bX_{ij}\|_\infty^2\mid U_i] \\
&\lesssim \frac{\log p}{nh^2}
\end{align*}
up to constants independent of $n$ and $i$. 
By Fubini, we have
\begin{align*}
\E\left[\max_{1\le \ell \le p} \left| \hat W_i^\ell -2\E[X_{ij}^\ell\mid U_i]\right|^2\right]=
O\left(\frac{\log p}{nh^2}\right).
\end{align*}
Thus we have $\max_{1\le \ell \le p}n^{-1} \sumi |\hat W_i^\ell|^2=O_P(1)$, so that
\begin{align*}
\Delta_{IV}=\Delta_{III}^{1/2}\cdot O_P(1)=O_P\left( \sqrt{\frac{\log p}{n}}\right).
\end{align*}
Conclude that that $\Delta_{I} \le \Delta_{III} +2\Delta_{IV}= O_P\left( (n^{-1}\log p)^{1/2}\right)$.

To bound $\Delta_{II}$, we first note that $\| \bS_n\|_\infty=O_P((n^{-1}\log p)^{1/2})$, which follows from the Gaussian approximation, Theorem \ref{thm: high-d CLT polyadic}. Combined with the rates for $|\hat a-a|$ and $\max_{1\le \ell \le p}|\hat b(y_\ell)-\overline b_h(y_\ell)|$, and the fact that
\begin{align*}
\tilde S_n^\ell \tilde S_n^{\ell'}- S_n^\ell S_n^{\ell'}=
(\tilde S_n^\ell -S_n^\ell )(\tilde S_n^{\ell'}- S_n^{\ell'}) +  S_n^\ell(\tilde S_n^{\ell'}-S_n^{\ell'})  +(\tilde S_n^\ell - S_n^\ell ) S_n^{\ell'},
\end{align*}
we have $\Delta_{II}=O_P((n^{-1} \log p)^{1/2} )$. By Condition (v) of the proposition, we conclude that
\begin{align*}
\|\tilde \Sigma - \hat \Sigma\|_\infty \log^2 p=O_P\left(\sqrt{\frac{\log^5 p}{n}} \right)=o_P(1).
\end{align*}

For the second statement of the proposition for the non-normalized statistic, under the conditions of the proposition, the bias can be controlled uniformly over $y\in\{y_1,\dots,y_p\}$ by
\begin{align*}
|\overline f_h(y)- f(y)|\le \frac{h^r}{r!}\|f^{(r)}\|_\infty \int |z^r K(z)|dz=O(h^r).
\end{align*}
By Lemma \ref{lem: AC}, we have
\begin{align*}
&\left|\Prob\left(\left (\overline f_h (y_\ell)\right )_{\ell=1}^p\in \mathcal I(1-\alpha)\right)-	\Prob\left(\left (f (y_\ell))\right)_{\ell=1}^p\in \mathcal I(1-\alpha) \right)\right|\\
&\lesssim \sqrt{n \log p}\cdot\max_{1\le \ell\le p}	|\overline f_h(y_\ell)- f(y_\ell)|
=O(h^r\sqrt{n\log p} )=o(1).
\end{align*}
The argument here follows from similar steps to those in Corollary 3 in \cite{KatoSasaki2018}. We omit the details for brevity.

Finally, the result for the normalized  statistic follows from the above results for the non-normalized statistic in view of Corollary \ref{cor:normalized_polyadic_MB} under the conditions of the proposition. This completes the proof. 
\qed

\subsection{Proof of Proposition \ref{proposition:lasso_MB_penalty}} 
In this proof, the notation $\lesssim$ means that the left-hand side is bounded by the right-hand side up to a constant independent of $n$. 

By Theorem \ref{thm: high-d CLT} (use Condition (\ref{eq:condition1_poly})), we have 
\[
\sup_{t \in \R} |\Prob (\| \sqrt{n}\bm{S}_{\bN} \|_{\infty} \le t) - \Prob (\| \bm{G} \|_{\infty} \le t) | \to 0,
\]
where $\bm{G} \sim N(\bm{0},\Sigma)$ with $\Sigma =\sum_{k=1}^K(n/N_k)\E[\bV_{k,1}\bV_{k,1}^T]$. 
Conditionally on $((Y_\i,\bm{Z}_\i^T)^T)_{\i\in[\bN]}$, we have 
\begin{align*}
\sum_{k=1}^K\frac{\sqrt{n}}{N_k}\sum_{i_k=1}^{N_k} \xi_{k,i_k} (\tilde{\bm{V}}_{k,i_k}-\tilde{\bm{S}}_N)\sim N(\bm{0},\tilde\Sigma ), \ \text{where }\tilde \Sigma=\sum_{k=1}^K (n/N_k^2) \sum_{i_k=1}^{N_k}(\tilde\bV_{k,i_k}-\tilde{\bm{S}}_{\bN})(\tilde\bV_{k,i_k}-\tilde{\bm{S}}_{\bN})^T.
\end{align*}
Thus, in view of Lemma \ref{lem:Gaussian_comparison}, it suffices to show that  $\|\tilde\Sigma-\Sigma\|_\infty\log^2 p=o_{P}(1)$. Further, the proof of Theorem \ref{thm:bootsrap_validity} under polynomial moment conditions (see also Remark \ref{remark:discussion_polynomial_rates}) implies that $\|\hat\Sigma-\Sigma\|_\infty=o_P((\log p)^{-2})$, where $\hat{\Sigma} = \sum_{k=1}^K (n/N_k^2)\sum_{i_k=1}^{N_k}(\hat \bV_{k,i_k}-\bm{S}_{\bN})(\hat\bV_{k,i_k}-\bm{S}_{\bN})^T$ and $\hat\bV_{k,i_k} = (\prod_{k'\ne k} N_{k'})^{-1}\sum_{i_1,\dots,i_{k-1},i_{k+1},\dots,i_K}\varepsilon_\i \bX_{\i}$.  Thus, it suffices to show that 
$\|\tilde\Sigma-\hat\Sigma\|_\infty=o_P((\log p)^{-2})$.

Recall that $\lambda^0 = (\log n) (n^{-1} \log p)^{1/2}$. 
We note that 
\[
\E[\| \bm{G} \|_{\infty}] \lesssim \max_{j,k} \sqrt{\E[(V_{k,1}^{j})^2] \log p} \lesssim \sqrt{\log p}, 
\]
so that $\lambda^0 \ge 2c\|\bm{S}_\bN\|_\infty$ with probability $1-o(1)$. By assumption, $\kappa(s,c_0)$ is  bounded away from zero with probability $1-o(1)$. Thus, Theorem 1 in \cite{BC2013} implies that
\begin{align*}
\sqrt{\frac{1}{N}\sum_{\i\in [\bN]}(\bX_\i^T(\tilde \beta - \beta_0))^2} =O_P\left(\sqrt{ \frac{s\log^3 (p \overline N)}{n}}\right).
\end{align*}

Observe that
\begin{align*}
\|\tilde\Sigma-\hat\Sigma\|_\infty
\le&
\sum_{k=1}^K\underbrace{\max_{1\le j,\ell\le p}\left|\frac{1}{N_k}\sum_{i_k=1}^{N_k}(\tilde V_{k,i_k}^j\tilde V_{k,i_k}^\ell-\hat V_{k,i_k}^j\hat V_{k,i_k}^\ell)\right|}_{=:(I_k)}+K\underbrace{\max_{1\le j,\ell\le p} \left|\tilde S_{\bN}^j\tilde S_{\bN}^\ell- S_{\bN}^jS_{\bN}^\ell\right|}_{=:(II)}.
\end{align*}
We first consider the term $(I_k)$. We shall focus on $k=1$ as similar bounds hold for other $k$. Observe that
\begin{align*}
\frac{1}{N_1}\sum_{i_1=1}^{N_1}(\tilde V_{1,i_1}^j\tilde V_{1,i_1}^\ell-\hat V_{1,i_1}^j\hat V_{1,i_1}^\ell)=&
\frac{1}{N_1}\sum_{i_1=1}^{N_1}(\tilde V_{1,i_1}^j-\hat V_{1,i_1}^j)(\tilde V_{1,i_1}^\ell-\hat V_{1,i_1}^\ell)
+\frac{1}{N_1}\sum_{i_1=1}^{N_1}(\tilde V_{1,i_1}^j-\hat V_{1,i_1}^j)\hat V_{1,i_1}^\ell\\
&+\frac{1}{N_1}\sum_{i_1=1}^{N_1}\hat V_{1,i_1}^j(\tilde V_{1,i_1}^\ell-\hat V_{1,i_1}^\ell).
\end{align*}
By Cauchy-Schwarz, we have
\begin{align*}
(I_1)\le & \underbrace{\max_{1\le j\le p}\frac{1}{N_1}\sum_{i_1=1}^{N_1}(\tilde V_{1,i_1}^j-\hat V_{1,i_1}^j)^2}_{=:(III)} +2 (III)^{1/2}\sqrt{\underbrace{\max_{1\le \ell \le p}\frac{1}{N_1}\sum_{i_1=1}^{N_1}|\hat V_{1,i_1}^\ell|^2}_{=:(IV)}}.
\end{align*}

To bound $(IV)$, we note that
\begin{align*}
\max_{1\le \ell \le p}\frac{1}{N_1}\sum_{i_1=1}^{N_1}|\hat V_{1,i_1}^\ell|^2
 \le
\frac{1}{N_1}\sum_{i_1=1}^{N_1}\left\|\hat \bV_{1,i_1}-\E[\hat \bV_{1,i_1}\mid U_{(i_1,0,\dots,0)}]\right\|_\infty^2
 +
\frac{1}{N_1}\sum_{i_1=1}^{N_1}\left\|\E[\hat \bV_{1,i_1}\mid U_{(i_1,0,\dots,0)}]\right\|_\infty^2
\end{align*}
Since $\E[\hat \bV_{1,i_1}\mid U_{(i_1,0,\dots,0)}]=\E\b[\varepsilon_{\bm{1}}X_{\bm{1}}\mid U_{(1,0,\dots,0)}]$ for all $i_1$, by Fubini and Jensen's inequality, we have
\begin{align*}
\E\left[\frac{1}{N_1}\sum_{i_1=1}^{N_1}\left\|\E[\hat \bV_{1,i_1}\mid U_{(i_1,0,\dots,0)}]\right\|_\infty^2\right]\le&
\left(\E\left[\left\|\E[\varepsilon_{\bm{1}}X_{\bm{1}}^\ell\mid U_{(1,0,\dots,0)}]\right\|_\infty^q\right]\right)^{2/q}\\
\le&\left(\E\left[\max_{1\le\ell \le p} |\varepsilon_{\bm{1}}X_{\bm{1}}^\ell|^q\right]\right)^{2/q}\le D_\bN^2.
\end{align*}
Conditionally on $U_{(i_1,0,\dots,0)}$,
\begin{align*}
\E\left[\left\|\hat \bV_{1,i_1}-\E[\hat \bV_{1,i_1}\mid U_{(i_1,0,\dots,0)}]\right\|_\infty^2\mid U_{(i_1,0,\dots,0)}\right]\le&\left(\E\left[\left\|\hat \bV_{1,i_1}-\E[\varepsilon_\i\bX_\i\mid U_{(i_1,0,\dots,0)}]\right\|_\infty^q\mid U_{(i_1,0,\dots,0)}\right]\right)^{2/q}.
\end{align*}
As in the proof of Theorem \ref{thm:bootsrap_validity} (see Step 1), conditionally on $U_{(i_1,0,\dots,0)}$, the array $( \varepsilon_{(i_1,\i_{-1})}\bX_{(i_1,\i_{-1})} )_{\i_{-1} \in \N^{K-1}}$ is separately exchangeable with mean vector $\E[\varepsilon_\i\bX_\i\mid U_{(i_1,0,\dots,0)}]$. 
By Corollary \ref{cor: maximal inequality}, we have
\begin{align*}
\E\left[\left\|\hat \bV_{1,i_1}-\E[\varepsilon_\i\bX_\i\mid U_{(i_1,0,\dots,0)}]\right\|_\infty^q\mid U_{(i_1,0,\dots,0)}\right]\lesssim& 
n^{-q/2}(\log p)^{q/2} \E[\|\varepsilon_{\i}\bX_{\i} \|_\infty^{q}\mid U_{(i_1,0,\dots,0)}].
\end{align*}
By Fubini, we have
\begin{align*}
\E\left[\left\|\hat\bV_{1,i_1}-\E[\varepsilon_\i\bX_\i\mid U_{(i_1,0,\dots,0)}]\right\|_\infty^q\right]\lesssim& 
n^{-q/2}(\log p)^{q/2} D_\bN^{q}.
\end{align*}
Conclude that  $|(IV)|=O_P(D_\bN^2)$.

Next, we shall bound the term $(III)$.  Observe that by Cauchy-Schwarz,
\begin{align*}
&|\tilde V_{1,i_1}^j-\hat V_{1,i_1}^j|=\left|\frac{1}{\prod_{k\ne 1}N_k}\sum_{i_2,\dots,i_K}X_\i^j(\bX_\i^T(\tilde \beta-\beta_0)+r_\i)\right|\\
&\quad \le
\sqrt{\frac{1}{\prod_{k\ne 1}N_k}\sum_{i_2,\dots,i_K}(X_\i^j)^2}\left(\sqrt{\frac{1}{\prod_{k\ne 1}N_k}\sum_{i_2,\dots,i_K}(\bX_\i^T(\tilde \beta -\beta_0))^2}+\sqrt{\frac{1}{\prod_{k\ne 1}N_k}\sum_{i_2,\dots,i_K}r_\i^2}\right),
\end{align*}
so that the term $(III)$ is bounded as
\begin{align*}
&\lesssim 
\max_{j} \frac{1}{N}\sum_{i_1=1}^{N_1}\left(\frac{1}{\prod_{k\ne 1}N_k}\sum_{i_2,\dots,i_K}(X_\i^j)^2\right)\left(\frac{1}{\prod_{k\ne 1}N_k}\sum_{i_2,\dots,i_K}(\bX_\i^T(\tilde \beta -\beta_0))^2+\frac{1}{\prod_{k\ne 1}N_k}\sum_{i_2,\dots,i_K}r_\i^2\right)\\
&\le
\left(\max_{j, i_1} \frac{1}{\prod_{k\ne 1}N_k}\sum_{i_2,\dots,i_K}(X_\i^j)^2\right) \underbrace{\left(\frac{1}{N}\sum_{\i\in[\bN]}(\bX_\i^T(\tilde \beta -\beta_0))^2+\frac{1}{N}\sum_{\i\in[\bN]}r_\i^2\right)}_{=O_P\left(\frac{s\log^3 (p \overline N)}{n}\right)}.
\end{align*}
Observe that
\[
\begin{split}
&\E\left[\max_{j,i_1} \frac{1}{\prod_{k\ne 1}N_k}\sum_{i_2,\dots,i_K}(X_\i^j)^2\right] \\
&\le \E\left[\max_{j,i_1} \left | \frac{1}{\prod_{k\ne 1}N_k}\sum_{i_2,\dots,i_K}\{ (X_\i^j)^2 - \E[(X_{(i_1,1,\dots,1)}^j)^2 \mid U_{(i_1,0,\dots,0)}] \} \right | \right] \\
&\quad + \E\left [ \max_{j,i_1} \E[(X_{(i_1,1,\dots,1)}^j)^2 \mid U_{(i_1,0,\dots,0)}] \right ]. 
\end{split}
\]
By H\"older's inequality, we have
\[
\begin{split}
\E\left [ \max_{j,i_1} \E[(X_{(i_1,1,\dots,1)}^j)^2 \mid U_{(i_1,0,\dots,0)}] \right ] &\le \E \left  [ \max_{i_1} \E[\| \bX_{(i_1,1,\dots,1)} \|_{\infty}^{2} \mid  U_{(i_1,0,\dots,0)}] \right ] \\
&\le \E \left  [ \max_{i_1} \left ( \E[\| \bX_{(i_1,1,\dots,1)} \|_{\infty}^{2q} \mid  U_{(i_1,0,\dots,0)}] \right )^{1/q}  \right ] \\
&\le \left ( \E \left [ \sum_{i_1} \E[\| \bX_{(i_1,1,\dots,1)} \|_{\infty}^{2q} \mid  U_{(i_1,0,\dots,0)}] \right ] \right )^{1/q} \\
&\le \overline{N}^{1/q} D_{\bN}. 
\end{split}
\]
Applying Corollary \ref{cor: maximal inequality} conditionally on $U_{(i_1,0,\dots,0)}$ (cf. Step 1 in the proof of Theorem \ref{thm:bootsrap_validity}), we have 
\[
\begin{split}
&\E\left[ \max_{j} \left | \frac{1}{\prod_{k\ne 1}N_k}\sum_{i_2,\dots,i_K}\{ (X_\i^j)^2 - \E[(X_{(i_1,1,\dots,1)}^j)^2 \mid U_{(i_1,0,\dots,0)}] \} \right |^{q} \mid U_{(i_1,0,\dots,0)} \right] \\
&\lesssim n^{-q/2} (\log p)^{q/2} \E[\| \bX_{(i_1,1,\dots,1)} \|_{\infty}^{2q} \mid U_{(i_1,0,\dots,0)}].
\end{split}
\]
Thus, we have 
\[
\begin{split}
&\E\left[\max_{j,i_1} \left | \frac{1}{\prod_{k\ne 1}N_k}\sum_{i_2,\dots,i_K}\{ (X_\i^j)^2 - \E[(X_{(i_1,1,\dots,1)}^j)^2 \mid U_{(i_1,0,\dots,0)}] \} \right | \right] \\
&\le\left ( \sum_{i_1} \E\left[ \max_{j} \left | \frac{1}{\prod_{k\ne 1}N_k}\sum_{i_2,\dots,i_K}\{ (X_\i^j)^2 - \E[(X_{(i_1,1,\dots,1)}^j)^2 \mid U_{(i_1,0,\dots,0)}] \} \right |^{q} \right] \right )^{1/q} \\
&\lesssim \overline{N}^{1/q} n^{-1/2} (\log p)^{1/2} D_{\bN}. 
\end{split}
\]
Conclude that $(III) = O_P\left(  \{ n^{-1}s\overline N^{1/q}D_\bN \log^3 ( p \overline N)\}^{1/2}\right)$ and consequently 
\[
|(I_1)|=O_P\left(  \{ n^{-1}s\overline N^{1/q}D_\bN^3\log^3 ( p \overline N)\}^{1/2}\right).
\]

Finally, to bound $|(II)|$, observe that
\[
\begin{split}
\tilde S_{\bN}^j\tilde S_{\bN}^\ell- S_{\bN}^jS_{\bN}^\ell &=(\tilde S_{\bN}^j-S_{\bN}^j)(\tilde S_{\bN}^\ell-S_{\bN}^\ell)+S_{\bN}^j(\tilde S_{\bN}^\ell-S_{\bN}^\ell) \\
&\quad +(\tilde S_{\bN}^j-S_{\bN}^j)S_{\bN}^\ell.
\end{split}
\]
Then, we have 
\begin{align*}
|(II)|\le &\max_{1\le j\le p}\left| \frac{1}{N}\sum_{\i\in [\bN]}(\tilde \varepsilon_\i - \varepsilon_\i)X_\i^j\right|^2 
+2 \|\bS_{\bN}\|_\infty\cdot \max_{1\le j \le p} \left| \frac{1}{N}\sum_{\i\in [\bN]}(\tilde \varepsilon_\i - \varepsilon_\i)X_\i^j\right| .
\end{align*}
By Cauchy-Schwarz, we have
\begin{align*}
\max_{1\le j\le p}\left| \frac{1}{N}\sum_{\i\in [\bN]}(\tilde \varepsilon_\i - \varepsilon_\i)X_\i^j\right|
&\le
\max_{1\le j\le p}\sqrt{\frac{1}{N}\sum_{\i \in [\bN]}(X_\i^j)^2}\left(\sqrt{ \frac{1}{N}\sum_{\i\in [\bN]}(\bX_\i^T(\beta_0-\tilde \beta))^2}+\|\bm{r}\|_{N,2}\right)\\
&=O_P\left(\sqrt{ \frac{sD_\bN\log^3 (\overline Np)}{n}}\right),
\end{align*}
so that $|(II)|=O_P\left(  n^{-1}s D_\bN^3\log^3 ( p \overline N) + \{n^{-2}s D_\bN(\log p)(\log^3 ( p \overline N))\}^{1/2}\right).$

Combining the above bounds, we have
$
\|\tilde \Sigma-\hat \Sigma\|_\infty=O_P\left(  \{n^{-1}s\overline N^{1/q} D_\bN^3\log^3 ( p \overline N)\}^{1/2}\right).
$
This implies  that $\|\tilde \Sigma-\hat \Sigma\|_\infty\log^2 p=o_P(1)$, as required.

Finally, by the Gaussian concentration, we have 
\[
\lambda =  O_{P} \left( \sqrt{\frac{\log  p }{n}} \bigvee \sqrt{\frac{\log(1/\eta)}{n}} \right).
\]
Together with Theorem 1 of \cite{BC2013}, we obtain the desired bound on $\| \hat{\bm{f}}^{\lambda} - \bm{f} \|_{N,2}$. 
\qed

\subsection{Proof of Proposition \ref{proposition:lasso_K=2}}
%The conclusion follows from combining Lemma \ref{lemma:sparse_eigenvalues} below and Proposition \ref{proposition:lasso_MB_penalty} with $B=N^{1/q}D_\bN$.
Recall that $K=2$. We write $\bX_{i,j}$ instead of $\bX_{(i,j)}$ for the notational simplicity. 
Define the $N \times p$ matrix $\mathbb{X} = (\bX_{1,1},\dots,\bX_{N_1, 1},\bX_{2,1},\dots,\bX_{N_1,N_2})^{T}$. 
The $s$-sparse eigenvalue with $1 \le s \le p$ for $\mathbb{X}$ is defined by 
\[
\phi_{\min}(s)=\min_{\|\theta\|_0\le s, \| \theta \|=1} \| \mathbb{X} \theta \|_{N,2}.
\]
By \citet[Lemma 2.7]{LecueMendelson2017}, if $\phi_{\min}(s) \ge \phi_{1}$, then for $2 \le s \le p$, we have 
\[
\| \mathbb{X} \theta \|_{N,2}^{2} \ge \phi_{1}^{2} \| \theta \|^{2}- \frac{\| \theta \|_{1}^{2}}{s-1} \times \underbrace{\max_{1 \le \ell \le p}  \sum_{(i,j) \in [\bN]} (X_{i,j}^{\ell})^{2}/N}_{=: \hat{\rho}}
\]
for all $\theta \in \R^{p}$. We can then deduce that for $s_{1} \le (s-1)\phi_{1}^2/(2(1+c_0)^2\hat{\rho})$, we have 
\[
\kappa (s_1,c_0) \ge \phi_1/\sqrt{2}.
\]
Lemma \ref{lem:sparse_eigenvalues} below implies that $\phi_{\min}(s)$ is bounded away from zero with probability $1-o(1)$. 
Further, observe that 
\[
\hat{\rho} \le \max_{1 \le \ell \le p} \E[(X_{1,1}^{\ell})^2] + \max_{1 \le \ell \le p}\left | N^{-1} \sum_{(i,j) \in [\bN]} \{ (X_{i,j}^{\ell})^2 - \E[(X_{1,1}^{\ell})^2]\} \right |.
\]
The first term on the right-hand side is $O(1)$, while the second term is $o_{P}(1)$ (which follows from Lemma \ref{lem:sparse_eigenvalues} below with $s=1$), so that $\hat{\rho} = O_{P}(1)$. The conclusion of the proposition follows from rescaling $s$. 
\qed

\begin{lemma}[Sparse eigenvalues for two-way clustering]\label{lem:sparse_eigenvalues}
Suppose  that  $(\bX_{i,j})_{(i,j) \in [\bN]}$ with $[\bN] = \{ 1,\dots, N_1 \} \times \{1,\dots,N_2\}$ is sampled from a separately exchangeable array $(\bX_{i,j})_{(i,j)\in \N^2}$ generated as $\bX_{i,j} = \mathfrak{g}(U_{i,0},U_{0,j},U_{i,j})$ for some Borel measurable map $\mathfrak{g}: [0,1]^{3} \to \R^{p}$ and i.i.d. $U[0,1]$ variables $U_{i,0},U_{j,0},U_{i,j}$. 
Pick any $1 \le s \le p \wedge n$. 
Let $B = \sqrt{\E[M^2]}$ with $M=\max_{(i,j) \in [\bN]}\|\bX_{i,j}\|_\infty$. 
Define
\[
\delta_{\bN}=\sqrt{s}B\left(\frac{1}{\sqrt{n}}\left\{\log^{1/2} p + (\log s)(\log^{1/2} \overline N) (\log^{1/2} p)\right\}\bigvee \frac{1}{\sqrt{N}}\left\{\log p + (\log \overline N)(\log p)\right\}\right).
\]
Then, we have 
\[
\E\left[ \sup_{\|\theta\|_0 \le s, \|\theta\|=1 } \left| \frac{1}{N} \sum_{(i,j)\in [\bN] }\{(\theta^T\bX_{i,j} )^2 - \E[(\theta^T\bX_{1,1} )^2]\}\right|\right]
\lesssim
\delta_{\bN}^2 +\delta_{\bN}\sup_{\|\theta\|_0\le s,\,\|\theta\|=1}\sqrt{\E[(\theta^T \bX_{1,1})^2]}
\]
up to a universal constant. 
In addition, we have
$\delta_{\bN} \lesssim \{n^{-1}sB^2\log^4( p\overline N)\}^{1/2}$ up to a universal constant. 
\end{lemma}
%%%%%%%%%%%%%%%%%%%%%%%%%%%%%%%%%%%%%%%%%%
\begin{proof}[Proof of Lemma \ref{lem:sparse_eigenvalues}]
In this proof, the notation $\lesssim$ means that the left-hand side is bounded by the right-hand side up to a universal constant.
 
Let $\Theta_{s}=\cup_{|T|=s}\{\theta\in \R^p: \|\theta\|=1,\, \supp(\theta)\subset T\}$.
Further, let $Z_{i,j}(\theta)=(\theta^T\bX_{i,j})^2 - \E[(\theta^T\bX_{1,1} )^2]$. Then, for each $\theta$, $Z_{i,j}(\theta)$ is a centered random variable. Consider the decomposition
\[
Z_{i,j}(\theta) = \E[Z_{i,1}(\theta)\mid U_{i,0}] + \E[Z_{1,j}(\theta)\mid U_{0,j}] + \underbrace{Z_{i,j}(\theta) - \E[Z_{i,1}(\theta)\mid U_{i,0}] - \E[Z_{1,j}(\theta)\mid U_{0,j}]}_{=:\hat{Z}_{i,j}(\theta)}. 
\]
We divide the rest of the proof into two steps. 
	
\underline{Step 1}. 
Consider first the term $\sum_{i,j}\E[Z_{i,j}(\theta)\mid U_{i,0}]=N_2\sum_{i=1}^{N_1}\E[Z_{i,1}(\theta)\mid U_{i,0}]$, which consists of i.i.d. variables. 
Observe that $\E[Z_{i,1}(\theta)\mid U_{i,0}]$ has mean $0$ and by symmetrization
\[
\begin{split}
\E\left[\sup_{\theta\in \Theta_{s}}\left|\sum_{i=1}^{N_1}\E[Z_{i,1}(\theta)\mid U_{i,0}]\right|\right]
&=
\E\left[\sup_{\theta\in \Theta_{s}}\left|\sum_{i=1}^{N_1}\left(\theta^T\E[\bX_{i,1}\bX_{i,1}^T\mid U_{i,0}]\theta-\E[(\theta^T \bX_{1,1})^2]\right)\right|\right]\\
&\le
2 \E\left[\E\left[\sup_{\theta\in \Theta_{s}}\left|\sum_{i=1}^{N_1}\epsilon_i \left(\theta^T\E[\bX_{i,1}\bX_{i,1}^T\mid U_{i,0}]\theta \right)\right|\mid \bX_{[\bN]}\right]\right] \\
&\le
2 \E\left[\E\left[\sup_{\theta\in \Theta_{s}}\left|\sum_{i=1}^{N_1}\epsilon_i (\theta^{T}\bX_{i,1})^2\right|\mid \bX_{[\bN]}\right]\right],
\end{split}
\]
where $(\epsilon_i)_{i=1}^{N_1}$ is a sequence of independent Rademacher random variables  that are independent of  $(\bX_{i,j})_{(i,j) \in [\bN]}$, and the second inequality follows from Jensen's inequality. 
Now, the following bound can be obtained by following the proof of Lemma P.1. in \cite{BCCW2018} with $\mathcal U$ set to be a singleton set:
\[
\E\left[\sup_{\theta\in \Theta_{s}}\left|\sum_{i=1}^{N_1}\epsilon_i (\theta^T\bX_{i,1})^2 \right| \mid \bX_{[\bN]}\right]\lesssim\sqrt{s} M R_1(\log^{1/2} p + (\log s) (\log^{1/2} \overline N)(\log^{1/2}p)),
\]
where $R_1=\sup_{\theta\in\Theta_{s}}\left(\sum_{i=1}^{N_1} (\theta^T \bX_{i,1})^2\right)^{1/2}$.

Choosing $\delta_{\bN,1}=BN_1^{-1/2}\sqrt{s}\{\log^{1/2} p + (\log s) (\log^{1/2} \overline N)(\log^{1/2}p)\}$, 
by Cauchy-Schwarz, we have 
\begin{align*}
I&:=
\E\left[\sup_{\theta\in \Theta_{s}}\left|\sum_{i=1}^{N_1}\epsilon_i (\theta^{T}\bX_{i,1})^2\right|\right]
\lesssim
\frac{\delta_{\bN,1} \E[MR_1]}{B\sqrt{N_1}}
\le\left(\frac{\delta_{\bN,1}}{B}\right) \left(\frac{\E[M^2]\E[R_1^2]}{N_1}\right)^{1/2}
\\
&\le \delta_{\bN,1}(\E[R_1^2/N_1])^{1/2}
\lesssim
\delta_{\bN,1}\left(I+\sup_{\theta\in\Theta_{s}}\E[(\theta^T\bX_{1,1})^2] \right)^{1/2}.
\end{align*}
Using the algebraic fact that $a^2\le \delta^2 a+\delta^2 b$ implies $a\le \delta^2 +a^{-1}\delta^2b$, we have 
\[
I\lesssim \delta_{\bN,1}^2 + \delta_{\bN,1} \sqrt{\sup_{\theta\in\Theta_{s}}\E[(\theta^T \bX_{1,1})^2]}.
\]
The same bound holds for 
$\E\left[\sup_{\theta\in \Theta_{s}}\left|N_2^{-1}\sum_{j=1}^{N_2}\E[Z_{1,j}(\theta)\mid U_{0,j}]\right|\right]$. Conclude that 
\begin{align*}
\E\left[\sup_{\theta\in \Theta_{s}}\left|\frac{1}{N}\sum_{i,j}\left(\E[Z_{i,j}(\theta)\mid U_{i,0}]+\E[Z_{i,j}(\theta)\mid U_{0,j}]\right)\right|\right]
\lesssim \delta_{\bN,2}^2 +\delta_{\bN,2}\sqrt{\sup_{\theta\in\Theta_{s}}\E[(\theta^T \bX_{1,1})^2]},
\end{align*}
where $\delta_{\bN,2}=Bn^{-1/2}\sqrt{s}\{\log^{1/2} p + (\log s) (\log^{1/2} \overline N)(\log^{1/2}p)\}\lesssim Bn^{-1/2}\sqrt{s} \log^{2}( p\overline N)$.
	
\underline{Step 2}. 
	Now, to obtain a bound on $\E[\sup_{\theta\in \Theta_{s}}|N^{-1}\sum_{i,j}\hat Z_{i,j}(\theta)|]$, by Lemma \ref{lem: symmetrization} (more formally, we apply Lemma \ref{lem: symmetrization} after approximating $\Theta_s$ by a sequence of finite sets and take limits), we have  the following symmetrization inequality
\begin{align*}
\E\left[\sup_{\theta\in \Theta_{s}}\left|\sum_{i,j} \hat Z_{i,j}(\theta)\right|\right]
	&\le 
	4\E\left[\E\left[\sup_{\theta\in \Theta_{s}}\left|\sum_{i,j} \epsilon_i \epsilon'_j \hat Z_{i,j}(\theta)\right|\mid \bX_{[\bN]}\right]\right]\\
	&\lesssim \E\left[\E\left[\sup_{\theta\in \Theta_{s}}\left|\sum_{i,j} \epsilon_i \epsilon'_j  (\theta^T\bX_{i,j})^2\right|\mid \bX_{[\bN]}\right]\right],
	\end{align*}
	where $(\epsilon_i)$ and $(\epsilon_i')$ are independent copies of Rademacher random variables independent of $(\bX_{i,j})_{(i,j) \in [\bN]}$, and the second inequality follows from Jensen's inequality. 
Conditionally on $(\bX_{i,j})_{(i,j) \in [\bN]}$, $\sum_{i,j} \epsilon_i \epsilon_j' (\theta^T\bX_{i,j})^2$ is a Rademacher chaos of degree $2$ (cf. the proof of Theorem \ref{thm: maximal inequality}). Hence, Corollary 5.1.8 in \cite{delaPenaGine1999} yields that 
\begin{align*}
II&:=\E\left[\sup_{\theta\in \Theta_{s}}\left|\sum_{i,j} \epsilon_i\epsilon_j' (\theta^T \bX_{i,j})^2\right|\mid \bX_{[\bN]}\right]
\lesssim \left\|\sup_{\theta \in \Theta_{s}} \left | \sum_{i,j} \epsilon_i \epsilon'_j  (\theta^T \bX_{i,j})^2 \right | \right\|_{\psi_1 \mid \bX} \\
&\lesssim
\int_{0}^{\text{diam}(\Theta_{s})} \log N(\Theta_{s},\rho_{\bX},t)d t,
\end{align*}
where $\|\cdot\|_{\psi_1 \mid \bX}$ is the $\psi_1$-norm evaluated conditionally on $(\bX_{i,j})_{(i,j) \in [\bN]}$, $\rho_{\bX}$ is a pseudometric on $\Theta_{s}$ defined by $\rho_{\bX}(\theta,\overline\theta)=\left( \sum_{i=1}^{N_1}\sum_{j=1}^{N_2}\{(\theta^T\bX_{i,j} )^2-(\overline\theta^T \bX_{i,j} )^2\}^2 \right)^{1/2}$, and $\text{diam} (\Theta_{s})$ is the $\rho_{\bX}$-diameter of $\Theta_{s}$. Now, for any two $\theta$, $\bar\theta\in \Theta_{s}$, 
\begin{align*}
\rho_{\bX} (\theta,\bar\theta) &=
\left(\sum_{i=1}^{N_1}\sum_{j=1}^{N_2}\left\{(\theta^T \bX_{i,j})^2-(\bar\theta^T \bX_{i,j})^2 \right\}^2\right)^{1/2}\\
&\le
\left(\sum_{i=1}^{N_1}\sum_{j=1}^{N_2}\left\{(\theta^T \bX_{i,j})+(\bar\theta^T \bX_{i,j}) \right\}^2\right)^{1/2} \max_{(i,j)\in [\bN]}|(\theta - \bar \theta)^T\bX_{i,j}| \\
&\le
\sqrt{2} R_{2} \|\theta - \bar \theta\|_{\bX},
\end{align*}
where $R_{2}=\sup_{\theta\in\Theta_{s}}\left(\sum_{(i,j) \in [\bN]}   (\theta^T \bX_{i,j})^2\right)^{1/2}$ and  $\| \theta \|_{\bX}=\max_{(i,j)\in [\bN]}|\theta^T \bX_{i,j}|$. Thus, we have
\begin{align*}
\int_{0}^{\text{diam}(\Theta_{s})} \log N(\Theta_{s},\rho_{\bX},t)dt
&\le \int_{0}^{2\sqrt{2s}MR_{2}} \log N \left(\Theta_{s}/\sqrt{s},\|\cdot\|_{\bX},t/(\sqrt{2s}R_{2}) \right)dt \\
&= 2\sqrt{2s} R_{2} \int_{0}^{M} \log N \left(\Theta_{s}/\sqrt{s},\|\cdot\|_{\bX},t\right)dt.
\end{align*}
Lemma 3.9 and Equation (3.10) in \cite{RudelsonVershynin2008} yield that for some universal constant $A$, 
\begin{align*}
&\int_0^{M} \log N\left(\Theta_{s}/\sqrt{s},\|\cdot\|_{\bX},t\right) dt\\ 
&\le \int_{0}^{M/\sqrt{s}} \log\left({p \choose s}(1+2M/t)^s \right)dt+\int_{M/\sqrt{s}}^M \log\left((2p)^{At^{-2}M^2 \log N}\right)dt\\
& \le \frac{M}{\sqrt{s}}\log{p \choose  s} + \sqrt{s}\int_{0}^{M/\sqrt{s}} \log (1+2M/t)dt + AM^2(\log N)(\log (2p)) \int_{M/\sqrt{s}}^M \frac{dt}{t^2}\\
& \lesssim M \sqrt{s}\log p+M  (1+2\sqrt{s})\log\left(1+\frac{1}{2\sqrt{s}}\right)+ A\sqrt{s}M(\log N)(\log (2p)) \\
& \lesssim \sqrt{s}M\left(\log p + (\log\overline N)(\log p)\right),
\end{align*}
where the second term follows from integration by parts
\begin{align*}
\sqrt{s}\int_{0}^{M/\sqrt{s}} \log (1+2M/t)dt
&\le\sqrt{s}t\log\left(1+\frac{2M}{t}\right)\Big|_{0}^{M/\sqrt{s}}+\sqrt{s}2M \log(t+2M)\Big|_{0}^{M/\sqrt{s}}\\
&\lesssim M  (1+2\sqrt{s})\log\left(1+\frac{1}{2\sqrt{s}}\right).
\end{align*}
Hence, we have $II \lesssim sR_{2} M\left\{\log p + (\log\overline N)(\log p)\right\}$.

Setting $\delta_{\bN,3}=sN^{-1/2}B\left(\log p + (\log\overline N)(\log p)\right)$, we have 
\begin{align*}
III &:=\E\left[\sup_{\theta\in \Theta_{s}}\left|\sum_{i,j} \epsilon_i \epsilon'_j  (\theta^T\bX_{i,j})^2\right|\right]
\lesssim
\frac{\delta_{\bN,3} \E[MR_2]}{B\sqrt{N}}
\le \left(\frac{\delta_{\bN,3}}{B}\right) \left(\frac{\E[M^2]\E[R_2^2]}{N}\right)^{1/2} \\
&\le\delta_{\bN,3}  \left(\frac{\E[R_2^2]}{N}\right)^{1/2}  
\lesssim
\delta_{\bN,3}\left(III+ \sup_{\theta\in\Theta_{s}} \E[(\theta^T\bX_{1,1})^2] \right)^{1/2}.
\end{align*}
 Using the same algebraic fact as in Step 1 yields that
 $III\lesssim \delta_{\bN,3}^2 + \delta_{\bN,3}\sqrt{\sup_{\theta\in\Theta_{s}} \E[(\theta^T\bX_{1,1})^2] }$. 

Finally, since $n\le \sqrt{N}$ and $s\le n$, we have 
\[
\frac{sB}{\sqrt{N}}\left(\log p + (\log\overline N)(\log p)\right)
\lesssim
\frac{sB}{n}\left(\log p + (\log\overline N)(\log p)\right)
\lesssim   \frac{\sqrt{s}B}{\sqrt{n}} \log^2(p\overline N).
\]
This completes the proof. 
\end{proof}
%%%%%%%%%%%%%%%%%%%%%%%%%%%%%%%%%%%%%%%%%%%%%
%%%%%%%%%%%%%%%%%%%%%%%%%%%%%%
%%%%%%%%%%%%%%%%%%%%%%%%%%%%%%%%%%%%%%%%%%%%%
%%%%%%%%%%%%%%%%%%%%%%%%%%%%%%
%%%%%%%%%%%%%%%%%%%%%%%%%%%%%%%%%%%%%%%%%%%%%%%%%%%%%%%%%%%%%%%%%%%%%

%%%%%%%%%%%%%%%%%%%%%%%%%%%%%%%%%%%%%%%%%%%%%%%%%%%%%%%%%%%%%%%%%%%%%

\section{Technical Tools}\label{sec:auxiliary_lemmas}
%%%%%%%%%%%%%%%%%%%%%%%%%%%%%%%%%%%%%%%%%%%%%%%%%%%%%%%%%%%%%%%%%%%%%

\begin{lemma}[Nazarov's inequality]
\label{lem: AC}
Let $\bm{Y}=(Y^{1},\dots,Y^{p})^{T}$ be a centered Gaussian random vector in $\R^p$ such that $\E[|Y^{j}|^{2}]\geq \underline{\sigma}^{2}$ for all $1 \le j \le p$ and some constant $\underline{\sigma}>0$. 
Then for every $\bm{y} \in \R^{p}$ and $\delta > 0$,
\[
\Prob(\bm{Y} \le \bm{y}+\delta)-\Prob(\bm{Y} \le \bm{y})\leq \frac{\delta}{\underline{\sigma}} (\sqrt{2\log p}+2).
\]
\end{lemma}

\begin{proof}
This is Lemma A.1 in \cite{CCK2017AoP}; see \cite{CCK2017Note} for its proof.
\end{proof}

\begin{lemma}[Gaussian comparison over rectangles]
\label{lem:Gaussian_comparison}
Let $\bm{Y}$ and $\bW$ be centered Gaussian random vectors in $\R^{d}$ with covariance matrices $\Sigma^{Y} = (\Sigma_{j,k}^{Y})_{1 \le j,k \le d}$ and $\Sigma^{W} = (\Sigma_{j,k}^{W})_{1 \le j,k \le d}$, respectively, and let $\Delta = \| \Sigma^{Y} - \Sigma^{W} \|_{\infty}$. 
Suppose that $\min_{1 \le j \le d} \Sigma_{j,j}^{Y} \bigvee \min_{1 \le j \le d} \Sigma_{j,j}^{W} \ge \underline{\sigma}^{2}$ for some constant $\underline{\sigma} > 0$. 
Then
\[
\sup_{R \in \calR} | \Prob (\bm{Y} \in R) - \Prob (\bW \in R) | \le C (\Delta \log^2 d)^{1/2},
\]
where $C$ is a constant that depends only on $\underline{\sigma}$.
\end{lemma}

\begin{proof}
See Corollary 5.1 in \cite{CCKK2019}. 
\end{proof}

%%%%%%%%%%%%%%%%%%%%%%%%%%%%%%%%%%%%%%%%%%%%%%%%%%%%%%%%%%%%%%%%%%%%%
\section{Additional simulation studies}\label{sec:additional_simulation}
%%%%%%%%%%%%%%%%%%%%%%%%%%%%%%%%%%%%%%%%%%%%%%%%%%%%%%%%%%%%%%%%%%%%%

%%%%%%%%%%%%%%%%%%%%%%%%%%%%%%%%%%%%%%%%%%%%%%%%%%%%%%%%%%%%%%%%%%%%%
\subsection{Gaussian design}\label{sec:simulation_gaussian}
%%%%%%%%%%%%%%%%%%%%%%%%%%%%%%%%%%%%%%%%%%%%%%%%%%%%%%%%%%%%%%%%%%%%%

In Section \ref{sec:simulation} %Sections \ref{sec:simulation_two_way} and \ref{sec:simulation_dyadic} 
in the main text, we experiment with simulation designs based on the mixture distribution:
$
\bm{Z}_{\i \odot \e} \sim B N(\bm{0}, \Sigma_{\bm{Z}}) + (1-B)N(\bm{0}, 2\Sigma_{\bm{Z}})
$
and
$
B \sim \text{Bernoulli}(0.5)
$
independently for $\i \in \{(i_1,i_2) \in \N^2: 1 \le i_1 \le N_1, 1 \le i_2 \le N_2\}$ and $\e \in \{0,1\}^2$ for separately exchangeable arrays%(Section \ref{sec:simulation_two_way})
, and 
$
\bm{Z}_{\i \odot \e} \sim B N(\bm{0}, \Sigma_{\bm{Z}}) + (1-B) N(\bm{0}, 2\Sigma_{\bm{Z}})
$
and
$
B \sim \text{Bernoulli}(0.5)
$
independently for $\i \in \{(i,j) \in \N^2: 1 \le i < j \le n \}$ and $\e \in \{1\} \times \{0,1\}$ symmetrically in $i$ and $j$ for jointly exchangeable arrays.
%independently for $\i \in \{(i,j) \in \N^2: 1 \le i,j \le n, i \ne j \}$ and $\e \in \{1\} \times \{0,1\}$ for jointly exchangeable arrays. %(Section \ref{sec:simulation_dyadic}).
In the current appendix section, we present additional simulation results under the Gaussian design:
$
\bm{Z}_{\i \odot \e} \sim N(\bm{0}, \Sigma_{\bm{Z}})
$
independently for $\i \in \{(i_1,i_2) \in \N^2: 1 \le i_1 \le N_1, 1 \le i_2 \le N_2\}$ and $\e \in \{0,1\}^2$ for separately exchangeable arrays, and
$
\bm{Z}_{\i \odot \e} \sim N(\bm{0}, \Sigma_{\bm{Z}})
$
independently for $\i \in \{(i,j) \in \N^2: 1 \le i < j \le n\}$ and $\e \in \{1\} \times \{0,1\}$ symmetrically in $i$ and $j$ for jointly exchangeable arrays.

Table \ref{tab:two_way_gaussian} summarizes simulation results under the separate exchangeability.
The columns consist of the dimension $p$ of $\bX$, and the two-way sample size $(N_1,N_2)$.
The displayed numbers indicate the simulated uniform coverage frequencies for the nominal probabilities of 90\% and 95\%.
For each dimension $p \in \{25,50,100\}$, sample sizes vary as $(N_1,N_2) \in \{(25,25), (50,50),(100,100)\}$.
Observe that, for each nominal probability, the uniform coverage frequencies approach the nominal probability as the sample size increases.
These results support the theoretical property of our multiplier bootstrap method.

\begin{table}[t]
	\centering
	  \scalebox{0.9}{
		\begin{tabular}{|r|ccc|ccc|ccc|}
		\multicolumn{10}{c}{}\\
		\hline
			Normalization & \multicolumn{9}{|c|}{No} \\
		\hline
			Dimension of $\bX_{\i}$: $p$ & 25  & 25  & 25  & 50  & 50  & 50  & 100 & 100 & 100\\
		\hline
			Sample Sizes: $N_1,N_2$ & 25  & 50  & 100 & 25  & 50  & 100 & 25  & 50  & 100\\
		\hline
			90\% Coverage      &0.928&0.921&0.909&0.935&0.925&0.906&0.943&0.916&0.910\\
			95\% Coverage      &0.973&0.964&0.955&0.973&0.963&0.954&0.976&0.962&0.960\\
		\hline
			Normalization & \multicolumn{9}{|c|}{Yes} \\
		\hline
			Dimension of $\bX_{\i}$: $p$ & 25  & 25  & 25  & 50  & 50  & 50  & 100 & 100 & 100\\
		\hline
			Sample Sizes: $N_1,N_2$ & 25  & 50  & 100 & 25  & 50  & 100 & 25  & 50  & 100\\
		\hline
			90\% Coverage      &0.895&0.889&0.888&0.883&0.889&0.904&0.858&0.896&0.895\\
			95\% Coverage      &0.946&0.944&0.938&0.938&0.946&0.948&0.920&0.944&0.940\\
		\hline
		\multicolumn{10}{c}{}
		\end{tabular}
		}
	\caption{Simulation results for separately exchangeable data with $K=2$ indices. Displayed are the dimension $p$ of $\bX$, the two-way sample size $(N_1,N_2)$ with $N_1=N_2$, and the simulated uniform  coverage frequencies for the nominal probabilities of 90\% and 95\%.}
	\label{tab:two_way_gaussian}
\end{table}

Table \ref{tab:polyadic_gaussian} summarizes simulation results under the joint exchangeability.
The columns consist of the dimension $p$ of $\bX$ and the dyadic sample size $N$.
The displayed numbers indicate the simulated uniform coverage frequencies for the nominal probabilities of 90\% and 95\%.
For each dimension $p \in \{25,50,100\}$, sample sizes vary as $n \in \{50, 100, 200\}$.
Observe that, for each nominal probability, the uniform coverage frequencies approach the nominal probability as the sample size increases.
These results support the theoretical property of our multiplier bootstrap method.

\begin{table}[t]
	\centering
		\scalebox{0.9}{
		\begin{tabular}{|r|ccc|ccc|ccc|}
		\multicolumn{10}{c}{}\\
		\hline
			Normalization & \multicolumn{9}{|c|}{No} \\
		\hline
			Dimension of $\bX_{i,j}$: $p$ & 25  & 25  & 25  & 50  & 50  & 50  & 100 & 100 & 100\\
		\hline
			Sample Size: $n$ & 50  & 100 & 200 & 50  & 100 & 200 & 50  & 100 & 200\\
		\hline
			90\% Coverage      &0.909&0.894&0.898&0.909&0.915&0.903&0.913&0.906&0.901\\
			95\% Coverage      &0.960&0.951&0.948&0.966&0.966&0.954&0.960&0.956&0.954\\
		\hline
			Normalization & \multicolumn{9}{|c|}{Yes} \\
		\hline
			Dimension of $\bX_{i,j}$: $p$ & 25  & 25  & 25  & 50  & 50  & 50  & 100 & 100 & 100\\
		\hline
			Sample Size: $n$ & 50  & 100 & 200 & 50  & 100 & 200 & 50  & 100 & 200\\
		\hline
			90\% Coverage      &0.858&0.874&0.886&0.836&0.889&0.872&0.807&0.857&0.886\\
			95\% Coverage      &0.923&0.934&0.942&0.904&0.948&0.929&0.891&0.916&0.940\\
		\hline
		\end{tabular}
		}
		\medskip
	\caption{Simulation results for dyadic data. Displayed are the dimension $p$ of $\bX$, the dyadic sample size $n$, and the simulated uniform  coverage frequencies for the nominal probabilities of 90\% and 95\%.}
	\label{tab:polyadic_gaussian}
\end{table}

%%%%%%%%%%%%%%%%%%%%%%%%%%%%%%%%%%%%%%%%%%%%%%%%%%%%%%%%%%%%%%%%%%%%%
\subsection{Separate exchangeability with three indices}\label{sec:simulation_three_way}
%%%%%%%%%%%%%%%%%%%%%%%%%%%%%%%%%%%%%%%%%%%%%%%%%%%%%%%%%%%%%%%%%%%%%

In Section \ref{sec:simulation} %\ref{sec:simulation_two_way} 
in the main text, we experimented with separately exchangeable arrays with $K=2$ indices.
In the current appendix section, we present simulation studies based on exchangeability with $K=3$ indices.
Samples are generated according to
\begin{align*}
\bX_{\i}
=
\frac{1}{12} \left( 
\bm{Z}_{(i_1,0,0)} + \bm{Z}_{(0,i_2,0)} + \bm{Z}_{(0,0,i_3)} +
\bm{Z}_{(i_1,i_2,0)} + \bm{Z}_{(i_1,0,i_3)} + \bm{Z}_{(0,i_2,i_3)} 
\right)
+ \frac{1}{2} \bm{Z}_{(i_1,i_2,i_3)},
\end{align*}
where 
(i)
$
\bm{Z}_{\i \odot \e} \sim N(\bm{0}, \Sigma_{\bm{Z}})
$
independently for $\i \in \{(i_1,i_2,i_3) \in \N^3: 1 \le i_1 \le N_1, 1 \le i_2 \le N_2, 1 \le i_3 \le N_3\}$ and $\e \in \{0,1\}^3$ in one design, and
(ii)
$
\bm{Z}_{\i \odot \e} \sim B N(\bm{0}, \Sigma_{\bm{Z}}) + (1-B)N(\bm{0}, 2\Sigma_{\bm{Z}})
$
and
$
B \sim \text{Bernoulli}(0.5)
$
independently for $\i \in \{(i_1,i_2,i_3) \in \N^3: 1 \le i_1 \le N_1, 1 \le i_2 \le N_2, 1 \le i_3 \le N_3\}$ and $\e \in \{0,1\}^3$ in the other design.
For each of these data generating designs, we run 2,500 Monte Carlo iterations to compute the uniform coverage frequencies of $\E[\bX_{\i}]$ for the nominal probabilities of 90\% and 95\% using our proposed multiplier bootstrap for separately exchangeable arrays with 2,500 bootstrap iterations.

Table \ref{tab:three_way} summarizes simulation results.
The columns consist of the dimension $p$ of $\bX$ and the three-way sample size $(N_1,N_2,N_3)$.
The displayed numbers indicate the simulated uniform coverage frequencies for the nominal probabilities of 90\% and 95\%.
For each dimension $p \in \{25,50,100\}$, sample sizes vary as $(N_1,N_2,N_3) \in \{(25,25,25), (50,50,50),(100,100,100)\}$.

\begin{table}
	\centering
		\scalebox{0.9}{
		\begin{tabular}{|r|ccc|ccc|ccc|}
		\multicolumn{10}{c}{}\\
		\hline
		  Distribution of $\bm{Z}_{\i \odot \e}$ & \multicolumn{9}{|c|}{(i) Gaussian} \\
		\hline
			Normalization & \multicolumn{9}{|c|}{No} \\
		\hline
			Dimension of $\bX_{\i}$: $p$ & 25  & 25  & 25  & 50  & 50  & 50  & 100 & 100 & 100\\
		\hline
			Sample Sizes: $N_1,N_2,N_3$ & 25  & 50  & 100 & 25  & 50  & 100 & 25  & 50  & 100\\
		\hline
			90\% Coverage      &0.912&0.912&0.910&0.932&0.914&0.908&0.929&0.918&0.902\\
			95\% Coverage      &0.952&0.958&0.951&0.971&0.958&0.956&0.973&0.962&0.956\\
		\hline
			Normalization & \multicolumn{9}{|c|}{Yes} \\
		\hline
			Dimension of $\bX_{\i}$: $p$ & 25  & 25  & 25  & 50  & 50  & 50  & 100 & 100 & 100\\
		\hline
			Sample Sizes: $N_1,N_2,N_3$ & 25  & 50  & 100 & 25  & 50  & 100 & 25  & 50  & 100\\
		\hline
			90\% Coverage      &0.882&0.892&0.908&0.891&0.897&0.904&0.888&0.889&0.894\\
			95\% Coverage      &0.942&0.944&0.959&0.946&0.949&0.956&0.944&0.939&0.942\\
		\hline
		\multicolumn{10}{c}{}\\
		\hline
		  Distribution of $\bm{Z}_{\i \odot \e}$ & \multicolumn{9}{|c|}{(ii) Mixture} \\
		\hline
			Normalization & \multicolumn{9}{|c|}{No} \\
		\hline
			Dimension of $\bX_{\i}$: $p$ & 25  & 25  & 25  & 50  & 50  & 50  & 100 & 100 & 100\\
		\hline
			Sample Sizes: $N_1,N_2,N_3$ & 25  & 50  & 100 & 25  & 50  & 100 & 25  & 50  & 100\\
		\hline
			90\% Coverage      &0.921&0.916&0.904&0.923&0.915&0.908&0.946&0.913&0.908\\
			95\% Coverage      &0.964&0.958&0.952&0.960&0.958&0.956&0.974&0.959&0.958\\
		\hline
			Normalization & \multicolumn{9}{|c|}{Yes} \\
		\hline
			Dimension of $\bX_{\i}$: $p$ & 25  & 25  & 25  & 50  & 50  & 50  & 100 & 100 & 100\\
		\hline
		\hline
			90\% Coverage      &0.896&0.904&0.892&0.899&0.900&0.906&0.894&0.886&0.886\\
			95\% Coverage      &0.943&0.945&0.945&0.948&0.948&0.948&0.942&0.940&0.936\\
		\hline
		\end{tabular}
		}
		\medskip
	\caption{Simulation results for three-way ($K=3$) cluster sampled data. Displayed are the dimension $p$ of $\bX$, the three-way sample size $(N_1,N_2,N_3)$ with $N_1=N_2=N_3$, and the simulated uniform  coverage frequencies for the nominal probabilities of 90\% and 95\%.}
	\label{tab:three_way}
\end{table}

%%%%%%%%%%%%%%%%%%%%%%%%%%%%%%%%%%%%%%%%%%%%%%%%%%%%%%%%%%%%%%%%%%%%%
\subsection{Uniform confidence band for densities of dyadic data}\label{sec:simulation_density}
%%%%%%%%%%%%%%%%%%%%%%%%%%%%%%%%%%%%%%%%%%%%%%%%%%%%%%%%%%%%%%%%%%%%%

In this section, we present simulation studies to evaluate finite sample performance of the proposed uniform confidence bands for probability density functions of dyadic data that is presented in Section \ref{sec:density}.
Dyadic data are generated 
symmetrically in $i$ and $j$
according to
$$
Y_{i,j} = \frac{1}{4}(U_{i,0}+U_{j,0}) + \frac{1}{2} U_{i,j},
$$ 
where 
(i) $U_{\i \odot \e} \sim N(0,1)$ independently for $\i \in \{(i,j) \in \N^2: 1 \le i < j \le n \}$ and $\e \in \{1\} \times \{0,1\}$ in one design, and
(ii) $U_{\i \odot \e} \sim Logistic(0,1)$ independently for $\i \in \{(i,j) \in \N^2: 1 \le i < j \le n \}$ and $\e \in \{1\} \times \{0,1\}$ in the other design.

We use the Epanechnikov kernel function $K$ for estimation and inference for the probability density functions $f$ of $Y_{i,j}$.
We use the $n^{1/5}$-undersmoothed version of two Silverman's rules of thumb, i.e., (a) $h^1_n = 1.06 \hat\sigma_{Y_{i,j}} n^{-2/5}$ and (b) $h^2_n = 0.9 \min \left\{ \hat\sigma_{Y_{i,j}}, \widehat{IQR}_{Y_{i,j}}/1.34 \right\} n^{-2/5}$ where $\hat\sigma_{Y_{i,j}}$ and $\widehat{IQR}_{Y_{i,j}}$ are the sample standard deviation and the sample interquartile range of $Y_{i,j}$, respectively.
Confidence bands for $f$ are constructed on the interval $[-2,2]$ with the grid size of 201.
We run 2,500 Monte Carlo iterations to compute the uniform coverage frequencies of $f$ on this grid for the nominal probabilities of 90\% and 95\% using our proposed multiplier bootstrap for inference about the probability density functions of dyadic data with 2,500 bootstrap iterations.

Table \ref{tab:density} shows simulation results.
The columns consist of the dyadic sample sizes $n \in \{250,500\}$.
The displayed numbers indicate the simulated uniform coverage frequencies for the nominal probabilities of 95\% and 95\%.
Observe that, for each nominal probability and for each data generating design, the uniform coverage frequencies approach the nominal probability as the sample size increases.
These results support the theoretical property of our multiplier bootstrap method for constructing uniform confidence bands for probability density functions of dyadic data.

\begin{table}[t]
	\centering
	  \scalebox{0.9}{
		\begin{tabular}{|r|ccc|ccc|}
		\multicolumn{7}{c}{}\\
		\hline
		  Distribution of $U_{\i \odot \e}$ & \multicolumn{6}{|c|}{(i) Gaussian}\\
		\hline
		  Bandwidth Rule & \multicolumn{3}{|c|}{(a) $h^1_n$} & \multicolumn{3}{|c|}{(b) $h^2_n$}\\
			\hline
			Sample Sizes: $n$ & 250 & 500 & 1000 & 250 & 500 & 1000\\
		\hline
%			80\% Coverage      & 0.712 & 0.788 & 0.790 & 0.678 & 0.778 & 0.787
%\\
			90\% Coverage      & 0.835 & 0.908 & 0.906 & 0.813 & 0.889 & 0.913
\\
			95\% Coverage      & 0.902 & 0.953 & 0.962 & 0.880 & 0.949 & 0.959\\
		\hline
		\multicolumn{7}{c}{}\\
		\hline
		  Distribution of $U_{\i \odot \e}$ & \multicolumn{6}{|c|}{(ii) Logistic}\\
		\hline
		  Bandwidth Rule & \multicolumn{3}{|c|}{(a) $h^1_n$} & \multicolumn{3}{|c|}{(b) $h^2_n$}\\
			\hline
			Sample Sizes: $n$ & 250 & 500 & 1000 & 250 & 500 & 1000\\
		\hline
	%		80\% Coverage      & 0.792 & 0.817 & 0.799 & 0.781 & 0.809 & 0.794\\
			90\% Coverage      & 0.906 & 0.916 & 0.914 & 0.899 & 0.914 & 0.908\\
			95\% Coverage      & 0.955 & 0.962 & 0.962 & 0.951 & 0.958 & 0.961\\
		\hline
		\end{tabular}
		}
		\medskip
	\caption{Simulation results for uniform confidence bands on $[-2,2]$ of probability density functions of dyadic data. Displayed are the dyadic sample sizes $n$ and the simulated uniform coverage frequencies for the nominal probabilities of 90\% and 95\%.}
	\label{tab:density}
\end{table}

%%%%%%%%%%%%%%%%%%%%%%%%%%%%%%%%%%%%%%%%%%%%%%%%%%%%%%%%%%%%%%%%%%%%%
\section{More on the Exchangeable Arrays}\label{sec:more_exchangeable_arrays}
%%%%%%%%%%%%%%%%%%%%%%%%%%%%%%%%%%%%%%%%%%%%%%%%%%%%%%%%%%%%%%%%%%%%%

In this section, we discuss more details about the separately and jointly exchangeable arrays,  their concrete examples, and their differences.
For simplicity, we focus on the case with $K=2$.
Recall that an array $(\bX_{(i_1,i_2)})_{(i_1,i_2) \in \N^{2}}$ is separately exchangeable if for any two permutations $\pi_1$ and $\pi_2$ of $\N$, the arrays $(\bX_{(i_1,i_2)})_{(i_1,i_2) \in \N^{2}}$ and $(\bX_{(\pi_1(i_1),\pi_2(i_2))})_{(i_1,i_2) \in \N^{2}}$ are identically distributed, cf. Definition 1. 
An array $(\bX_{(i_1,i_2)})_{(i_1,i_2) \in I_{\infty,2}}$ is called jointly exchangeable if for any permutation $\pi$ of $\N$, the arrays $(\bX_{i_1,i_2})_{(i_1,i_2) \in I_{\infty,2}}$ and $(\bX_{(\pi(i_1),\pi(i_2))})_{(i_1,i_2)\in I_{\infty,2}}$ are identically distributed, cf. Definition 2.
We first introduce market data in Example \ref{ex:market_data} as an example of a separately exchangeable array, and network data in Example \ref{ex:network} as an example of a jointly exchangeable array.

\begin{example}[Market Data - Separately Exchangeable Array]\label{ex:market_data}
Market data $(\bX_{i_1,i_2})_{(i_1,i_2) \in \N^{2}}$ commonly in use for marketing analysis are typically indexed by two indices, namely $i_1$ for markets and $i_2$ for products.
An observation $\bX_{(i_1,i_2)}$ may consist of $p$ attributes, such as the market share $X_{(i_1,i_2)}^1$ of product $i_2$ in market $i_1$, the price $X_{(i_1,i_2)}^2$ of product $i_2$ in market $i_2$, and so on.
The array $(\bX_{(i_1,i_2)})_{(i_1,i_2) \in \N^{2}}$ is separately exchangeable if its distribution is identical after a permutation of the market indices $i_1$ and a permutation of the product indices $i_2$.
\end{example}

\begin{example}[Network Data]\label{ex:network}
Network edge data $(\bX_{(i_1,i_2)})_{(i_1,i_2) \in I_{\infty,2}}$, such as those of international trade are typically indexed by two indices, namely $i_1$ for originating country and $i_2$ for destination country.
An observation $\bX_{(i_1,i_2)}$ may consist of $p$ attributes, such as the trade volume $X_{(i_1,i_2)}^1$ of wheat products from country $i_1$ to country $i_2$, trade volume $X_{(i_1,i_2)}^1$ of dairy products from country $i_1$ to country $i_2$, and so on.
The array $(\bX_{(i_1,i_2)})_{(i_1,i_2) \in I_{\infty,2}}$ is jointly exchangeable if its distribution is identical after a permutation of the country indices.
\end{example}

Given these two concrete examples, we can now describe key differences between the two notions of exchangeability.
In Example \ref{ex:market_data}, $i_1$ and $i_2$ index different sets of units, namely markets and products, respectively.
For this data structure, an identical distribution may well hold even after permuting the markets and products separately.
In Example \ref{ex:network}, on the other hand, $i_1$ and $i_2$ index the same set of units, namely countries.
For this data structure, an identical distribution is less plausible after permuting the origins and destinations separately.

To see this point more concretely, consider the sub-array $(\bX_{(1,2)}, \bX_{(2,1)})$ shown in the left matrix below.
\begin{align*}
&&
\text{Separately Exchanged}
\\
\left(\begin{array}{cc}
            & \bX_{(1,2)} \\
\bX_{(2,1)} &
\end{array}\right)
&&
\left(\begin{array}{cc}
            & \bX_{(3,5)} \\
\bX_{(4,6)} &
\end{array}\right)
\end{align*}
Consider separate permutations $\pi_1$ on $i_1$ and $\pi_2$ on $i_2$ such that $\pi_1(1)=3$, $\pi_1(2)=4$, $\pi_2(1)= 6$ and $\pi_2(2)=5$ to obtain the right matrix above which yields $(\bX_{(3,5)}, \bX_{(4,6)})$.
The separate exchangeability requires that $(\bX_{(1,2)}, \bX_{(2,1)}) \stackrel{d}{=} (\bX_{(3,5)}, \bX_{(4,6)})$ in particular.
This is plausible in Example \ref{ex:market_data} if different markets, 1, 2, 3, and 4, are \textit{ex ante} identical and different products, 1, 2, 5, and 6, are \textit{ex ante} identical in terms of the distribution.

On the other hand, $(\bX_{(1,2)}, \bX_{(2,1)}) \stackrel{d}{=} (\bX_{(3,5)}, \bX_{(4,6)})$ is not plausible in Example \ref{ex:network}.
The two observations $(\bX_{(1,2)}, \bX_{(2,1)})$ are likely to be highly correlated because they measure exports and imports among the identical pair $(1,2)$ of countries.
In contrast, the two observations $(\bX_{(3,5)}, \bX_{(4,6)})$ are likely to be less correlated because they measure exports and imports among two distinct pairs, $(3,5)$ and $(4,6)$, of countries.
Thus, the joint distributions of $(\bX_{(1,2)}, \bX_{(2,1)})$ and $(\bX_{(3,5)}, \bX_{(4,6)})$ are not plausibly assumed to be identical in this example.

Next, consider a joint permutation $\pi$ on $i_1$ and $i_2$ such that $\pi(1)=3$ and $\pi(2)=4$ to obtain the right matrix below which yields $(\bX_{(3,4)}, \bX_{(4,3)})$.
\begin{align*}
&&
\text{Jointly Exchanged}
\\
\left(\begin{array}{cc}
            & \bX_{(1,2)} \\
\bX_{(2,1)} &
\end{array}\right)
&&
\left(\begin{array}{cc}
            & \bX_{(3,4)} \\
\bX_{(4,3)} &
\end{array}\right)
\end{align*}
Observe that $(\bX_{(1,2)}, \bX_{(2,1)}) \stackrel{d}{=} (\bX_{(3,5)}, \bX_{(4,6)})$ is plausible even in Example \ref{ex:network}.
The two observations $(\bX_{(1,2)}, \bX_{(2,1)})$ are likely to be highly correlated because they measure exports and imports among the identical pair $(1,2)$ of countries, as well as $(\bX_{(3,4)}, \bX_{(4,3)})$ are similarly likely to be highly correlated because they also measure exports and imports among the identical pair $(3,4)$ of countries.

The above illustrations show that the separate exchangeability is legitimate for Example \ref{ex:market_data} but not for Example \ref{ex:network}.
On the other hand, the joint exchangeability is more relevant to Example \ref{ex:network}.
As can be seen from the above illustrations, the separately exchangeability implies the joint exchangeability.
That being said, we want to emphasize that we use different array structures between the definition of separately exchangeable arrays and that of jointly exchangeable arrays in that the index set is $\mathbb{N}^2$ for the former and the index set is $I_{\infty,2}$ for the latter.
In this sense, the two definitions of exchangeable arrays do not nest each other.
The main purpose of using $I_{\infty,2}$ for the latter is to preclude a self link, as there is no trade flow from a country to itself.

We conclude by commenting on a comparison of the high-dimensional CLTs and bootstrap validity results under the two notions of exchangeability. Although at the first glimpse Theorems \ref{thm: high-d CLT polyadic} and \ref{thm:bootstrap validity polyadic} under joint exchangeability look similar to Theorems \ref{thm: high-d CLT} and \ref{thm:bootsrap_validity}, respectively, under separately exchangeable arrays, they are in fact fundamentally different and neither is nested or implied by the other. The fundamental differences are rooted in the different dependence structures that result in different Hoeffding decompositions, H\'ajek projections, as well as distinctive techniques for handling higher order terms. These in turn lead to two dedicated bootstrap procedures. Therefore, both sets of results are irreplaceable by the other.

\setlength{\baselineskip}{7.1mm}
\bibliographystyle{Chicago}
\bibliography{biblio}

\end{document}